\documentclass{article}
\usepackage{amsmath,amssymb,amsfonts,amsthm}
\usepackage[margin=1in]{geometry}
\usepackage{aliascnt}
\usepackage{hyperref}
\usepackage{enumitem}
\usepackage{float} 
\usepackage{placeins}
\usepackage{mathtools} 
\usepackage{todonotes} 
\usepackage{tikz-cd}

\usepackage{graphicx} 
\theoremstyle{definition}
\newtheorem{theorem}{Theorem}[section] 
\newtheorem{lemma}[theorem]{Lemma} 
\newtheorem{definition}[theorem]{Definition}
\newtheorem{proposition}[theorem]{Proposition}
\newtheorem{corollary}[theorem]{Corollary} 
\newaliascnt{remark}{theorem}
\newtheorem*{theorem*}{Theorem}
\newtheorem{remark}[remark]{Remark}
\aliascntresetthe{remark}
\newcommand{\APSreg}{\text{APS,reg}}

\DeclareMathOperator{\APS}{APS}

\DeclareMathOperator{\Dom}{Dom}   
\hypersetup{colorlinks,breaklinks, linkcolor = purple, citecolor = teal, urlcolor = purple, }
\numberwithin{equation}{section}

\title{Higher-Rank Orthogonal Twists, APS Boundary Conditions, and O(2)-Equivariant Spectral Flow on a Warped Cylinder}
\author{Taro Kimura$^*$, Sanchita Sharma%
\thanks{Université Bourgogne Europe, CNRS, IMB UMR 5584, 21000 Dijon, France}}
\date{}

\begin{document}

\maketitle

\begin{abstract}
We study $O(2)$-equivariant spectral flow for Dirac operators on a finite warped cylinder equipped with fixed admissible regularized APS boundary conditions.  The twisting bundle is a real higher-rank orthogonal bundle, and reflection symmetry is implemented by a fiber involution.  After complexifying the twisting bundle and diagonalizing the orthogonal twist, the Dirac equation decomposes into a scalar Fourier-mode radial equation, with moving rotating blocks and stationary neutral blocks. After regrouping conjugate and reflection-paired blocks, the crossing
contributions define real $RO(O(2))$-classes. Consequently, we obtain an explicit blockwise formula for the $RO(O(2))$-valued spectral flow of the resulting regularized APS family.  Under the standard self-adjoint Fredholm, endpoint-invertibility, and regular-crossing hypotheses, together with a fixed neutral-sector convention, this formula is obtained by assembling the local crossing contributions of the separated blocks.  It refines ordinary integer-valued
spectral flow and shows explicitly how the dimension map $RO(O(2))\to\mathbb Z$ loses representation-theoretic information.  We also discuss the rank-three case, including the role of the fixed neutral sector, and the corresponding endpoint $\eta$/APS index interpretation.
\end{abstract}

\setcounter{tocdepth}{2}
\tableofcontents

\section{Introduction}
\label{sec:introduction}

Atiyah-Patodi-Singer boundary conditions are among the basic nonlocal elliptic boundary conditions for Dirac-type operators on manifolds with boundary \cite{APS1,APS2,APS3}. They are defined from the spectral decomposition of the induced self-adjoint boundary Dirac operator and play a central role in index theory, spectral asymmetry, and spectral flow. Analytic treatments of spectral flow for self-adjoint Fredholm operators appear, for example, in \cite{Phillips1996,Waterstraat2017FredholmSpectralFlow}, while elliptic boundary-value problems for Dirac operators and APS-type realizations are treated in \cite{BoossWojciechowski1993,BaerBallmann2012}. Spectral flow for Dirac operators with boundary conditions has also been studied in several directions, including \cite{GL14,van,BaerZiemke2025SpectralFlowAPSIndex}.

When symmetries are present, ordinary integer-valued spectral flow is often too coarse. A crossing kernel may carry a nontrivial group representation, and one should remember this representation rather than only its dimension. This leads to equivariant spectral flow, valued in a representation ring. Equivariant spectral flow has appeared in several
forms, including the Lefschetz-type formula of Fang \cite{Fang}, the equivariant spectral-flow and $\eta$-invariant framework of Lim-Wang
\cite{LimWang2021EquivariantSpectralFlowEtaBoundary}, the equivariant APS index framework of Hochs-Wang-Wang
\cite{HochsWangWang2019EquivariantAPSI,HochsWangWang2020EquivariantAPSII}, the equivariant Toeplitz index theory of Lim-Wang
\cite{LimWang2021EquivariantToeplitzBoundary}, and recent uniqueness and family-level treatments \cite{HochsYanes2024EquivariantSpectralFlowDiracType,
IzydorekJanczewskaStarostkaWaterstraat2026}. Related symmetry-protected refinements, including $\mathbb Z_2$-valued spectral flow, parity, and mod-two APS-type phenomena, occur in real and time-reversal symmetric settings \cite{DollSchulzBaldesWaterstraat2019,
Doll_SchulzBaldes_arXiv_2021,
BravermanHajSaeediSadegh2024Z2ToeplitzSpectralFlow,
FukayaFurutaMatsukiMatsuoOnogiYamaguchiYamashita2022}.
From the mathematical-physics point of view, reflection and more general crystalline symmetries are also central in the study of topological crystalline phases \cite{Ando_Fu_2015}.

This paper continues the finite warped-cylinder program developed in the first paper \cite{KS}. There, the basic model was the Dirac operator on the curved cylinder $M=[0,T]\times S^1$ with metric $g=dt^2+f(t)^2\,d\theta^2$, coupled to a background $U(1)$ gauge field. That work identified the endpoint boundary operators defining the
APS condition, described the modewise APS spectral problem, proved the cancellation of endpoint reduced $\eta$ contributions in the constant-gauge invertible case, and introduced a regularized APS-type boundary family for studying zero-mode crossings along gauge paths.

The second paper \cite{KSReflectionPartI} added reflection symmetry to this scalar line-twist model. The geometric symmetry group is $O(2)$, generated by rotations of the circle and the reflection $\mathbf{r}(t,\theta)=(t,-\theta)$. For a complex line twist with holonomy parameter $A$, that paper showed that the reflection lifts to the twisted
Dirac setting exactly in the half-integral holonomy case, equivalently $2A\in\mathbb Z$. In that reflection-compatible case, reflection pairs the shifted Fourier parameters $m=k+A$ and $-m$, and this pairing governs the APS boundary blocks, the reflection trace, and the residual parity information for holonomy deformations.

The current paper is the higher-rank orthogonal analogue of the scalar reflection model. Instead of a complex line twist, we consider the trivial real Euclidean twisting bundle $E_{\mathbb R}^{(n)}=M\times\mathbb R^n$ with orthogonal connection $\nabla^{E_{\mathbb R}^{(n)}}_{A_p}=d+A_pJ_n\,d\theta$, where
$J_n\in\mathfrak{so}(n)$ and $p\mapsto A_p$ is a path of real parameters. Reflection is implemented on the twisting factor by an orthogonal involution $C_n$ satisfying $C_nJ_nC_n^{-1}=-J_n$, which compensates for the sign change $d\theta\mapsto-d\theta$ under reflection. After complexification, $J_n$ decomposes into rotating two-plane blocks and neutral blocks. If the nonzero rotation weights are
$\mu_1,\ldots,\mu_r>0$, then the rotating scalar blocks are governed by the effective parameters $m_{k,j,+}(p)=k+\mu_jA_p$ and $m_{k,j,-}(p)=k-\mu_jA_p$. Thus the higher-rank problem reduces modewise to the same scalar warped-cylinder APS block, but with a weighted list of crossing equations, while the neutral blocks corresponding to $\ker J_n$ are stationary and require a separate convention.

This should be distinguished from the scalar complex line-twist obstruction in \cite{KSReflectionPartI}.  There, reflection compatibility was an arithmetic condition on the holonomy class, namely $2A\in\mathbb Z$.  Here the twisting data are different: the fiber involution $C_n$ is part of the real orthogonal model and is chosen to satisfy $C_nJ_nC_n^{-1}=-J_n$.  Thus, reflection compatibility is built into the higher-rank orthogonal twisting data, rather than imposed as a half-integrality condition on a complex line holonomy.

There is one analytic issue that is essential for moving families. The ordinary APS projector jumps when the boundary operator has kernel. Hence, along a path crossing a boundary-zero value, the ordinary APS domains do not form a continuous self-adjoint Fredholm family. We therefore fix an admissible regularized APS convention. This regularized boundary family agrees with the ordinary APS condition away from small regularization intervals, is Lagrangian for the boundary Green form, and is fixed so that the scalar transfer determinant has a simple zero at each regular crossing. All spectral-flow statements below refer to this fixed regularized APS family.

In the local crossing analysis, we use the standard crossing-form viewpoint on
spectral flow. For ordinary self-adjoint Fredholm paths, regular crossings are measured by crossing forms, whose signatures give the local contributions to spectral flow \cite{Waterstraat2017FredholmSpectralFlow, RobbinSalamon1995SpectralFlowMaslov}.
For Dirac operators with boundary, this viewpoint is closely related to the Maslov-index formulation of spectral flow for Lagrangian boundary data
\cite{BoossBavnbekFurutani1998Maslov,KirkLesch2004EtaMaslovSpectralFlow}.
In the equivariant setting considered here, the relevant crossing spaces carry finite-dimensional $O(2)$-representations, and the crossing forms are $O(2)$-invariant. Thus the local signs may be read representation sector by representation sector, giving the blockwise $RO(O(2))$-valued contributions used below. This is the local crossing-form perspective underlying the
equivariant spectral-flow frameworks of
\cite{LimWang2021EquivariantSpectralFlowEtaBoundary, IzydorekJanczewskaStarostkaWaterstraat2026, 
IzydorekJanczewskaStarostkaWaterstraat2025PartI}.

The main theorem should, therefore, be read as an explicit blockwise calculation
in this warped-cylinder regularized APS model. It is not intended as a general
compact-group or arbitrary-boundary-condition equivariant APS spectral-flow
theorem.

\medskip

\noindent
\textbf{Main results.}
The contributions of the paper are as follows.

\begin{enumerate}[label=\textup{(\roman*)}]
\item We construct reflection-compatible higher-rank orthogonal twisting data
$(E_{\mathbb R}^{(n)},J_n,C_n)$ and show that the lifted reflection gives an $O(2)$-equivariant higher-rank Dirac family. The rank of the twisting bundle changes the crossing multiplicities and weights, but the representation ring remains $RO(O(2))$, because $O(2)$ is the geometric symmetry group of the base circle.

\item We diagonalize the complexified higher-rank twist into scalar separated blocks. Each rotating two-plane contributes two scalar parameters $k+\mu_jA_p$ and $k-\mu_jA_p$, while each neutral direction contributes the stationary parameter $k$. For $\mu_j > 0$, the moving crossings occur at
$A_{p_*}=\pm\frac{k}{\mu_j}$.

\item We introduce a scalar regularized APS boundary family and impose it blockwise on
the rotating separated blocks.  The admissibility condition is expressed by the
simple-zero condition for the scalar transfer determinant.

\item Under endpoint invertibility, regular-crossing, admissibility, and
neutral-sector hypotheses, we assemble the local crossings into an
$RO(O(2))$-valued spectral-flow formula. Nonzero reflected Fourier pairs
contribute signed copies of the real two-dimensional representation $\rho_k$, while the self-paired zero Fourier block can split into the two one-dimensional sectors $\mathbf 1$ and $\det$. The detailed local Maslov-frame sign check is given in \autoref{app:maslov_reflected_frame_check}.

\item We give explicit examples showing that the ordinary dimension-valued spectral
flow forgets representation-theoretic information in two distinct ways. Under the
dimension map
$\dim:RO(O(2))\longrightarrow \mathbb Z$,
one has
\begin{equation}
[\rho_k]\longmapsto 2,
\qquad
[\mathbf 1]\longmapsto 1,
\qquad
[\det]\longmapsto 1.
\end{equation}
Thus ordinary spectral flow cannot distinguish different nonzero Fourier
representations of the same dimension; for example,
$\dim[\rho_k]=\dim[\rho_\ell]=2$, for $k,\ell\geq1$,
even though $[\rho_k]\neq[\rho_\ell]$ in $RO(O(2))$ when $k\neq\ell$.

There is also a second, stronger loss of information. With the scalar
$m$-slope sign convention used in the affine examples, a rotating zero-mode
crossing can contribute the signed class
$[\mathbf 1]-[\det]$.
This class is generally nonzero in $RO(O(2))$, but it is killed by the ordinary
dimension map:
$\dim\left([\mathbf 1]-[\det]\right)=1-1=0$.
Thus ordinary spectral flow can miss a nontrivial signed equivariant zero-mode
crossing entirely.

\item Finally, we describe the corresponding fixed-endpoint APS index and $\eta$-term interpretation. In the two-dimensional flat-twist model, the local index density has no degree-two contribution, so the fixed-parameter chiral APS index is controlled by endpoint $\eta$ terms and the finite kernel correction determined by the chosen APS
convention.
\end{enumerate}

\medskip

The paper is organized as follows. 
\autoref{sec:setup} fixes the warped-cylinder Dirac operator, the Fourier-mode
reduction, and the scalar APS boundary conventions.
\autoref{sec:higher_rank_orthogonal_twists} introduces reflection-compatible
higher-rank orthogonal twists and diagonalizes the complexified twisting data.
\autoref{sec:hrank_separated_blocks_regularized_APS} constructs the separated
block picture, the regularized APS boundary condition, and the scalar
admissibility criterion.
\autoref{sec:hrank_ROO2_spectral_flow_formula} proves the higher-rank
$RO(O(2))$-valued spectral-flow formula.
\autoref{sec:rank_three_neutral_block} works out the rank-three case explicitly,
emphasizing how the moving rank-two block is supplemented by a fixed neutral
sector.
\autoref{sec:examples_dimension_map_loss} gives examples showing the loss of information under the ordinary dimension map.
Finally, \autoref{sec:hrank_APS_index_endpoint_eta} describes the APS index and endpoint $\eta$ interpretation for the fixed-parameter higher-rank model.~\autoref{app:maslov_reflected_frame_check} gives the explicit Maslov-frame calculation explaining why nonzero reflected pairs give one signed $\rho_k$-block, while the self-paired zero mode may produce separate $\mathbf 1$ and $\det$ contributions.

\paragraph{Acknowledgments.}
We would like to thank Prof. Waterstraat for useful communications.
This work was supported by EIPHI Graduate School (No.~ANR-17-EURE-0002) and the Bourgogne-Franche-Comté region.

\section{Geometric setup and scalar block conventions}
\label{sec:setup}

We fix the warped-cylinder Dirac model and the scalar block notation used throughout the paper. The higher-rank orthogonal twist will be introduced in~\autoref{sec:higher_rank_orthogonal_twists}; the scalar line-twist notation below is used as the model calculation obtained after diagonalizing the higher-rank twist into separated blocks. We adapt the geometric and analytic setup from \cite{KS}.

\subsection{Geometry, bundles, and function spaces}
Let $M$ be the finite warped cylinder
\begin{equation}\label{eq:MT_metric} M=[0,T]\times S^1, \qquad g=dt^2+f(t)^2\,d\theta^2,
\end{equation}
where $f\in C^\infty([0,T])$ is strictly positive. Thus, $t\in[0,T]$ is the longitudinal coordinate, $\theta$ is the angular coordinate on $S^1$, and the metric $g$ is the warped product line element $ds^2=dt^2+f(t)^2\,d\theta^2$, so that the circle fiber at $t$ has radius $f(t)$.
Its boundary is
\begin{equation}\label{eq:boundary_decomp} \partial M=Y_0\sqcup Y_T, \qquad Y_0=\{0\}\times S^1, \qquad Y_T=\{T\}\times S^1.
\end{equation}

We use the orthonormal coframe
\begin{equation}\label{eq:coframe} e^1=dt, \qquad e^2=f(t)\,d\theta,
\end{equation}
with dual frame
\begin{equation}\label{eq:dual_frame} e_1=\partial_t, \qquad e_2=\frac{1}{f(t)}\partial_\theta.
\end{equation}
We orient $M$ by $dt\wedge d\theta$. With this convention, the inward unit normal is
\begin{equation}\label{eq:inward_normal}
N=
\begin{cases}
+\partial_t,& \text{on }Y_0,\\
-\partial_t,& \text{on }Y_T.
\end{cases}
\end{equation}

Fix a spin structure on $M$, and let $S\to M$ denote the corresponding complex spinor bundle. In dimension two, $S$ has rank $2$ and splits chirally as
\begin{equation}\label{eq:chirality_split} S=S^+\oplus S^-.
\end{equation}
Let $E\to M$ be a Hermitian complex line bundle with a unitary connection
\begin{equation}\label{eq:connectionE} \nabla^E=d+iA\,d\theta, \qquad A\in\mathbb R.
\end{equation}
Only the class
\begin{equation}\label{eq:A_class} [A]\in \mathbb R/\mathbb Z
\end{equation}
is gauge invariant: if one changes the local trivialization of the line bundle by a unitary gauge transformation, then the parameter $A$ may change by an integer. However, its class modulo $\mathbb Z$ remains unchanged. Equivalently, different real numbers $A$ and $A+\ell$, with $\ell\in\mathbb Z$, determine gauge-equivalent connections and hence the same holonomy around the circle. Thus, the twisted Dirac operator acts on sections of the twisted bundle,

\begin{equation}\label{eq:twisted_bundle} S\otimes E\to M.
\end{equation}
From this point onward, we suppress the twist $E$ from the notation when no confusion can arise, and write $S$, $S^\pm$, and $S_{Y_{t_0}}$ for the corresponding twisted bundles.

For any vector bundle $V\to M$, we write $\Gamma(V)$ for the smooth sections of $V$. We use the standard Hilbert and Sobolev spaces
\begin{equation}\label{eq:bulk_spaces} L^2(M;V), \qquad H^1(M;V),
\end{equation}
defined using the Riemannian volume form of $g$ and the Hermitian bundle metric. On each boundary component $Y_{t_0}$, where $t_0\in\{0,T\}$, we similarly write
\begin{equation}\label{eq:boundary_spaces} L^2(Y_{t_0};V|_{Y_{t_0}}), \qquad H^1(Y_{t_0};V|_{Y_{t_0}}), \qquad H^{1/2}(Y_{t_0};V|_{Y_{t_0}})
\end{equation}
for the corresponding boundary spaces, where $H^{1/2}$ is the standard trace space. In particular, for the twisted spinor bundle $S$, the trace map takes the form
\begin{equation}\label{eq:trace_map} H^1(M;S)\longrightarrow H^{1/2}(Y_{t_0};S|_{Y_{t_0}}).
\end{equation}

\subsection{Clifford conventions and the twisted Dirac operator}

Let $\sigma_1,\sigma_2,\sigma_3$ represent the Pauli matrices, and set
\begin{equation}\label{eq:gamma_def} \gamma_1=\sigma_1, \qquad \gamma_2=\sigma_2, \qquad \gamma_3=i\gamma_1\gamma_2=-\sigma_3.
\end{equation}
Thus the chirality decomposition \eqref{eq:chirality_split} is the $\pm1$-eigenspace decomposition of $\gamma_3$.

Using the tensor product connection $
\nabla^{S\otimes E}=\nabla^S\otimes 1 + 1\otimes \nabla^E$, we define the twisted Dirac operator
\begin{equation}\label{eq:Dirac_def} D:\Gamma(S)\longrightarrow \Gamma(S).
\end{equation}
In the orthonormal frame \eqref{eq:dual_frame}, one obtains
\begin{equation}\label{eq:Dexplicit}
D
=
i\sigma_1\Bigl(\partial_t+\frac{f'(t)}{2f(t)}\Bigr)
+
i\sigma_2\,\frac{1}{f(t)}\bigl(\partial_\theta+iA\bigr).
\end{equation}
Equivalently,
\begin{equation}\label{eq:Dmatrix}
D
=
i
\begin{pmatrix}
0 &
\partial_t+\dfrac{f'(t)}{2f(t)}-\dfrac{i}{f(t)}(\partial_\theta+iA)
\\
\partial_t+\dfrac{f'(t)}{2f(t)}+\dfrac{i}{f(t)}(\partial_\theta+iA)
&
0
\end{pmatrix}.
\end{equation}

\subsection{Fourier reduction and mode spaces}

Since the metric and connection are $\theta$-independent, the operator $D$ commutes with rotations in the $\theta$-variable and therefore preserves the Fourier-mode decomposition in the angular direction. In other words, $D$ acts diagonally with respect to the Fourier basis, so the full operator splits into independent mode operators, one for each allowed angular mode. Depending on the chosen spin structure on $S^1$, the allowed Fourier modes are
\begin{equation}\label{eq:Kdef}
\mathcal K=\mathbb Z
\quad\text{(periodic case)},
\qquad
\mathcal K=\mathbb Z+\frac12
\quad\text{(anti-periodic case)}.
\end{equation}
For $k\in\mathcal K$, we write
\begin{equation}\label{eq:fourier_ansatz} \psi(t,\theta)=e^{ik\theta}\binom{u(t)}{v(t)}.
\end{equation}
Since $(\partial_\theta+iA)e^{ik\theta} = i(k+A)e^{ik\theta}$, it is convenient to introduce the shifted mode parameter
\begin{equation}\label{eq:mdef} m=k+A.
\end{equation}

For each mode $k$, the natural radial Hilbert and Sobolev spaces are
\begin{equation}\label{eq:mode_spaces} H_k=L^2([0,T],f(t)\,dt;\mathbb C^2), \qquad W_k=H^1([0,T],f(t)\,dt;\mathbb C^2).
\end{equation}
In the $k$-th mode, $D$ reduces to 
\begin{equation}\label{eq:Dk}
D^{(k)}
=
i
\begin{pmatrix}
0 & \mathcal A^+\\
\mathcal A^- & 0
\end{pmatrix},\qquad
\mathcal A^\pm = \partial_t+\frac{f'(t)}{2f(t)}\pm \frac{m}{f(t)}.
\end{equation}
Hence, the eigenvalue equation
\begin{equation}\label{eq:eigen_eq} D\psi=\lambda\psi, \qquad \lambda\in\mathbb R,
\end{equation}
becomes the modewise first-order system, $\mathcal A^+v=-i\lambda u$, $\mathcal A^-u=-i\lambda v$. 
Equivalently,
\begin{equation}\label{eq:mode_system_components}
v'(t)+\frac{f'(t)}{2f(t)}v(t)+\frac{m}{f(t)}v(t)=-i\lambda u(t),
\qquad
u'(t)+\frac{f'(t)}{2f(t)}u(t)-\frac{m}{f(t)}u(t)=-i\lambda v(t).
\end{equation}
Eliminating one component yields a second-order scalar equation; in this work, the boundary conditions and spectral data are organized entirely in terms of $m$.

\subsection{The boundary Green form}
\label{subsec:boundary_green_form}

For later use, we state the boundary form associated with the Dirac operator. If $\psi,\phi\in H^1(M;S\otimes E)$, Green's formula gives
\begin{equation}\label{eq:green_formula}
\langle D\psi,\phi\rangle_{L^2(M)}
-
\langle \psi,D\phi\rangle_{L^2(M)}
=
\mathcal G(\psi|_{\partial M},\phi|_{\partial M}),
\end{equation}
where the boundary Green form is
\begin{equation}\label{eq:boundary_green_form}
\mathcal G(\psi_{\partial},\phi_{\partial})
=
\int_{\partial M}
\left\langle i\gamma(N)\psi_{\partial},\phi_{\partial}\right\rangle
\,d\operatorname{vol}_{\partial M}.
\end{equation}
Here $N$ is the inward unit normal fixed in \eqref{eq:inward_normal}.  Equivalently, using the decomposition $\partial M=Y_0\sqcup Y_T$, this form is the sum of the two endpoint boundary contributions determined by $N=+\partial_t$ on $Y_0$ and $N=-\partial_t$ on $Y_T$.

A boundary subspace $L\subset H^{1/2}(\partial M;S\otimes E)$ is called isotropic for the Green form if
\begin{equation}
\mathcal G(\psi_{\partial},\phi_{\partial})=0
\qquad
\text{for all }\psi_{\partial},\phi_{\partial}\in L.
\end{equation}
It is called Lagrangian if it is maximal among boundary subspaces with this property. Imposing a Lagrangian boundary condition makes the boundary term in Green's formula vanish on the domain, which is the boundary mechanism behind self-adjointness.

\subsection{Boundary operators and APS boundary conditions}

Let $Y_{t_0}$ be one of the boundary circles, and let
$S_{Y_{t_0}}=S|_{Y_{t_0}}$. For the unit tangent vector
$U=e_2|_{t=t_0}$, define the induced boundary Clifford action by
\begin{equation}\label{eq:boundary_clifford} c(U)=-\,i\,\gamma(N)\gamma(U).
\end{equation}
With the inward normal convention \eqref{eq:inward_normal}, this gives
\begin{equation}\label{eq:cU_formula}
c(U)=
\begin{cases}
\ \sigma_3,& \text{on }Y_0,\\
-\sigma_3,& \text{on }Y_T.
\end{cases}
\end{equation}

The intrinsic tangential operator is
\begin{equation}\label{eq:Dt0} D_{t_0} = \frac{1}{f(t_0)}(-i\partial_\theta+A),
\end{equation}
and hence the self-adjoint boundary Dirac operators entering the APS projector are
\begin{equation}\label{eq:B0BT} B_0=\frac{1}{f(0)}\,\sigma_3\,(-i\partial_\theta+A), \qquad B_T=-\frac{1}{f(T)}\,\sigma_3\,(-i\partial_\theta+A).
\end{equation}
Equivalently, on the $k$-th Fourier mode, with $m=k+A$,
\begin{equation}\label{eq:Bk_mode} B_0^{(k)}=\frac{m}{f(0)}\sigma_3, \qquad B_T^{(k)}=-\frac{m}{f(T)}\sigma_3.
\end{equation}

Let $P_{>0}(B_{t_0})$ denote the orthogonal projection onto the strictly
positive spectral subspace of $B_{t_0}$. For $m\neq 0$, the APS boundary
condition is
\begin{equation}\label{eq:APSabstract} P_{>0}(B_0)\bigl(\psi|_{Y_0}\bigr)=0, \qquad P_{>0}(B_T)\bigl(\psi|_{Y_T}\bigr)=0.
\end{equation}
Since \eqref{eq:Bk_mode} is diagonal, the corresponding APS domain of the
mode operator $D^{(k)}$ is
\begin{equation}\label{eq:DomDkAPS}
\Dom(D^{(k)}_{\APS})
=
\begin{cases}
\displaystyle
\left\{\binom{u}{v}\in W_k:\ u(0)=0,\ v(T)=0\right\},
& m>0,\\
\displaystyle
\left\{\binom{u}{v}\in W_k:\ v(0)=0,\ u(T)=0\right\},
& m<0.
\end{cases}
\end{equation}

\begin{remark}[kernel APS convention]
\label{rem:global_APS_convention}
On every nonzero boundary mode, we impose the standard strict APS projection
$P_{>0}(B_t)=1_{(0,\infty)}(B_t)$.
Thus no auxiliary choice is involved when the boundary operator is invertible.

If $B_t$ has nontrivial kernel, however, the strict projection
$P_{>0}(B_t)$ does not prescribe a boundary condition on $\ker B_t$.
In this case the APS condition must be completed by a finite-dimensional
self-adjoint kernel condition. Equivalently, on the total boundary-kernel
space one chooses a Lagrangian subspace for the boundary Green form and adds
this choice to the strict APS condition. This is the convention under which
we use the notation $P_t^{\APS}$ for the completed APS projector and
$D^{\APS}$ for the corresponding completed APS realization.

In the reflection-symmetric parts of the paper, this kernel completion is
taken to be invariant under the boundary reflection, denoted later by
$\mathcal R_t$. Only the self-paired sector $m=0$ depends on this
auxiliary finite-dimensional completion; for $m\neq0$, the boundary
operators are invertible and the ordinary strict APS condition determines the
domain uniquely.
\end{remark}

\subsection{Relation with the scalar line-twist model}
\label{subsec:relation_scalar_line_twist}

The scalar line-twist model recalled in this section is the local block from which the higher-rank construction below is assembled. It is also the point at which the current paper connects with the preceding reflection-compatible scalar analysis of \cite{KSReflectionPartI}.

In the scalar model, the twisting parameter is a holonomy class
$[A]\in \mathbb R/\mathbb Z$,
and the $k$-th Fourier mode is controlled by the shifted parameter
$m=k+A$.
The ordinary scalar APS boundary condition is therefore governed, mode by
mode, by the sign of $m$. The boundary-zero case $m=0$ is the only case in
which a kernel convention or regularized APS convention is needed.

The reflection of the cylinder is
$\mathbf r(t,\theta)=(t,-\theta)$.
In the scalar line-twist model, this reflection sends the holonomy class
$[A]$ to $[-A]$. Hence a scalar reflection lift exists exactly in the
fixed-holonomy reflection-compatible case
\begin{equation}
[A]=[-A]\quad\text{in }\mathbb R/\mathbb Z,
\qquad\text{equivalently}\qquad
2A\in\mathbb Z.
\end{equation}
When this condition holds, reflection sends the shifted mode parameter $m$
to $-m$, and hence exchanges the two nonzero APS boundary regimes.

We use the scalar radial APS block as the local model for the real higher-rank orthogonal twist introduced in~\autoref{sec:higher_rank_orthogonal_twists}. After the, complexification of the higher-rank twist, each rotating real two-plane
gives two scalar-type blocks, and the scalar APS calculation is applied
separately in each block. What changes is not the scalar boundary calculation itself, but the list of effective shifted parameters and the way the resulting crossing spaces are regrouped into real $O(2)$-representation classes.

In the scalar line-twist model, a nonconstant path of holonomy parameters generally cannot remain pointwise reflection-compatible because the condition $2A\in\mathbb Z$ forces a continuous compatible path to be locally constant. In the higher-rank orthogonal model, reflection compatibility is instead built into the real
twisting data by a fiber involution introduced below. This allows the
higher-rank family to vary with the external parameter while remaining within
the $O(2)$-equivariant framework used later.

Thus, the scalar line-twist setup is the common local APS block from which the higher-rank $RO(O(2))$-valued spectral-flow formula is assembled.

\section{Higher-rank orthogonal twists and reflection symmetry}
\label{sec:higher_rank_orthogonal_twists}

In this section, we replace the scalar block model of
\autoref{sec:setup} by a real orthogonal twisting bundle of arbitrary finite
rank.  The base geometry and symmetry group remain the warped cylinder with
its $O(2)$-action; the higher rank refers only to the twisting bundle.

\subsection{Reflection-compatible higher-rank twisting}
\label{subsec:reflection_compatible_higher_rank_data}

Let
\begin{equation}\label{eq:hrank_real_bundle} E_{\mathbb R}^{(n)}=M\times\mathbb R^n
\end{equation}
be the trivial real Euclidean bundle of rank $n$. A higher-rank orthogonal twist is determined by a fixed skew-symmetric matrix $J_n\in\mathfrak{so}(n)$, along with a real parameter $A$. For a continuously differentiable real-valued path $p\mapsto A_p$, we define the corresponding orthogonal connection by
\begin{equation}\label{eq:hrank_orthogonal_connection} \nabla^{E_{\mathbb R}^{(n)}}_{A_p}=d+A_pJ_n\,d\theta .
\end{equation}
Since $J_n\in\mathfrak{so}(n)$, the one-form $A_pJ_n\,d\theta$ is $\mathfrak{so}(n)$-valued, equivalently skew-adjoint with respect to the Euclidean fiber metric, and therefore $\nabla^{E_{\mathbb R}^{(n)}}_{A_p}$ is a metric connection on $E_{\mathbb R}^{(n)}$.

We denote the complexification of $E_{\mathbb R}^{(n)}$ by
\begin{equation}\label{eq:hrank_complexified_bundle} E_{\mathbb C}^{(n)}=E_{\mathbb R}^{(n)}\otimes_{\mathbb R}\mathbb C .
\end{equation}
The connection $\nabla^{E_{\mathbb R}^{(n)}}_{A_p}$ complexifies to the unitary connection
\begin{equation}\label{eq:hrank_complexified_connection} \nabla^{E_{\mathbb C}^{(n)}}_{A_p}=d+A_pJ_n\,d\theta ,
\end{equation}
where $J_n$ now acts $\mathbb C$-linearly on the complex vector space
$\mathbb C^n$.

The base reflection is $\mathbf r(t,\theta)=(t,-\theta)$. Since $\mathbf r^*(d\theta)=-d\theta$, the angular connection form changes sign under pullback by $\mathbf r$. To obtain a reflection-compatible twist, the fiber action must therefore reverse the skew-symmetric generator $J_n$.

\begin{definition}[Reflection-compatible higher-rank orthogonal twist]
\label{def:hrank_reflection_compatible_twist}
A reflection-compatible higher-rank orthogonal twist is a real triple $(E_{\mathbb R}^{(n)},J_n,C_n)$, where $E_{\mathbb R}^{(n)}=M\times\mathbb R^n$, $J_n\in\mathfrak{so}(n)$, and $C_n\in O(n)$, such that $C_n$ is an orthogonal involution and
\begin{equation}\label{eq:hrank_twist_definition_condition} C_nJ_nC_n^{-1}=-J_n, \qquad C_n^2=I .
\end{equation}
Therefore, $C_n$ conjugates $J_n$ to $-J_n$. For such a real triple, a path $p\mapsto A_p$ determines a family of real orthogonal connections
$\nabla^{E_{\mathbb R}^{(n)}}_{A_p}=d+A_pJ_n\,d\theta$.
Its complexification acts on $E_{\mathbb C}^{(n)}=E_{\mathbb R}^{(n)}\otimes_{\mathbb R}\mathbb C$ by the same formula, with $J_n$ extended $\mathbb C$-linearly.
\end{definition}

\subsection{Normal form of the twisting matrix}
\label{subsec:hrank_normal_form}

We use the following normal form throughout the higher-rank analysis. Since $J_n$ is real and skew-symmetric, there is an orthonormal basis of $\mathbb R^n$ in which $J_n$ is block diagonal. In other words, there exist real numbers  $\mu_1,\ldots,\mu_{r}>0$ and an integer $z\geq0$, with $2r+z=n$, such that
\begin{equation}\label{eq:hrank_Jn_normal_form} J_n=\mu_1J\oplus\mu_2J\oplus\cdots\oplus\mu_{r}J\oplus 0_z , \qquad n = 2r +z
\end{equation}
where
\begin{equation}\label{eq:hrank_standard_J_block}
J=
\begin{pmatrix}
0&-1\\
1&0
\end{pmatrix}.
\end{equation}
The block $0_z$ denotes the zero matrix on the neutral subspace
$\mathbb R^z$. If $n$ is odd, then $z\geq1$. If $z=0$, the twist has no neutral directions. Since every nonzero skew-symmetric block has real dimension $2$, the neutral dimension $z=\dim\ker J_n$ has the same parity as $n$. In particular, if $n$ is odd, then $z\geq1$. If $n$ is even, then $z$ is even and may be zero.

In this normal form, a canonical choice of reflection-compatible fiber involution is obtained by putting the rank-two reflection on each rotating two-plane and an arbitrary orthogonal involution on the neutral part. For the rotating blocks we use
\begin{equation}\label{eq:hrank_standard_C_block}
C=
\begin{pmatrix}
1&0\\
0&-1
\end{pmatrix},
\end{equation}
so that $CJC^{-1}=-J$. Thus, one convenient choice is
\begin{equation}\label{eq:hrank_Cn_normal_form} C_n=C\oplus C\oplus\cdots\oplus C\oplus C_0 ,
\end{equation}
where there are $r$ copies of $C$, and $C_0\in O(z)$ is any orthogonal involution on the neutral block. Since $0_z$ commutes with every $C_0$, and since $C^2=I$ and $C_0^2=I$, this choice satisfies \eqref{eq:hrank_twist_definition_condition}.

The positive numbers $\mu_j$ measure the angular velocity with which the connection rotates the $j$-th real two-plane in the twisting bundle. When all $\mu_j=1$, the higher-rank model is simply a direct sum of identical rank-two twists. When the $\mu_j$'s are different, different internal two-planes cross at different values of the external parameter $A_p$. This is the source of the weighted crossing equations
appearing later.

\begin{remark}[Gauge convention for the higher-rank parameter]
\label{rem:hrank_gauge_convention}
In the scalar line model of~\autoref{sec:setup}, the parameter $A$ is naturally understood modulo the usual integer gauge shifts. In the higher-rank orthogonal model, however, we keep a fixed real representative $A_p$ along the chosen path. Thus, the
crossing equations below are equations for this chosen real representative. The effective crossing parameters are $k\pm\mu_jA_p$.  In the scalar line-bundle case, integer shifts of $A$ may be absorbed by gauge transformations.  In the higher-rank model, however, an integer shift $A_p\mapsto A_p+\ell$ changes the effective parameters by $\pm\mu_j\ell$.  For arbitrary real weights $\mu_j$, this need not be an integer Fourier shift and can not be identified with the original path unless additional integrality conditions are imposed.  We therefore fix the trivialization in which
$\nabla^{E_{\mathbb R}^{(n)}}_{A_p}=d+A_pJ_n\,d\theta$
is written, and compute crossing sets for this fixed real value $A_p$.

\end{remark}

\subsection{The lifted reflection}
\label{subsec:hrank_lifted_reflection}

We lift the base reflection $\mathbf r(t,\theta)=(t,-\theta)$ to the spinor factor by the unitary involution
\begin{equation}\label{eq:hrank_spinorial_reflection_lift} U_{\mathbf r}=\sigma_1.
\end{equation}
With the Clifford conventions of~\autoref{sec:setup}, this lift satisfies
\begin{equation}\label{eq:hrank_spinorial_reflection_relations}
U_{\mathbf r}^2=I,
\qquad
U_{\mathbf r}\sigma_1U_{\mathbf r}^{-1}=\sigma_1,
\qquad
U_{\mathbf r}\sigma_2U_{\mathbf r}^{-1}=-\sigma_2.
\end{equation}
For a reflection-compatible twist
$(E_{\mathbb R}^{(n)},J_n,C_n)$, we define the lifted reflection on
$S\otimes E_{\mathbb C}^{(n)}$ by
\begin{equation}\label{eq:hrank_reflection_lift} \mathcal U_{\mathbf r}^{E_{\mathbb C}^{(n)}}=U_{\mathbf r}\,C_n\,\mathbf r^* .
\end{equation}
Here $\mathbf r^*$ acts on the base variable, $C_n$ acts on the twisting factor
$E_{\mathbb C}^{(n)}$, and $U_{\mathbf r}$ acts on the spinor factor.

The corresponding twisted Dirac operator is
\begin{equation}\label{eq:hrank_twisted_Dirac}
D_p^{E_{\mathbb C}^{(n)}}
=
i\sigma_1\left(\partial_t+\frac{f'(t)}{2f(t)}\right)
+
i\sigma_2\,\frac{1}{f(t)}\left(\partial_\theta+A_pJ_n\right).
\end{equation}
This is the same warped-cylinder Dirac operator as before, except that the scalar line-twist term is replaced by the skew-symmetric endomorphism $A_pJ_n$.

\begin{proposition}[Reflection equivariance of the higher-rank twisted operator]
\label{prop:hrank_reflection_equivariance}
Assume that $C_n^2=I$ and $C_nJ_nC_n^{-1}=-J_n$. Then the lifted reflection $\mathcal U_{\mathbf r}^{E_{\mathbb C}^{(n)}}$ is an involutive lift of the base reflection, and for every $p\in[0,1]$, the higher-rank twisted Dirac operator satisfies
\begin{equation}\label{eq:hrank_reflection_commutation}
\mathcal U_{\mathbf r}^{E_{\mathbb C}^{(n)}}D_p^{E_{\mathbb C}^{(n)}}
\left(\mathcal U_{\mathbf r}^{E_{\mathbb C}^{(n)}}\right)^{-1}
=
D_p^{E_{\mathbb C}^{(n)}} .
\end{equation}
\end{proposition}

\begin{proof}
The $t$-derivative part of $D_p^{E_{\mathbb C}^{(n)}}$ is unchanged by the base reflection, and  $U_{\mathbf r}\sigma_1U_{\mathbf r}^{-1}=\sigma_1$. For the angular part, the pullback by $\mathbf r$ sends $\partial_\theta$ to $-\partial_\theta$, while the fiber operator $C_n$ sends $J_n$ to $-J_n$. Therefore the angular covariant derivative transforms as
\begin{equation}\label{eq:hrank_angular_covariant_transform}
C_n\mathbf r^*\left(\partial_\theta+A_pJ_n\right)(\mathbf r^*)^{-1}C_n^{-1}
=
-\left(\partial_\theta+A_pJ_n\right).
\end{equation}
On the other hand, the spinorial reflection satisfies $U_{\mathbf r}\sigma_2U_{\mathbf r}^{-1}=-\sigma_2$. Thus the sign change in the angular covariant derivative is exactly cancelled by the
sign change in the Clifford matrix $\sigma_2$. Hence the full operator is preserved by the lifted reflection, which proves \eqref{eq:hrank_reflection_commutation}.
\end{proof}

\begin{remark}[Complexification, real structure, and representation-ring convention]
\label{rem:representation_ring_convention}
The separated-block calculation is performed here after complexifying the
real orthogonal twisting bundle $E_{\mathbb R}^{(n)}$. This is only a diagonalization tool: a real skew-symmetric block has no real eigenlines, but after complexification
it splits into the two conjugate eigenlines with eigenvalues $+i\mu_j$ and
$-i\mu_j$. Thus an individual separated block
$\xi_{j,+}$ or $\xi_{j,-}$ is a complex block, not by itself a real
one-dimensional $O(2)$-representation.

However, the real structure does not disappear. Complex conjugation exchanges the two fiber
eigenlines $\xi_{j,+}$ and $\xi_{j,-}$, while reflection exchanges the Fourier
modes $k$ and $-k$. Therefore, the final representation class is assigned only after
regrouping the conjugate and reflection-paired complex blocks into the real
$O(2)$-representations from which they arose.

For $k>0$, the Fourier pair $e^{ik\theta}$ and $e^{-ik\theta}$ gives the real
two-dimensional irreducible $O(2)$-representation $\rho_k$. At $k=0$, the
self-paired block decomposes according to the reflection action into the trivial and
determinant representations $\mathbf 1$ and $\det$. Hence the symbols
$[\rho_k]$, $[\mathbf 1]$, and $[\det]$ mentioned are real
$O(2)$-representations, while the crossing multiplicities are computed using the complex operator on $S\otimes E_{\mathbb C}^{(n)}$.

Throughout the paper, $\operatorname{sf}_{O(2)}$ denotes the representation-valued
spectral flow obtained from the complex crossing spaces after regrouping conjugate and
reflection-paired blocks into their associated real $O(2)$-representations. Thus the coefficient
of $[\rho_k]$, $[\mathbf 1]$, or $[\det]$ gives the crossing-form signature on the
corresponding complex crossing block, while the symbol gives the associated real
$O(2)$-representation. This is the convention under which the dimension map gives the ordinary integer-valued spectral flow of the complexified operator family. With this convention, applying the real dimension map
$\dim:RO(O(2))\to\mathbb Z$ recovers the ordinary
integer-valued spectral flow obtained after forgetting the $O(2)$-representation.
\end{remark}

\subsection{Complex weight blocks of the higher-rank twist}
\label{subsec:hrank_complex_weight_blocks}
We now diagonalize the fiber action of $J_n$ after complexification.  The
real representation-ring convention is the one fixed in
Remark~\ref{rem:representation_ring_convention}.

For each rotating
two-plane in the normal form \eqref{eq:hrank_Jn_normal_form}, let $\xi_{j,+}$ and
$\xi_{j,-}$ be complex eigenvectors satisfying
\begin{equation}\label{eq:hrank_rotating_eigenvectors} J_n\xi_{j,+}=i\mu_j\xi_{j,+}, \qquad J_n\xi_{j,-}=-i\mu_j\xi_{j,-}, \qquad 1\leq j\leq r .
\end{equation}
For the neutral block $\ker J_n$, let $\eta_1,\ldots,\eta_z$ be a complex basis
satisfying
\begin{equation}\label{eq:hrank_neutral_eigenvectors} J_n\eta_a=0, \qquad 1\leq a\leq z .
\end{equation}
Thus, the complexified twisting space decomposes as a direct sum of rotating blocks
and neutral blocks.

For a Fourier mode $k\in\mathcal K$, the rotating separated sections are
\begin{equation}\label{eq:hrank_rotating_separated_section}
\Psi_{k,j,\pm}(t,\theta)
=
e^{ik\theta}
\binom{u_{k,j,\pm}(t)}{v_{k,j,\pm}(t)}
\otimes \xi_{j,\pm}.
\end{equation}
Using \eqref{eq:hrank_rotating_eigenvectors}, the angular covariant derivative acts on
these sections by
\begin{equation}\label{eq:hrank_rotating_block_action}
\left(\partial_\theta+A_pJ_n\right)\Psi_{k,j,\pm}
=
i\left(k\pm\mu_jA_p\right)\Psi_{k,j,\pm}.
\end{equation}
Therefore the effective scalar parameters in the rotating blocks are
\begin{equation}\label{eq:hrank_rotating_block_parameters} m_{k,j,+}(p)=k+\mu_jA_p, \qquad m_{k,j,-}(p)=k-\mu_jA_p .
\end{equation}

The neutral separated sections are
\begin{equation}\label{eq:hrank_neutral_separated_section}
\Psi_{k,a,0}(t,\theta)
=
e^{ik\theta}
\binom{u_{k,a,0}(t)}{v_{k,a,0}(t)}
\otimes \eta_a .
\end{equation}
Since $J_n\eta_a=0$, the angular covariant derivative acts on a neutral block by
\begin{equation}\label{eq:hrank_neutral_block_action} \left(\partial_\theta+A_pJ_n\right)\Psi_{k,a,0} = ik\,\Psi_{k,a,0}.
\end{equation}
Thus the neutral effective parameter is
\begin{equation}\label{eq:hrank_neutral_block_parameter} m_{k,a,0}(p)=k .
\end{equation}

Substituting the rotating separated section \eqref{eq:hrank_rotating_separated_section}
into the higher-rank Dirac operator gives the same scalar radial operator as in the
rank-two calculation, with $m$ replaced by $m_{k,j,\pm}(p)$. More explicitly, the
eigenvalue equation
$D_p^{E_{\mathbb C}^{(n)}}\Psi_{k,j,\pm}=\lambda\Psi_{k,j,\pm}$ becomes
\begin{equation}\label{eq:hrank_rotating_radial_system}
\left(
\partial_t+\frac{f'(t)}{2f(t)}+\frac{m_{k,j,\pm}(p)}{f(t)}
\right)v_{k,j,\pm}(t)
=
-i\lambda u_{k,j,\pm}(t),
\qquad
\left(
\partial_t+\frac{f'(t)}{2f(t)}-\frac{m_{k,j,\pm}(p)}{f(t)}
\right)u_{k,j,\pm}(t)
=
-i\lambda v_{k,j,\pm}(t).
\end{equation}
Similarly, in a neutral block, the same formula holds with $m_{k,j,\pm}(p)$
replaced by the constant parameter $k$.

\subsection{Boundary crossings}
\label{subsec:hrank_boundary_crossing_consequences}

The scalar radial calculation is unchanged in higher rank. The only new feature is the list of scalar parameters to which the calculation is applied. In each rotating block, the regularized APS crossing occurs when $m_{k,j,\pm}(p)=0$. Therefore, the moving crossing equations for the $j$-th rotating two-plane are
\begin{equation}\label{eq:hrank_crossing_equations} A_{p_*}=\frac{k}{\mu_j} \qquad\text{or}\qquad A_{p_*}=-\frac{k}{\mu_j}.
\end{equation}
Thus, each positive weight $\mu_j$ rescales the values of $A_p$ at which the $k$-th Fourier block can cross.

The neutral blocks are stationary.  Since $m_{k,a,0}(p)=k$, they do not
produce $A_p$-dependent crossings.  For $k\neq0$, the neutral block stays
away from the boundary-kernel value.  In the periodic sector, however, the
neutral zero mode satisfies $m_{0,a,0}(p)=0$ for all $p$.  This is not an
isolated crossing, so it is handled by a fixed APS kernel completion rather
than by the rotating-block regularization.

\begin{definition}[Reflection-compatible neutral-sector convention]
\label{def:hrank_neutral_sector_convention}
Let $z>0$. In the periodic spin sector, the neutral $k=0$ blocks satisfy
$m_{0,a,0}(p)=0$
for every $p\in[0,1]$. This is a stationary boundary-kernel sector, not an
isolated moving crossing. We therefore complete the APS condition on this
finite-dimensional sector by a fixed neutral convention.

Let
\begin{equation}
\widehat\gamma_0\psi
=
f(0)^{1/2}\psi|_{Y_0}^{k=0,\ker J_n},
\qquad
\widehat\gamma_T\psi
=
f(T)^{1/2}\psi|_{Y_T}^{k=0,\ker J_n}
\end{equation}
denote the scaled boundary traces in the neutral zero sector. We impose the
two-endpoint graph condition
\begin{equation}\label{eq:neutral_global_completion} \widehat\gamma_T\psi=-\widehat\gamma_0\psi.
\end{equation}
Equivalently,
\begin{equation}
\psi|_{Y_T}^{k=0,\ker J_n}
=
-\left(\frac{f(0)}{f(T)}\right)^{1/2}
\psi|_{Y_0}^{k=0,\ker J_n}.
\end{equation}

This is a $p$-independent finite-dimensional completion of the APS condition on
the neutral boundary-kernel sector. It is not an additional condition on a
sector where strict APS has already imposed something: on the boundary kernel,
the strict projection $P_{>0}$ imposes no condition, so a kernel completion has
to be specified.

The sign in \eqref{eq:neutral_global_completion} is compatible with the
two-endpoint Green form. Indeed, Green's formula contains
\begin{equation}
\mathcal G(\psi_\partial,\phi_\partial)
=
\int_{\partial M}
\left\langle i\gamma(N)\psi_\partial,\phi_\partial\right\rangle
\,d\operatorname{vol}_{\partial M},
\end{equation}
and the inward normal is $+\partial_t$ on $Y_0$ and $-\partial_t$ on $Y_T$.
After the rescaling by $f(0)^{1/2}$ and $f(T)^{1/2}$, the neutral zero-sector
part of the boundary form has the form
\begin{equation}
\mathcal G_{\text{neut}}(\psi,\phi)
=
\mathfrak b(\widehat\gamma_0\psi,\widehat\gamma_0\phi)
-
\mathfrak b(\widehat\gamma_T\psi,\widehat\gamma_T\phi),
\end{equation}
where $\mathfrak b$ is the endpoint Hermitian form induced by
$i\gamma(\partial_t)$ on the neutral zero sector. If both boundary values satisfy
\eqref{eq:neutral_global_completion}, then
\begin{equation}
\mathfrak b(\widehat\gamma_T\psi,\widehat\gamma_T\phi)
=
\mathfrak b(\widehat\gamma_0\psi,\widehat\gamma_0\phi),
\end{equation}
so the two endpoint contributions cancel. Thus the boundary Green form vanishes
on this graph. Since the graph has half the dimension of the total two-endpoint
neutral boundary space, it is Lagrangian for the total boundary Green form and
defines a self-adjoint finite-dimensional APS completion.

The same condition also removes the stationary neutral zero solution. A
zero-energy neutral solution has constant scaled trace,
$\widehat\gamma_T\psi=\widehat\gamma_0\psi$.
Together with \eqref{eq:neutral_global_completion}, this gives
$\widehat\gamma_0\psi=\widehat\gamma_T\psi=0$.
Hence the persistent neutral zero mode is removed, and the neutral zero sector
does not contribute to the moving spectral flow.
\end{definition}

\subsection{Hypotheses for the higher-rank regularized family}
\label{subsec:hrank_standing_assumptions}

For the higher-rank formula later in the paper, we impose the following 
hypotheses. These are the standard hypotheses used in the crossing-form approach to
spectral flow for paths of self-adjoint Fredholm operators; see, for example,
\cite{Phillips1996,Waterstraat2017FredholmSpectralFlow,RobbinSalamon1995SpectralFlowMaslov}.

\begin{enumerate}[label=\textup{(\roman*)}]

\item \textbf{Regularity of the parameter path.}
The path $p\mapsto A_p$ is a continuously differentiable real-valued path on $[0,1]$.

\item \textbf{Fixed admissible regularized APS convention.}
We fix the admissible regularized APS boundary family used in~\autoref{sec:hrank_separated_blocks_regularized_APS}. This includes the scalar
regularized APS boundary family in the moving blocks, the chosen zero-block
regularization, and the fixed neutral-sector convention of
Definition~\ref{def:hrank_neutral_sector_convention}.

\item \textbf{Endpoint invertibility.}
The endpoint operators
\begin{equation}\label{eq:hrank_endpoint_invertibility_assumption} D^{E_{\mathbb C}^{(n)}}_{0,\APSreg}, \qquad D^{E_{\mathbb C}^{(n)}}_{1,\APSreg}
\end{equation}
are invertible. Hence the spectral flow of the regularized path is defined without
endpoint correction terms.

\item \textbf{Crossing form and regular crossings.}
Let $p_*\in(0,1)$ be a crossing parameter, so that
\begin{equation}\label{eq:hrank_crossing_kernel_condition} \ker D^{E_{\mathbb C}^{(n)}}_{p_*,\APSreg}\neq 0 .
\end{equation}
After using the local graph trivialization supplied by the regularized boundary family,
the crossing form at $p_*$ is the Hermitian form
\begin{equation}\label{eq:hrank_operator_crossing_form}
\Gamma_{p_*}(u,v)
=
\left\langle
\left.\frac{d}{dp}\right|_{p=p_*}
D^{E_{\mathbb C}^{(n)},\text{loc}}_{p,\APSreg}u,
v
\right\rangle_{L^2(M)},
\qquad
u,v\in\ker D^{E_{\mathbb C}^{(n)}}_{p_*,\APSreg}.
\end{equation}
Here $D^{E_{\mathbb C}^{(n)},\text{loc}}_{p,\APSreg}$ denotes the locally trivialized
operator family near $p_*$. In a separated scalar block, this crossing form is
represented by the derivative of the corresponding scalar transfer determinant. Thus,
for a moving block with effective parameter $m_{k,j,\sigma}(p)$, regularity is ensured by
\begin{equation}\label{eq:hrank_regular_crossing_condition} \left.\frac{d}{dp}\right|_{p=p_*}m_{k,j,\sigma}(p)\neq0,
\end{equation}
together with the admissibility condition that the scalar transfer determinant has a
simple zero at $m=0$.

The crossing is called regular if the form $\Gamma_{p_*}$ is nondegenerate. In that
case, its contribution to spectral flow is the signature
\begin{equation}\label{eq:hrank_crossing_sign_signature} \operatorname{sign}\Gamma_{p_*}.
\end{equation}

Throughout the spectral-flow formula, we require that all moving zeroes of
$m_{k,j,\pm}(p)$ are isolated and regular; paths for which a rotating block
satisfies $m_{k,j,\pm}(p)\equiv0$ are not included in the regular-crossing
case.
 \end{enumerate}

These hypotheses are used only to put the regularized APS family into the standard
regular-crossing framework for spectral flow. The higher-rank structure enters through
the separated effective parameters
\begin{equation}\label{eq:hrank_standing_effective_parameters} m_{k,j,+}(p)=k+\mu_jA_p, \qquad m_{k,j,-}(p)=k-\mu_jA_p, \qquad m_{k,a,0}(p)=k,
\end{equation}
and through the resulting $O(2)$-representation carried by each crossing space.

\begin{remark}[Simultaneous crossings.]
If several rotating blocks cross zero at the same parameter value, we require either
that the path has been perturbed within the reflection-compatible class so that the
crossings occur at distinct nearby parameter values, or that the total crossing form
$\Gamma_{p_*}$ is nondegenerate and block-diagonal with respect to the separated
block decomposition.

Under this block-diagonal convention, the local representation-valued contribution is
the sum of the contributions of the individual crossing blocks. If the crossing form
mixes blocks carrying the same $O(2)$-representation class, then the individual block
signs are not assigned separately. One first restricts $\Gamma_{p_*}$ to the combined
multiplicity space for that representation class and uses the signature of this
restricted form as the local coefficient.
\end{remark}

\section{Separated blocks and regularized APS boundary conditions}
\label{sec:hrank_separated_blocks_regularized_APS}

We now impose the regularized APS convention used in the spectral-flow formula.
The scalar construction is applied to each moving separated block, while the stationary
neutral zero sector uses the fixed completion of
Definition~\ref{def:hrank_neutral_sector_convention}. The resulting family is then
reassembled into the real $O(2)$-representation blocks specified in
Remark~\ref{rem:representation_ring_convention}.

\subsection{Separated scalar blocks and boundary operators}
\label{subsec:hrank_separated_scalar_blocks_boundary_operators}

We use the separated-block decomposition from
\autoref{subsec:hrank_complex_weight_blocks}.  The allowed Fourier modes are
$k\in\mathcal K$, with $\mathcal K$ as in \eqref{eq:Kdef}.  The rotating
two-planes of $J_n$ are indexed by
$\mathcal J=\{1,\ldots,r\}$,
and the neutral directions in $\ker J_n$ are indexed by
$\mathcal N=\{1,\ldots,z\}$.

For $j\in\mathcal J$, $a\in\mathcal N$, and
$\sigma\in\{+,-\}$, the effective scalar parameters are
\begin{equation}\label{eq:sep_effective_parameters} m_{k,j,+}(p)=k+\mu_jA_p, \qquad m_{k,j,-}(p)=k-\mu_jA_p, \qquad m_{k,a,0}(p)=k .
\end{equation}
Thus only the rotating blocks depend on the path parameter $A_p$; the neutral
blocks are stationary.

In every separated block, the eigenvalue equation
$D_p^{E_{\mathbb C}^{(n)}}\Psi=\lambda\Psi$
reduces to the same scalar radial system
\begin{equation}\label{eq:sep_scalar_radial_system}
\left(
\partial_t+\frac{f'(t)}{2f(t)}+\frac{m}{f(t)}
\right)v(t)
=
-i\lambda u(t),
\qquad
\left(
\partial_t+\frac{f'(t)}{2f(t)}-\frac{m}{f(t)}
\right)u(t)
=
-i\lambda v(t),
\end{equation}
with
\begin{equation}
m=m_{k,j,\sigma}(p)
\quad\text{in a rotating block},
\qquad
m=m_{k,a,0}(p)=k
\quad\text{in a neutral block}.
\end{equation}

Let $B_{0,p}^{(n)}$ and $B_{T,p}^{(n)}$ denote the higher-rank boundary
Dirac operators on $Y_0$ and $Y_T$.  With the inward-normal convention fixed
in \autoref{sec:setup}, they are
\begin{equation}\label{eq:sep_boundary_operators_higher_rank}
B_{0,p}^{(n)}
=
\frac{1}{f(0)}\,\sigma_3
\left(-i\partial_\theta-iA_pJ_n\right),
\qquad
B_{T,p}^{(n)}
=
-\frac{1}{f(T)}\,\sigma_3
\left(-i\partial_\theta-iA_pJ_n\right).
\end{equation}
On a rotating block $(k,j,\sigma)$, these become
\begin{equation}\label{eq:sep_boundary_operators_rotating_block}
B_{0,p}^{(k,j,\sigma)}
=
\frac{m_{k,j,\sigma}(p)}{f(0)}\,\sigma_3,
\qquad
B_{T,p}^{(k,j,\sigma)}
=
-\frac{m_{k,j,\sigma}(p)}{f(T)}\,\sigma_3.
\end{equation}
On a neutral block $(k,a,0)$, they become
\begin{equation}\label{eq:sep_boundary_operators_neutral_block}
B_{0,p}^{(k,a,0)}
=
\frac{k}{f(0)}\,\sigma_3,
\qquad
B_{T,p}^{(k,a,0)}
=
-\frac{k}{f(T)}\,\sigma_3.
\end{equation}

For $m\neq0$, the ordinary APS condition in a scalar block is determined by
the sign of $m$.  If $m>0$, the domain condition is
$u(0)=0$ and $v(T)=0$,
whereas if $m<0$, it is
$v(0)=0$ and $u(T)=0$.
The jump at $m=0$ is the only scalar boundary-kernel phenomenon that must be
regularized in the moving rotating blocks.

\subsection{The scalar regularized APS boundary family}
\label{subsec:hrank_scalar_regularized_APS}

We fix one scalar regularized APS model and use it in every moving block. Let
$\alpha:\mathbb R\to[0,\pi/2]$ be a smooth nondecreasing function satisfying
\begin{equation}\label{eq:sep_alpha_conditions}
\alpha(x)=0\ \text{for }x\leq -1,
\qquad
\alpha(x)=\frac{\pi}{2}\ \text{for }x\geq1,
\qquad
\alpha(-x)=\frac{\pi}{2}-\alpha(x).
\end{equation}
The symmetry condition implies $\alpha(0)=\pi/4$.

After the finite set of rotating crossing parameters has been determined, let
$\delta>0$ be fixed so that the regularization intervals are disjoint and
the intervals avoid the endpoints $p=0,1$. The additional scalar admissibility
restriction on $\delta$ will be derived from the scalar transfer determinant
in \autoref{subsec:hrank_scalar_transfer_admissibility}.

If several separated blocks cross at the same parameter value, we use the
simultaneous-crossing convention fixed in
\autoref{subsec:hrank_standing_assumptions}. Thus the local contribution is
computed either after separating the crossings by a small
reflection-compatible perturbation, or from the signature of the total
crossing form on the relevant multiplicity space.

Let
\begin{equation}\label{eq:sep_spinor_basis} e_+= \binom{1}{0}, \qquad e_-= \binom{0}{1}.
\end{equation}
For a scalar block with parameter $m$, define the regularized positive boundary
lines by
\begin{equation}\label{def:scalar_regularized_APS}
\ell_0^{\text{reg}}(m)
=
\mathbb C\left(
\cos\alpha\left(\frac{m}{\delta}\right)e_-
+
i\sin\alpha\left(\frac{m}{\delta}\right)e_+
\right),
\qquad
\ell_T^{\text{reg}}(m)
=
\mathbb C\left(
\cos\alpha\left(\frac{m}{\delta}\right)e_+
-
i\sin\alpha\left(\frac{m}{\delta}\right)e_-
\right).
\end{equation}
These are the projected-out positive APS lines. The corresponding domain lines are the
complementary lines
\begin{equation}\label{eq:sep_reg_domain_lines}
\ell_0^{\text{dom}}(m)
=
\mathbb C\left(
\cos\alpha\left(\frac{m}{\delta}\right)e_+
+
i\sin\alpha\left(\frac{m}{\delta}\right)e_-
\right),
\qquad
\ell_T^{\text{dom}}(m)
=
\mathbb C\left(
-i\sin\alpha\left(\frac{m}{\delta}\right)e_+
+
\cos\alpha\left(\frac{m}{\delta}\right)e_-
\right).
\end{equation}
Thus, for each $t\in\{0,T\}$, the regularized APS domain line is
\begin{equation}\label{eq:sep_domain_line_kernel} L_t^{\text{reg}}(m)=\ell_t^{\text{dom}}(m)=\ker P_t^{\text{APS,reg}}(m).
\end{equation}

The regularization agrees with the ordinary APS convention outside the regularization interval
$|m|<\delta$. More precisely,
\begin{subequations}\label{eq:sep_reg_lines_negative_m}
\begin{align}
m<-\delta
\quad\Longrightarrow\quad
\ell_0^{\text{reg}}(m)=\mathbb C e_-,
\qquad
\ell_T^{\text{reg}}(m)=\mathbb C e_+,
\\
m>\delta
\quad\Longrightarrow\quad
\ell_0^{\text{reg}}(m)=\mathbb C e_+,
\qquad
\ell_T^{\text{reg}}(m)=\mathbb C e_-.
\end{align}
\end{subequations}
Since $P_t^{\text{APS,reg}}(m)$ projects onto $\ell_t^{\text{reg}}(m)$, the imposed
boundary condition is that the boundary value lies in the complementary domain line
$\ell_t^{\text{dom}}(m)$. Thus, for fixed $m<0$ and sufficiently small $\delta$,
\eqref{eq:sep_reg_lines_negative_m} gives
\begin{equation}
\ell_0^{\text{dom}}(m)=\mathbb C e_+,
\qquad
\ell_T^{\text{dom}}(m)=\mathbb C e_-,
\end{equation}
which is exactly the APS condition $v(0)=0$ and $u(T)=0$. Similarly, for fixed
$m>0$ and sufficiently small $\delta$, it gives
\begin{equation}
\ell_0^{\text{dom}}(m)=\mathbb C e_-,
\qquad
\ell_T^{\text{dom}}(m)=\mathbb C e_+,
\end{equation}
which is exactly the APS condition $u(0)=0$ and $v(T)=0$. Hence, the regularized
boundary condition reduces to the ordinary APS boundary condition away from the
regularization interval, and in the limit $\delta\to0$ for each fixed $m\neq0$.

The phases in \eqref{def:scalar_regularized_APS} and \eqref{eq:sep_reg_domain_lines} are
chosen so that the domain lines are Lagrangian for the boundary Green form defined in
\eqref{eq:boundary_green_form}. Since each crossing block is two-dimensional over
$\mathbb C$, a one-dimensional domain line is Lagrangian exactly when the Green form
vanishes on that line.

\begin{definition}[Scalar regularized APS projection]
\label{def:sep_scalar_regularized_APS_projection}
For a scalar block with effective parameter $m$, the regularized APS projection $P_t^{\text{APS,reg}}(m)$ is the orthogonal projection onto $\ell_t^{\text{reg}}(m)$. Its kernel is the regularized APS domain line $\ell_t^{\text{dom}}(m)$.
\end{definition}

\subsection{Scalar transfer determinant and admissibility}
\label{subsec:hrank_scalar_transfer_admissibility}

The higher-rank construction reduces the local crossing analysis to a scalar
boundary-value problem with an effective parameter $m$.  This scalar calculation
will be applied later to the rotating blocks by substituting
$m=m_{k,j,+}(p)$ or $m=m_{k,j,-}(p)$.

Let
\begin{equation}\label{eq:hrank_Lf_def} L_f=\int_0^T\frac{d\tau}{f(\tau)}.
\end{equation}
In a scalar block with parameter $m$, the zero-energy radial equation is
\begin{equation}\label{eq:hrank_scalar_zero_energy_system}
\left(
\partial_t+\frac{f'(t)}{2f(t)}+\frac{m}{f(t)}
\right)v(t)=0,
\qquad
\left(
\partial_t+\frac{f'(t)}{2f(t)}-\frac{m}{f(t)}
\right)u(t)=0.
\end{equation}
Hence
\begin{equation}\label{eq:hrank_scalar_zero_energy_solutions}
u(t)=a\,f(t)^{-1/2}
\exp\left(m\int_0^t\frac{d\tau}{f(\tau)}\right),
\qquad
v(t)=b\,f(t)^{-1/2}
\exp\left(-m\int_0^t\frac{d\tau}{f(\tau)}\right).
\end{equation}
After discarding the common nonzero scalar factor coming from
$f(T)^{-1/2}/f(0)^{-1/2}$, the zero-energy transfer matrix from $Y_0$ to
$Y_T$ is
\begin{equation}\label{eq:hrank_scalar_transfer_matrix}
\Phi_m
=
\begin{pmatrix}
e^{mL_f}&0\\
0&e^{-mL_f}
\end{pmatrix}.
\end{equation}

\begin{definition}[Admissible scalar regularized crossing]
\label{def:hrank_admissible_scalar_crossing}
The scalar regularized APS boundary family is called admissible at $m=0$ if
the zero-energy transfer determinant
\begin{equation}\label{eq:hrank_scalar_transfer_determinant_def} F(m) = \det\left( \ell_T^{\text{dom}}(m), \Phi_m\ell_0^{\text{dom}}(m) \right)
\end{equation}
satisfies
\begin{equation}\label{eq:hrank_admissibility_conditions} F(0)=0, \qquad F'(0)\neq0,
\end{equation}
and, after possibly replacing the regularization interval by a smaller fixed
positive interval,
$F(m)\neq0$ for $0<|m|<\delta$.
\end{definition}

We denote by $D_m^{\text{reg}}$ the scalar radial Dirac operator with effective parameter $m$ and with the regularized domain lines $L_0^{\text{reg}}(m)$ and $L_T^{\text{reg}}(m)$.

\begin{lemma}[Transfer determinant and scalar zero modes]
\label{lem:transfer_determinant_detects_zero_modes}
For the scalar regularized APS boundary problem with domain lines
$\ell_0^{\text{dom}}(m)$ and $\ell_T^{\text{dom}}(m)$, one has
$F(m)=0$ if and only if $\ker D_m^{\text{reg}}\neq 0$.
Consequently, if the scalar regularized APS boundary family is admissible then, after possibly replacing the regularization interval by a smaller fixed positive interval, one has
\begin{equation}
\ker D^{\text{reg}}_m = 0
\qquad
\text{for }0<|m|<\delta,
\qquad
\ker D^{\text{reg}}_0 \neq 0 .
\end{equation}
\end{lemma}

\begin{proof}
A zero-energy solution is determined by its initial boundary value at $Y_0$.
If this initial value lies in $\ell_0^{\text{dom}}(m)$, then its terminal
value at $Y_T$ lies in $\Phi_m\ell_0^{\text{dom}}(m)$.  Hence a nonzero
solution satisfying both endpoint boundary conditions exists exactly when
$\Phi_m\ell_0^{\text{dom}}(m) \cap \ell_T^{\text{dom}}(m) \neq \{0\}$.
Since both spaces are complex lines in $\mathbb C^2$, this is equivalent to
their linear dependence, which is detected by the determinant
$F(m)$.  The final statement follows from admissibility.
\end{proof}

\begin{lemma}[Admissibility of the scalar regularized APS boundary family]
\label{lem:hrank_scalar_boundary_family_admissible}
For the regularized domain lines in \eqref{eq:sep_reg_domain_lines}, the scalar
transfer determinant is, up to a nonzero scalar factor,
\begin{equation}\label{eq:hrank_scalar_transfer_determinant}
F(m)
=
\sin^2\alpha\left(\frac{m}{\delta}\right)e^{-mL_f}
-
\cos^2\alpha\left(\frac{m}{\delta}\right)e^{mL_f}.
\end{equation}
If
\begin{equation}\label{eq:hrank_scalar_admissibility_condition} \frac{2\alpha'(0)}{\delta}\neq L_f,
\end{equation}
then the scalar regularized APS boundary family is admissible at $m=0$.
This condition is open and can be imposed by avoiding at most one value of
$\delta$.  After replacing $\delta$ by a positive fixed interval, the determinant has no
additional zero in the punctured regularization interval
$0<|m|<\delta$.
\end{lemma}

\begin{proof}
Using \eqref{eq:hrank_scalar_transfer_matrix} and the domain vectors in
\eqref{eq:sep_reg_domain_lines}, we have
\begin{equation}\label{eq:hrank_transfer_on_domain_vector}
\Phi_m\ell_0^{\text{dom}}(m)
=
e^{mL_f}\cos\alpha\left(\frac{m}{\delta}\right)e_+
+
i e^{-mL_f}\sin\alpha\left(\frac{m}{\delta}\right)e_- .
\end{equation}
Taking the determinant of the two column vectors
$\ell_T^{\text{dom}}(m)$ and
$\Phi_m\ell_0^{\text{dom}}(m)$ gives
\eqref{eq:hrank_scalar_transfer_determinant}.  Since
$\alpha(0)=\pi/4$, this gives $F(0)=0$.  Differentiating
\eqref{eq:hrank_scalar_transfer_determinant} at $m=0$ gives
$\frac{2 \alpha^{'}}{\delta} -L_f$.  Therefore $F'(0)\neq0$
whenever \eqref{eq:hrank_scalar_admissibility_condition} holds.  The final
claim follows by replacing the regularization interval.
\end{proof}

\begin{remark}[Meaning of the scalar transversality condition]
The term $2\alpha'(0)/\delta$ measures the first-order rotation speed of
the regularized APS boundary lines, while $L_f$ measures the first-order
effect of zero-energy transfer through the cylinder.  Thus
$\frac{2\alpha'(0)}{\delta}\neq L_f$
says that the boundary-line rotation is not exactly cancelled by the interior
transfer.
\end{remark}

\begin{remark}[Maslov interpretation]
The scalar transfer determinant is a local coordinate for the Maslov
intersection condition.  Indeed, after transporting the initial domain line by
the zero-energy transfer map, a zero mode occurs exactly when
$\Phi_m\ell_0^{\text{dom}}(m) \cap \ell_T^{\text{dom}}(m) \neq \{0\}$.
Thus the zeros of $F(m)$ are precisely the local intersections of the
corresponding path of Lagrangian pairs.
\end{remark}

\begin{corollary}[Scalar operator crossing]
\label{cor:hrank_scalar_boundary_to_operator_crossing}
\label{lem:hrank_scalar_boundary_to_operator_crossing}
Assume that the scalar regularized APS boundary family is admissible in the
sense of Lemma~\ref{lem:hrank_scalar_boundary_family_admissible}.  Let
$p\mapsto m(p)$ be continuously differentiable with $m(p_*)=0$.  If
$\left.\frac{d}{dp}\right|_{p=p_*}m(p)\neq0$,
then the corresponding scalar operator crossing at $p_*$ is regular.
\end{corollary}

\begin{proof}
The local crossing function is $F(m(p))$.  By
Lemma~\ref{lem:hrank_scalar_boundary_family_admissible}, $F'(0)\neq0$.
Hence
\begin{equation}
\left.\frac{d}{dp}\right|_{p=p_*}F(m(p))
=
F'(0)
\left.\frac{d}{dp}\right|_{p=p_*}m(p),
\end{equation}
which is nonzero by assumption.
\end{proof}

\subsection{Regularized APS projections in rotating higher-rank blocks}
\label{subsec:hrank_reg_APS_rotating_blocks}

We now apply the scalar regularization to one representative from each nonzero
reflection-paired block and obtain its reflected partner by equivariance. If one regularized a block and its reflected partner independently using
the scalar parameter $m$ and $-m$, the local sign convention could be obscured.
Instead, for $k\geq1$, $j\in\mathcal J$, and $\sigma\in\{+,-\}$, we first fix
the representative block $(k,j,\sigma)$ and define the regularized positive
boundary line at $Y_t$ by
\begin{equation}\label{eq:sep_rotating_positive_block_line}
E_{t,k,j,\sigma,p}^{+,\text{reg},\text{rep}}
=
\ell_t^{\text{reg}}\left(m_{k,j,\sigma}(p)\right)
\otimes e^{ik\theta}\xi_{j,\sigma}.
\end{equation}
We do not define the reflected partner independently. Instead, it is obtained by
applying the lifted reflection $\mathcal U_{\mathbf r}^{E_{\mathbb C}^{(n)}}$. Thus the
regularization on the reflected block is the reflected regularization of the
representative block, not a second independent scalar choice.  The displayed definition
uses the scalar parameter $m_{k,j,\sigma}(p)$.  In the optional affine sign convention
used later, the same Lagrangian path may be reparametrized inside the regularization interval by
the local $A$-coordinate; this changes only the orientation convention for the local
crossing sign and not the crossing set or the main blockwise formula. Outside the regularization
interval this agrees with the ordinary APS positive space on the reflected block; inside
the interval it is the equivariant continuation required for an $O(2)$-equivariant
family. The regularized positive space over the nonzero reflection-paired Fourier block
$\{k,-k\}$ is therefore

\begin{equation}\label{eq:sep_nonzero_block_positive_space}
E_{t,k,p}^{+,\text{reg}}
=
\bigoplus_{j\in\mathcal J}
\bigoplus_{\sigma\in\{+,-\}}
\left(
E_{t,k,j,\sigma,p}^{+,\text{reg},\text{rep}}
\oplus
\mathcal U_{\mathbf r}^{E_{\mathbb C}^{(n)}}
\left(
E_{t,k,j,\sigma,p}^{+,\text{reg},\text{rep}}
\right)
\right),
\qquad
k\geq1.
\end{equation}
It is $O(2)$-equivariant by construction: rotations preserve the
Fourier mode, while the reflection exchanges $k$ and $-k$ and simultaneously sends
the chosen block to its reflected partner. In particular, a single chosen complex block is not counted separately in
$RO(O(2))$. The $RO(O(2))$-class is attached to the whole
reflection-paired block after the complex blocks have been regrouped into the
underlying real $O(2)$-representation.

At the zero Fourier mode $k=0$, the rotating blocks are self-paired under reflection.
For each $j\in\mathcal J$, the two moving parameters are $\mu_jA_p$ and
$-\mu_jA_p$. The possible moving zero-mode crossing occurs only when $A_p=0$.
At such a crossing, the corresponding real $O(2)$-representation is a copy of
$\mathbf 1\oplus\det$. We denote this copy by
\begin{equation}\label{eq:sep_zero_block_representation_copy} (\mathbf 1\oplus\det)_j.
\end{equation}
On this zero Fourier block we fix an $O(2)$-equivariant regularized APS
boundary family
\begin{equation}\label{eq:sep_zero_block_projection_j} P_{t,p}^{(0,j),\APSreg}, \qquad t\in\{0,T\},
\end{equation}
whose kernel is a continuous family of Lagrangian domain subspaces in the corresponding
finite-dimensional zero Fourier boundary space. It agrees with the ordinary APS
convention away from the interval $|A_p|<\delta$. We require this zero-block
boundary family to be admissible in the same transfer-determinant sense as the scalar
boundary family above: the associated zero-energy transfer determinant has a simple zero
at $A_p=0$ and has no other zero inside the punctured regularization interval.

\begin{definition}[Admissible rotating zero-block regularization]
\label{def:admissible_zero_block_regularization}
For each rotating two-plane $j$, an admissible rotating zero-block regularization is an
$O(2)$-equivariant finite-dimensional family of Lagrangian boundary projectors
$P^{(0,j),\operatorname{APS,reg}}_{t,p}$, for $t\in\{0,T\}$,
on the zero Fourier block satisfying the following conditions:
\begin{enumerate}
\item outside the regularization interval $|A_p|<\delta$, the projectors agree with the
ordinary completed APS projectors;
\item their kernels form a continuous family of Lagrangian domain subspaces for the total boundary
Green form;
\item the associated zero-energy transfer determinant has a simple zero at $A_p=0$ and no other
zero in the punctured regularization interval;
\item the crossing form at $A_p=0$ is nondegenerate. In the blockwise convention used below, it is
also diagonal with respect to the real decomposition
$(\mathbf 1\oplus\det)_j = \mathbf 1_j\oplus \det_j$.
\end{enumerate}
The zero-block regularization is part of the fixed admissible APS convention used in the
spectral-flow formula.
\end{definition}

\begin{lemma}[Explicit zero-block model and sign convention]
\label{lem:hrank_admissible_zero_block_regularization}
In the warped-cylinder model, admissible rotating zero-block regularizations in the sense of
Definition~\ref{def:admissible_zero_block_regularization} exist. Moreover, they may be chosen
to be diagonal with respect to the reflection splitting $
(\mathbf 1\oplus\det)_j=\mathbf 1_j\oplus\det_j$.
For such a diagonal choice, the local crossing signs at a zero-mode crossing $p_*$, where
$A_{p_*}=0$, are well-defined numbers
$\epsilon_{\mathbf{1},j,p_*}$ and $\epsilon_{\det,j,p_*}\in\{\pm1\}$.
With the zero-block labeling used in the sequel, the $+$-branch
$m_{0,j,+}(p)=\mu_jA_p$ contributes to the $\mathbf 1_j$-line and the
$-$-branch $m_{0,j,-}(p)=-\mu_jA_p$ contributes to the $\det_j$-line. Thus, for an
affine path
$A_p=A+\nu p$, with $\nu\neq0$,
the signs are
\begin{equation}
\epsilon_{\mathbf{1},j,p_*}
=
\operatorname{sign}\!\left(F'(0)\nu\right),
\qquad
\epsilon_{\det,j,p_*}
=
-\operatorname{sign}\!\left(F'(0)\nu\right).
\end{equation}
In particular, if the scalar regularization is oriented so that $F'(0)>0$, then
$\epsilon_{\mathbf{1},j,p_*}=\operatorname{sign}(\nu)$ and
$\epsilon_{\det,j,p_*}=-\operatorname{sign}(\nu)$.
\end{lemma}

\begin{proof}
Fix a rotating two-plane $j$. At the zero Fourier mode, the two complex weight
blocks are self-paired by reflection. After regrouping the conjugate and
reflection-paired complex blocks into real representation classes, the corresponding
real $O(2)$-representation is
$(\mathbf 1\oplus\det)_j$.
Thus the finite-dimensional zero Fourier boundary space splits into the
reflection-even and reflection-odd coefficient directions, carrying
$\mathbf 1_j$ and $\det_j$, respectively.

We choose the zero-block boundary path in this reflection-diagonal chart. On the
$\mathbf 1_j$-direction we use the scalar regularized APS construction of
\autoref{subsec:hrank_scalar_regularized_APS} with local scalar parameter
$m_{0,j,+}(p)=\mu_jA_p$,
and on the $\det_j$-direction we use the same scalar construction with local scalar
parameter
$m_{0,j,-}(p)=-\mu_jA_p$.
Since these two directions are reflection eigenspaces and rotations act trivially on
the zero Fourier coefficient, their direct sum is $O(2)$-equivariant.

Each scalar piece is Lagrangian for its scalar boundary Green form, and the
$\mathbf 1_j$- and $\det_j$-directions are orthogonal under the reflection
splitting. Hence their direct sum is Lagrangian for the total zero-block
boundary Green form. The scalar construction agrees with the ordinary APS choice
outside the regularization interval, so the zero-block construction also agrees
with the completed APS convention outside
$|A_p|<\delta$.

It remains to check admissibility. The zero-energy transfer equation in each
coefficient direction is exactly the scalar zero-energy transfer problem from
the nonzero-block calculation, with $m_{0,j,+}(p)$ or $m_{0,j,-}(p)$ replacing
the scalar parameter $m$. Therefore, the scalar determinant computation from
\autoref{subsec:hrank_scalar_transfer_admissibility} applies verbatim to each
coefficient direction. Thus the determinant $F(m)$ satisfies
$F(0)=0$ and $F'(0) = \frac{2\alpha'(0)}{\delta}-L_f$.
Choosing $\delta$ and the angle function $\alpha$ so that
$\frac{2\alpha'(0)}{\delta}\neq L_f$
makes the zero at $m=0$ simple. After replacing the regularization interval if
necessary to smaller positive interval, there are no other zeroes in the punctured interval. This proves the
existence of admissible diagonal zero-block regularizations.

The local crossing sign in each direction is the sign of the derivative of the scalar
transfer determinant along the corresponding scalar branch:
\begin{equation}
\epsilon_{\mathbf{1},j,p_*}
=
\operatorname{sign}
\left(
F'(0)
\left.
\frac{d}{dp}
\right|_{p=p_*}
m_{0,j,+}(p)
\right),
\qquad
\epsilon_{\det,j,p_*}
=
\operatorname{sign}
\left(
F'(0)
\left.
\frac{d}{dp}
\right|_{p=p_*}
m_{0,j,-}(p)
\right).
\end{equation}
For an affine path $A_p=A+\nu p$, we have
\begin{equation}
\left.
\frac{d}{dp}
\right|_{p=p_*}
m_{0,j,+}(p)=\mu_j\nu,
\qquad
\left.
\frac{d}{dp}
\right|_{p=p_*}
m_{0,j,-}(p)=-\mu_j\nu .
\end{equation}
Since $\mu_j>0$, this gives
\begin{equation}
\epsilon_{\mathbf{1},j,p_*}
=
\operatorname{sign}\!\left(F'(0)\nu\right),
\qquad
\epsilon_{\det,j,p_*}
=
-\operatorname{sign}\!\left(F'(0)\nu\right).
\end{equation}
This is the sign convention used in the affine examples. For any other admissible
zero-block convention, the spectral-flow formula remains the same after replacing
these signs by the actual crossing-form signs.
\end{proof}

\begin{remark}[Existence of admissible zero-block choices in the model]
\label{rem:hrank_zero_block_existence}
On the zero Fourier block of one rotating two-plane, the boundary data split into the reflection-even and reflection-odd
real lines corresponding to $\mathbf 1_j$ and $\det_j$.  One may use Lagrangian boundary paths on these two lines by the same scalar angle construction used in \eqref{eq:sep_alpha_conditions}, with the local parameter $A_p$ replacing $m$. For a generic choice of the interval size and of the two endpoint angles, the associated zero-energy transfer determinant has a simple zero at $A_p=0$ and no other zero in the
punctured interval.  If the two coefficient signs are required to be separated, the even
and odd lines are chosen independently inside this diagonal chart.  Thus, the
zero-block convention is explicit up to the same admissibility requirement already
imposed for the nonzero scalar blocks.
\end{remark}

For the blockwise formula below, we use the signs defined by
Lemma~\ref{lem:hrank_admissible_zero_block_regularization}. Thus, at a zero-mode
crossing $p_*$, the diagonal zero-block crossing form determines two signs
\begin{equation}\label{eq:sep_zero_block_signs_j} \epsilon_{\mathbf{1},j,p_*}\in\{\pm1\}, \qquad \epsilon_{\det,j,p_*}\in\{\pm1\}.
\end{equation}
If several zero blocks are treated together, the simultaneous-crossing
convention of \autoref{subsec:hrank_standing_assumptions} applies: one uses
the signature of the total crossing form on the relevant multiplicity space
rather than assigning separate signs to individual $j$-blocks.

\begin{remark}[Zero-mode multiplicity]
\label{rem:sep_simultaneous_zero_mode_crossings}
If $r>1$, then all rotating zero-mode blocks cross at the same value
$A_p=0$. Thus, the zero-mode crossing may have multiplicity $2r$. This is the example of the simultaneous-crossing convention from
\autoref{subsec:hrank_standing_assumptions}.
\end{remark}

\subsection{Neutral blocks}
\label{subsec:hrank_neutral_blocks_regularized}

The neutral blocks have effective parameter
$m_{k,a,0}(p)=k$.
Hence they do not produce moving crossings.  If $k\neq0$, the ordinary APS
condition is fixed throughout the path.  In the periodic $k=0$ sector, the
boundary operator has persistent kernel, and we use the fixed neutral
completion of Definition~\ref{def:hrank_neutral_sector_convention}.

Thus the neutral sector contributes no moving crossing term to the
$RO(O(2))$-valued spectral-flow formula.  Its only role is to specify the
completed APS domain in the stationary zero sector.

\subsection{The global regularized APS domain}
\label{subsec:hrank_global_reg_APS_domain}

We now assemble the blockwise projections into a global regularized APS condition.
Let $P_{t,p}^{(n),\text{APS,reg}}$ be the orthogonal projection onto the direct sum
of the following subspaces on all non-neutral-zero sectors: the regularized positive spaces
$E_{t,k,p}^{+,\text{reg}}$ for all nonzero reflection-paired rotating blocks
$k\geq1$, the regularized rotating zero-mode positive spaces determined by the
projectors
$P_{t,p}^{(0,j),\APSreg}$, for $j=1,\ldots,r$,
and the ordinary APS positive spaces on non-crossing rotating and nonzero neutral blocks.
The neutral zero sector is completed separately by the global condition
\eqref{eq:neutral_global_completion}.

Equivalently, the rotating zero-mode part of $P_{t,p}^{(n),\text{APS,reg}}$ is the direct
sum over $j$ of the previously chosen $O(2)$-equivariant regularized APS projectors
$P_{t,p}^{(0,j),\APSreg}$.

The corresponding regularized APS domain is obtained by imposing the assembled
regularized APS projections on all rotating and nonzero neutral blocks, together with the
global neutral zero-sector condition \eqref{eq:neutral_global_completion}:
\begin{equation}\label{eq:sep_global_reg_APS_domain}
\begin{split}
\Dom\left(D_{p,\text{APS,reg}}^{E_{\mathbb C}^{(n)}}\right)
=
\Bigl\{
\Psi\in H^1\left(M;S\otimes E_{\mathbb C}^{(n)}\right):
&\ P_{0,p}^{(n),\text{APS,reg}}\left(\Psi|_{Y_0}\right)=0,\\
&\ P_{T,p}^{(n),\text{APS,reg}}\left(\Psi|_{Y_T}\right)=0
\text{ on all non-neutral-zero sectors},\\
&\ \widehat\gamma_T\Psi=-\widehat\gamma_0\Psi
\text{ on the neutral $k=0$ sector}
\Bigr\}.
\end{split}
\end{equation}
We write
\begin{equation}\label{eq:sep_global_reg_APS_realization}
D_{p,\text{APS,reg}}^{E_{\mathbb C}^{(n)}}
=
D_p^{E_{\mathbb C}^{(n)}}
\big|_{\Dom(D_{p,\text{APS,reg}}^{E_{\mathbb C}^{(n)}})}.
\end{equation}

Here $\mathcal R_t$ denotes the boundary action on $Y_t$ induced by the lifted reflection $\mathcal U_{\mathbf r}^{E_{\mathbb C}^{(n)}}$.

The assembled
projection $P_{t,p}^{(n),\operatorname{APS,reg}}$ commutes with the lifted boundary
reflection:
\begin{equation}
P_{t,p}^{(n),\operatorname{APS,reg}}\mathcal R_t
=
\mathcal R_t P_{t,p}^{(n),\operatorname{APS,reg}},
\qquad
t\in\{0,T\}.
\end{equation}
The global neutral zero-sector condition is also preserved by the lifted reflection. Hence, the full regularized APS boundary condition is reflection invariant. The boundary condition depends continuously on $p$, and it is smooth inside each regularization interval. Outside the finite-dimensional
regularization intervals, it agrees with the ordinary APS projection, together with the fixed
neutral-sector convention.

The following proposition is the analytic framework for the finite-dimensional crossing calculation. In this model, the regularized APS projection differs from the ordinary completed APS projection only on finitely many Fourier blocks inside
small regularization intervals.  Such a finite-rank Lagrangian change is a smoothing
perturbation of the APS projector and does not alter the principal APS elliptic
condition.  The Lagrangian condition gives self-adjointness through the Green form,
while ellipticity and Fredholmness are inherited from the APS realization.  This is the
standard elliptic-boundary framework for Dirac-type operators with APS and Lagrangian
boundary conditions, as in the treatments of Booss-Wojciechowski and B\"ar-Ballmann
\textup{\cite{BoossWojciechowski1993,BaerBallmann2012}}.

\begin{proposition}[Well-posed higher-rank regularized APS family]
\label{prop:sep_wellposed_higher_rank_reg_APS}
Consider the reflection-compatible higher-rank twisting data
$(E_{\mathbb R}^{(n)},J_n,C_n)$, the scalar regularized APS convention above, the
admissible rotating zero-block convention, and the fixed neutral-sector convention.
Then the assembled boundary data define a self-adjoint elliptic APS-type boundary
condition, with a finite-dimensional global completion on the neutral zero sector, for
every $p\in[0,1]$.  The assembled boundary condition is also preserved by the lifted
reflection. Consequently,
for each $p$, the operator
$D_{p,\text{APS,reg}}^{E_{\mathbb C}^{(n)}}$
is self-adjoint and Fredholm.  Moreover, the family
\begin{equation}\label{eq:sep_gap_continuous_family} p\longmapsto D_{p,\text{APS,reg}}^{E_{\mathbb C}^{(n)}}
\end{equation}
is continuous in the gap topology.  Under the reflection-compatible neutral convention,
it is pointwise $O(2)$-equivariant.
\end{proposition}

\begin{proof}
In a crossing block, the ordinary APS domain line is replaced by the regularized domain line $\ell_t^{\text{dom}}(m)$, and in the rotating zero block by the chosen admissible
zero-block Lagrangian line.  In the neutral zero sector, the finite-dimensional global
condition \eqref{eq:neutral_global_completion} completes the stationary boundary-kernel
part.  These replacements occur only in finitely many Fourier blocks.  Indeed, the image of $A_p$ is bounded on the compact interval $[0,1]$, and for
fixed weights $\mu_j$ only finitely many Fourier labels $k$ can satisfy
$|k\pm\mu_jA_p|<\delta$ for some $p\in[0,1]$.  Hence the changes are finite-rank,
and therefore smoothing, perturbations of the completed APS projector on the boundary.

The principal APS boundary condition is unchanged by a smoothing finite-rank
perturbation.  Thus ellipticity is inherited from the completed APS boundary problem;
no new principal boundary symbol is introduced by the finite-dimensional regularization.
The new domain subspaces are Lagrangian for the boundary Green form, and the neutral
zero-sector graph \eqref{eq:neutral_global_completion} is Lagrangian for the total Green
form on the two neutral endpoint spaces.  Therefore the boundary term in Green's formula
vanishes on the domain and the adjoint domain is the same.  Equivalently, the boundary
relation is maximal isotropic in the boundary symplectic space.  The regularized
realization is therefore self-adjoint.  Since the boundary condition is elliptic and
differs from the Fredholm APS realization by a finite-rank boundary perturbation, the
corresponding operator is Fredholm.

The coefficients of $D_p^{E_{\mathbb C}^{(n)}}$ are continuously differentiable functions of $p$.
The only $p$-dependent boundary projections are the finite-dimensional regularized
block projections, and these vary continuously in $p$ by construction.  Standard
continuity of elliptic realizations under continuous variation of finite-dimensional
Lagrangian boundary data gives gap-continuity of
\eqref{eq:sep_gap_continuous_family}.

Finally, rotation equivariance follows from the $\theta$-independence of the metric,
the connection, and the blockwise construction.  Reflection equivariance on the rotating
sector follows from $C_n^2=I$ and $C_nJ_nC_n^{-1}=-J_n$, together with the fact that the
regularized projection was defined by fixing representative blocks and adding their
reflected partners.  On the neutral zero sector, the lifted reflection acts by the same
operator on both endpoint trace spaces, so it preserves the graph
\eqref{eq:neutral_global_completion}.  Hence the full domain is preserved by the
$O(2)$-action, and the family is pointwise $O(2)$-equivariant.
\end{proof}

\subsection{Crossing sets for the higher-rank family}
\label{subsec:hrank_crossing_sets}

The moving crossings of the higher-rank family occur only in the rotating blocks.
For $k\geq1$ and $j\in\mathcal J$, define the nonzero-mode crossing set by
\begin{equation}\label{eq:sep_crossing_set_nonzero_higher_rank}
\mathcal C_{k,j}
=
\left\{
p_*\in(0,1):
A_{p_*}=\frac{k}{\mu_j}
\text{ or }
A_{p_*}=-\frac{k}{\mu_j}
\right\}.
\end{equation}
For the rotating zero mode of the $j$-th two-plane, define
\begin{equation}\label{eq:sep_crossing_set_zero_higher_rank} \mathcal C_{0,j} = \left\{ p_*\in(0,1): A_{p_*}=0 \right\}.
\end{equation}
The neutral sector has no moving crossing set. Its contribution is fixed by
Definition~\ref{def:hrank_neutral_sector_convention} and is zero in the spectral-flow
calculation.

If $p_*\in\mathcal C_{k,j}$, then one of the scalar parameters
$m_{k,j,+}(p)$, $m_{k,j,-}(p)$, $m_{-k,j,+}(p)$, or $m_{-k,j,-}(p)$ vanishes.
At such a crossing, the two reflected zero blocks together carry the real
$O(2)$-representation $\rho_k$. If $p_*\in\mathcal C_{0,j}$, then the rotating
zero-mode crossing space carries a copy of $\mathbf 1\oplus\det$.

Thus the higher-rank regularized APS construction reduces the spectral-flow computation
to a finite sum of scalar crossing contributions, one for each rotating weight
$\mu_j$ and each Fourier block $k$.

\section{The higher-rank \texorpdfstring{$RO(O(2))$}{RO(O(2))}-valued spectral-flow formula}
\label{sec:hrank_ROO2_spectral_flow_formula}

We now compute the representation-valued spectral flow of the regularized APS family
constructed in~\autoref{sec:hrank_separated_blocks_regularized_APS}. The formula is
obtained by summing the local crossing contributions of the rotating weighted blocks;
the stationary neutral sector is fixed as in
Definition~\ref{def:hrank_neutral_sector_convention}.

Throughout this section, we work in the periodic spin sector unless explicitly stated otherwise. Thus
\begin{equation}\label{eq:hrank_formula_periodic_lattice} \mathcal K=\mathbb Z, \qquad \mathcal K_+=\{k\in\mathbb Z:\ k\geq1\}.
\end{equation}
We also keep the normal form notation
\begin{equation}\label{eq:hrank_formula_Jn_normal_form_recall} J_n=\mu_1J\oplus\cdots\oplus\mu_{r}J\oplus0_z, \qquad \mu_j>0, \qquad 2r+z=n.
\end{equation}

\subsection{Local signs and regular crossings}
\label{subsec:hrank_local_signs_regular_crossings}

Let
\begin{equation}\label{eq:hrank_formula_family} p\longmapsto D_{p,\text{APS,reg}}^{E_{\mathbb C}^{(n)}}
\end{equation}
be the higher-rank regularized APS family. By
Proposition~\ref{prop:sep_wellposed_higher_rank_reg_APS}, this is a continuous path of
self-adjoint Fredholm operators. Under the reflection-compatible neutral convention,
it is pointwise $O(2)$-equivariant.

A crossing time is a parameter $p_*\in(0,1)$ for which
\begin{equation}\label{eq:hrank_crossing_time_def} \ker D_{p_*,\text{APS,reg}}^{E_{\mathbb C}^{(n)}}\neq0.
\end{equation}
Let the path be continuously differentiable in the relevant finite-dimensional crossing charts. The
crossing form at $p_*$ is the quadratic form on
$\ker D_{p_*,\text{APS,reg}}^{E_{\mathbb C}^{(n)}}$ obtained by differentiating
the locally trivialized family, as in \eqref{eq:hrank_operator_crossing_form}. The
crossing is called regular if this crossing form is nondegenerate.

In a separated scalar block with effective parameter $m=m_{k,j,\sigma}(p)$, the
local scalar detector is the transfer determinant $F(m)$. Under the admissibility
condition from Definition~\ref{def:hrank_admissible_scalar_crossing}, $F(0)=0$
and $F'(0)\neq0$. After the local graph trivialization of the regularized APS
domains, the scalar crossing form is represented by the derivative of the transfer
determinant along the path:
\begin{equation}\label{eq:hrank_scalar_crossing_form_via_F}
\left.
\frac{d}{dp}
\right|_{p=p_*}
F(m_{k,j,\sigma}(p))
=
F'(0)\,m'_{k,j,\sigma}(p_*).
\end{equation}
Thus the scalar crossing sign is the sign of the actual $m$-motion through zero:
\begin{equation}\label{eq:hrank_scalar_epsilon_definition}
\epsilon_{k,j,\sigma,p_*}^{\text{sc}}
=
\operatorname{sign}
\left(
F'(0)\,m'_{k,j,\sigma}(p_*)
\right).
\end{equation}
This convention is the one directly induced by the operator crossing form
\eqref{eq:hrank_operator_crossing_form}; it is not an auxiliary sign convention
coming from an $A$-coordinate.

For a nonzero Fourier block $\{k,-k\}$, with $k\geq1$, and a rotating two-plane
$j$, the displayed sign $\epsilon_{k,j,p_*}$ in the $RO(O(2))$-valued
formula is defined by the scalar branch that vanishes at $p_*$. More explicitly,
if
\begin{equation}\label{eq:hrank_epsilon_nonzero_plus_branch} A_{p_*}=-\frac{k}{\mu_j},
\end{equation}
then the vanishing representative branch is
$m_{k,j,+}(p)=k+\mu_jA_p$, and we define
\begin{equation}\label{eq:hrank_nonzero_epsilon_plus_definition} \epsilon_{k,j,p_*} = \operatorname{sign} \left( F'(0)\,m'_{k,j,+}(p_*) \right).
\end{equation}
If
\begin{equation}\label{eq:hrank_epsilon_nonzero_minus_branch} A_{p_*}=\frac{k}{\mu_j},
\end{equation}
then the vanishing representative branch is
$m_{k,j,-}(p)=k-\mu_jA_p$, and we define
\begin{equation}\label{eq:hrank_nonzero_epsilon_minus_definition} \epsilon_{k,j,p_*} = \operatorname{sign} \left( F'(0)\,m'_{k,j,-}(p_*) \right).
\end{equation}
Equivalently, $\epsilon_{k,j,p_*}$ is the sign of the crossing form on the
corresponding reflection-paired $O(2)$-crossing block, computed through the
representative scalar branch whose effective parameter actually crosses zero.

\begin{remark}[Representative branch, reflected coordinate, and chart consistency]
\label{rem:representative_branch_reflected_coordinate}
For $k\geq1$, the coefficient $\epsilon_{k,j,p_*}$ is attached to the
whole reflection-paired $O(2)$-crossing block, not to two independently counted
scalar crossings.  For example, if
$A_{p_*}=-\frac{k}{\mu_j}$,
we use the representative branch
$q(p)=m_{k,j,+}(p)=k+\mu_jA_p$.
The reflected partner is $(-k,j,-)$.  Written in its raw scalar coordinate, its
parameter is
$\widetilde q(p)=m_{-k,j,-}(p)=-k-\mu_jA_p=-q(p)$.
Thus, the raw reflected coordinate has the opposite $m$-slope.  This does not
mean that the reflected partner contributes an independent opposite signed crossing.
The crossing sign is the sign of the derivative of the scalar transfer determinant
in a fixed local scalar chart:
\begin{equation}
\operatorname{sign}
\left(
\left.
\frac{d}{dp}
\right|_{p=p_*}
F(q(p))
\right).
\end{equation}

If the reflected coordinate $\widetilde q=-q$ is used instead, then the transfer determinant must be written in the reflected chart as
$\widetilde F(\widetilde q)=F(-\widetilde q)$. Here $\widetilde{F}$ denotes the same scalar transfer determinant expressed as
a function of the reflected coordinate $\widetilde q$; it is the pullback of $F$ under the coordinate change $q=-\widetilde q$. Consequently, $\widetilde F'(0)=-F'(0)$ and $\widetilde q'(p_*)=-q'(p_*)$, and hence $\widetilde F'(0)\widetilde q'(p_*)=F'(0)q'(p_*)$.
So, reversing the scalar coordinate does not reverse the invariant crossing-form sign when the transfer-determinant chart is transformed consistently.

Equivalently, if $u$ is a crossing vector in the representative branch and $\mathcal U_{\mathbf r}$ denotes the lifted reflection, reflection equivariance of the locally trivialized family gives
\begin{equation}
\Gamma_{p_*}(\mathcal U_{\mathbf r}u,\mathcal U_{\mathbf r}u)
=
\Gamma_{p_*}(u,u).
\end{equation}
Thus, the reflected half has the same crossing-form sign as the representative half
when the crossing form is viewed on the real $O(2)$-paired block.  The pair
$(k,j,+)\oplus(-k,j,-)$ therefore contributes a single signed copy of
$\rho_k$, with coefficient computed from the chosen representative branch.
\end{remark}

For the rotating zero mode of the $j$-th two-plane, the crossing eigenspace
decomposes as a copy of $\mathbf 1\oplus\det$. The chosen zero-block
regularization fixes which scalar zero branch corresponds to the $\mathbf 1$-line
and which corresponds to the $\det$-line. With that convention, we define
\begin{equation}\label{eq:hrank_zero_local_signs_def}
\epsilon_{\mathbf{1},j,p_*}
=
\operatorname{sign}
\left(
F'(0)\,m'_{0,j,+}(p_*)
\right),
\qquad
\epsilon_{\det,j,p_*}
=
\operatorname{sign}
\left(
F'(0)\,m'_{0,j,-}(p_*)
\right),
\end{equation}
where
\begin{equation}\label{eq:hrank_zero_branch_parameters_for_epsilon} m_{0,j,+}(p)=\mu_jA_p, \qquad m_{0,j,-}(p)=-\mu_jA_p.
\end{equation}
If a different admissible zero-block convention swaps the two reflection lines, then
the labels $\mathbf 1$ and $\det$ in
\eqref{eq:hrank_zero_local_signs_def} must be swapped accordingly. In all cases,
the signs are the actual crossing-form signs on the corresponding one-dimensional
$O(2)$-representation blocks.

\begin{remark}[What we mean by a block sign]
\label{rem:hrank_block_sign_meaning}
The notation $\epsilon_{k,j,p_*}$ does not mean that the crossing kernel is
one-dimensional. It means that, on the relevant $O(2)$-irreducible crossing block,
the regular crossing form has the coefficient sign defined above. 
\eqref{eq:hrank_scalar_crossing_form_via_F} shows that the coefficient sign is computed from the actual scalar parameter $m$ crossing zero. If several blocks cross simultaneously, we either treat the total crossing form directly, or we impose the
genericity condition that the crossing blocks are separated. In either case the formula below is a sum of the corresponding representation-valued local contributions.
\end{remark}

\begin{remark}[Sign normalization]
\label{rem:hrank_sign_normalization}
The signs $\epsilon_{k,j,p_*}$, $\epsilon_{\mathbf{1},j,p_*}$, and
$\epsilon_{\det,j,p_*}$ are the signs of the regular crossing form on the
corresponding $O(2)$-representation crossing blocks. In the scalar block
description, this is the same as the sign of $F'(0)m'(p_*)$, where $m(p)$ is
the effective scalar parameter of the branch that crosses zero.

For an affine path $A_p=A+\nu p$, the raw scalar slopes are
\begin{equation}\label{eq:hrank_affine_raw_m_slopes_for_signs} m'_{k,j,+}(p_*)=\mu_j\nu, \qquad m'_{k,j,-}(p_*)=-\mu_j\nu.
\end{equation}
Thus, at a crossing with $A_{p_*}=-k/\mu_j$, one has
\begin{equation}\label{eq:hrank_affine_plus_sign_from_m} \epsilon_{k,j,p_*} = \operatorname{sign}\left(F'(0)\mu_j\nu\right),
\end{equation}
whereas at a crossing with $A_{p_*}=k/\mu_j$, one has
\begin{equation}\label{eq:hrank_affine_minus_sign_from_m} \epsilon_{k,j,p_*} = \operatorname{sign}\left(-F'(0)\mu_j\nu\right).
\end{equation}
In particular, if the scalar regularization is oriented so that $F'(0)>0$, then
the $m_{k,j,+}$-branch contributes with $\operatorname{sign}(\nu)$, while
the $m_{k,j,-}$-branch contributes with $-\operatorname{sign}(\nu)$. This
is same as the convention coming from the operator crossing form
\eqref{eq:hrank_operator_crossing_form}.
\end{remark}

\subsection{Boundary-to-operator crossing correspondence in higher rank}
\label{subsec:hrank_boundary_operator_correspondence}

The scalar regularized APS boundary family from
\autoref{subsec:hrank_scalar_regularized_APS} was chosen so that a scalar block
crosses exactly when its effective parameter $m$ vanishes. We state the consequence
for the higher-rank family.

\begin{lemma}[Higher-rank boundary-to-operator crossing correspondence]
\label{lem:hrank_boundary_to_operator_crossing}
Using the higher-rank regularized APS convention of
\autoref{sec:hrank_separated_blocks_regularized_APS}. In a rotating block
$(k,j,\sigma)$, the regularized scalar boundary-value problem has a kernel when
\begin{equation}\label{eq:hrank_scalar_kernel_condition} m_{k,j,\sigma}(p)=0.
\end{equation}
If the scalar zero is simple, namely if
\begin{equation}\label{eq:hrank_simple_scalar_zero} \frac{d}{dp}m_{k,j,\sigma}(p_*)\neq0,
\end{equation}
then the corresponding operator crossing is regular.

Consequently, for $k\geq1$, the nonzero reflection-paired block belonging to the $j$-th rotating two-plane crosses at the points
\begin{equation}\label{eq:hrank_nonzero_crossing_equation} A_{p_*}=\frac{k}{\mu_j} \qquad\text{or}\qquad A_{p_*}=-\frac{k}{\mu_j}.
\end{equation}
At such a crossing, the crossing eigenspace carries the real $O(2)$-representation
$\rho_k$. For the rotating zero Fourier mode of the $j$-th two-plane, the crossing
equation is
\begin{equation}\label{eq:hrank_zero_crossing_equation} A_{p_*}=0,
\end{equation}
and the corresponding crossing eigenspace carries a copy of $\mathbf 1\oplus\det$.
The neutral sector contributes no moving crossings under
Definition~\ref{def:hrank_neutral_sector_convention}.
\end{lemma}

\begin{proof}
By Corollary~\ref{lem:hrank_scalar_boundary_to_operator_crossing}, the scalar regularized
boundary-value problem in a chosen representative block has a kernel when the representative scalar parameter $m$ vanishes, and the crossing is regular whenever that scalar parameter vanishes but its $p$-derivative must not vanish. For the reflected partner, the same conclusion follows by applying the lifted reflection
$\mathcal U_{\mathbf r}^{E_{\mathbb C}^{(n)}}$, since the reflected boundary condition was
defined as the reflected image of the representative regularization. Thus, the
representative scalar calculation gives the crossing condition, and equivariance gives
the paired reflected crossing. The neutral blocks are treated separately, since their
parameters are independent of $p$.

In a rotating block, the effective parameters are
$m_{k,j,+}(p)=k+\mu_jA_p$ and $m_{k,j,-}(p)=k-\mu_jA_p$. Thus a scalar kernel occurs
exactly when one of these quantities vanishes. For the paired Fourier block
$\{k,-k\}$, the four rotating parameters are
\begin{equation}\label{eq:hrank_four_block_parameters_j}
m_{k,j,+}(p)=k+\mu_jA_p,
\qquad
m_{k,j,-}(p)=k-\mu_jA_p,
\qquad
m_{-k,j,+}(p)=-k+\mu_jA_p,
\qquad
m_{-k,j,-}(p)=-k-\mu_jA_p.
\end{equation}
A zero of one of these parameters occurs exactly when
$A_{p_*}=k/\mu_j$ or $A_{p_*}=-k/\mu_j$. At such a crossing, the two zero blocks
are exchanged by the lifted reflection. Therefore their real crossing space is the
reflection-paired $O(2)$-representation $\rho_k$.

For $k=0$, the moving rotating parameters are $\mu_jA_p$ and $-\mu_jA_p$. Hence
the rotating zero-mode crossing equation is $A_{p_*}=0$. The reflection action on
the corresponding real crossing space has one even line and one odd line, so the
$O(2)$-representation is $\mathbf 1\oplus\det$. The neutral parameters are
$m_{k,a,0}(p)=k$, independent of $p$. Therefore they do not produce moving
crossings. The neutral $k=0$ sector is handled by the fixed convention from
Definition~\ref{def:hrank_neutral_sector_convention}, and hence contributes no moving
spectral-flow term.
\end{proof}

\subsection{Block decomposition of the higher-rank family}
\label{subsec:hrank_block_decomposition}

We next organize the higher-rank family into $O(2)$-invariant blocks. The $O(2)$
representation labels come from the angular Fourier modes. The higher-rank fiber only
adds multiplicities and weights.

\begin{lemma}[Periodic higher-rank block structure]
\label{lem:hrank_periodic_block_structure}
We work in the periodic spin sector. Then the higher-rank regularized APS family decomposes as an orthogonal direct sum of $O(2)$-invariant block paths:
\begin{equation}\label{eq:hrank_operator_block_decomposition}
D_{p,\text{APS,reg}}^{E_{\mathbb C}^{(n)}}
\cong
D_{p,\text{APS,reg}}^{\text{rot},(0)}
\oplus
D_{p,\text{APS,reg}}^{\text{neu}}
\oplus
\bigoplus_{k\in\mathcal K_+}
D_{p,\text{APS,reg}}^{\text{rot},\{k,-k\}} .
\end{equation}
Here $D_{p,\text{APS,reg}}^{\text{rot},\{k,-k\}}$ is the rotating part of the
reflection-paired nonzero Fourier block, $D_{p,\text{APS,reg}}^{\text{rot},(0)}$
is the rotating zero Fourier block, and $D_{p,\text{APS,reg}}^{\text{neu}}$ is the
neutral sector.

In other words, for each $k\geq1$ and each rotating block $j$, a regular crossing
in the $j$-th block of $D_{p,\text{APS,reg}}^{\text{rot},\{k,-k\}}$ carries
the real $O(2)$-representation $\rho_k$. For each $j$, a regular crossing in the rotating
zero Fourier block carries the real $O(2)$-representation $\mathbf 1\oplus\det$. The neutral
sector contributes no moving crossing term.
\end{lemma}

\begin{proof}
The coefficients of the bulk operator and of the regularized APS boundary
conditions are invariant under rotations of the circle.  Equivalently, the
bulk operator and the boundary projectors commute with the $S^1$-action.
Hence the family preserves the Fourier-mode decomposition.  The reflection
then pairs the $k$- and $-k$-blocks, so the natural $O(2)$-invariant
blocks are the zero-mode block and the paired blocks indexed by
$\{k,-k\}$ for $k\geq1$.

For $k\geq1$, the base reflection sends the $k$-th Fourier sector to the $(-k)$-th
Fourier sector. Therefore the direct sum of the two sectors $\{k,-k\}$ is invariant
under the full $O(2)$-action. On the real angular plane spanned by
$\cos(k\theta)$ and $\sin(k\theta)$, rotations act by the standard rotation of
angle $k\varphi$, and reflection acts by changing the sign of the sine direction.
This is precisely the irreducible real $O(2)$-representation $\rho_k$. The fiber
label $j$ does not change this angular $O(2)$-representation; it only labels which rotating
two-plane of the twisting bundle produced the crossing.

For $k=0$, the Fourier mode is fixed by reflection. In each rotating two-plane, the
two complex blocks with parameters $\mu_jA_p$ and $-\mu_jA_p$ combine into a real
two-dimensional crossing space. The restricted reflection has one even and one odd
real line, so the $O(2)$-representation is $\mathbf 1\oplus\det$. Finally, the neutral sector
is fixed by a $p$-independent convention and contributes no moving crossings.
\end{proof}

\subsection{The higher-rank formula}
\label{subsec:hrank_main_formula}

We now state the representation-valued formula. For $k\geq1$ and
$j\in\{1,\ldots,r\}$, let
\begin{equation}\label{eq:hrank_Ckj_def}
\mathcal C_{k,j}
=
\left\{
p_*\in(0,1):
A_{p_*}=\frac{k}{\mu_j}
\text{ or }
A_{p_*}=-\frac{k}{\mu_j}
\right\}.
\end{equation}
For the rotating zero block of the $j$-th two-plane, let
\begin{equation}\label{eq:hrank_C0j_def} \mathcal C_{0,j} = \left\{ p_*\in(0,1): A_{p_*}=0 \right\}.
\end{equation}

\begin{center}\fbox{\begin{minipage}{0.94\textwidth}\textbf{Hypotheses for Theorem~\ref{thm:hrank_ROO2_spectral_flow}.} We fix one admissible regularized APS convention throughout: the scalar regularized APS boundary family of Definition~\ref{def:hrank_admissible_scalar_crossing}, the rotating zero-block regularization of Definition~\ref{def:admissible_zero_block_regularization}, and the fixed neutral-sector convention of Definition~\ref{def:hrank_neutral_sector_convention}. We also impose endpoint invertibility, regular rotating crossings, and the separated or block-diagonal convention for simultaneous crossings. With these choices, Proposition~\ref{prop:sep_wellposed_higher_rank_reg_APS} supplies the required continuous path of $O(2)$-equivariant self-adjoint Fredholm operators; the theorem then computes its representation-valued spectral flow by crossing forms.
\end{minipage}}
\end{center}

\begin{theorem}[Higher-rank blockwise \texorpdfstring{$RO(O(2))$}{RO(O(2))}-valued spectral flow]
\label{thm:hrank_ROO2_spectral_flow}
In the periodic spin sector, let the endpoints $p=0,1$ be invertible
for the higher-rank regularized APS family. Let the scalar regularized APS boundary family be admissible in the sense of
Definition~\ref{def:hrank_admissible_scalar_crossing}. Let the rotating zero-block regularization be admissible in the sense of
Definition~\ref{def:admissible_zero_block_regularization}, and require all rotating-block crossings to be regular. If simultaneous crossings occur, we require either that they have been separated by the chosen reflection-compatible genericity convention, or that the total crossing form is nondegenerate and block-diagonal with respect to the block decomposition used in \autoref{sec:hrank_separated_blocks_regularized_APS}.
Finally, we impose the fixed neutral-sector convention of
Definition~\ref{def:hrank_neutral_sector_convention}.

Then, in the real-representation convention of
Remark~\ref{rem:representation_ring_convention}, the $RO(O(2))$-valued spectral flow of
the higher-rank regularized APS family is
\begin{equation}\label{eq:hrank_ROO2_formula}
\begin{split}
\operatorname{sf}_{O(2)}
\left(
D_{p,\text{APS,reg}}^{E_{\mathbb C}^{(n)}}
\right)
=
&
\sum_{j=1}^{r}
\sum_{k\in\mathcal K_+}
\sum_{p_*\in\mathcal C_{k,j}}
\epsilon_{k,j,p_*}\,[\rho_k]
\\
&
+
\sum_{j=1}^{r}
\sum_{p_*\in\mathcal C_{0,j}}
\left(
\epsilon_{\mathbf{1},j,p_*}[\mathbf 1]
+
\epsilon_{\det,j,p_*}[\det]
\right).
\end{split}
\end{equation}
The neutral sector contributes no term. Applying the dimension homomorphism
$RO(O(2))\to\mathbb Z$ gives the ordinary spectral-flow formula
\begin{equation}\label{eq:hrank_integer_sf_formula}
\begin{split}
\operatorname{sf}
\left(
D_{p,\text{APS,reg}}^{E_{\mathbb C}^{(n)}}
\right)
=
&
2
\sum_{j=1}^{r}
\sum_{k\in\mathcal K_+}
\sum_{p_*\in\mathcal C_{k,j}}
\epsilon_{k,j,p_*}
\\
&
+
\sum_{j=1}^{r}
\sum_{p_*\in\mathcal C_{0,j}}
\left(
\epsilon_{\mathbf{1},j,p_*}
+
\epsilon_{\det,j,p_*}
\right).
\end{split}
\end{equation}
\end{theorem}

\begin{proof}
By Proposition~\ref{prop:sep_wellposed_higher_rank_reg_APS}, the path
$p\mapsto D_{p,\text{APS,reg}}^{E_{\mathbb C}^{(n)}}$ is a continuous path of
$O(2)$-equivariant self-adjoint Fredholm operators under the reflection-compatible
neutral convention. Since the endpoints are invertible,
the $RO(O(2))$-valued spectral flow is defined.

By Lemma~\ref{lem:hrank_periodic_block_structure}, the family decomposes into
$O(2)$-invariant rotating nonzero blocks, rotating zero blocks, and the fixed neutral
sector. The neutral sector has no moving crossing contribution by the imposed
neutral-sector convention.

For $k\geq1$, a crossing in the $j$-th rotating two-plane occurs at
$p_*\in\mathcal C_{k,j}$, by
Lemma~\ref{lem:hrank_boundary_to_operator_crossing}. After regrouping the conjugate
complex crossing blocks as in Remark~\ref{rem:representation_ring_convention}, the
crossing eigenspace carries the real $O(2)$-representation $\rho_k$, and its local
contribution is therefore $\epsilon_{k,j,p_*}[\rho_k]$. Summing these contributions over all $j$, all
positive Fourier modes $k$, and all crossing points $p_*\in\mathcal C_{k,j}$
gives the first term in \eqref{eq:hrank_ROO2_formula}.

For the rotating zero Fourier block, the $j$-th two-plane crosses only at points
$p_*\in\mathcal C_{0,j}$, by the admissible zero-block regularization fixed in
Definition~\ref{def:admissible_zero_block_regularization}. After applying the
same real-representation regrouping convention, the crossing eigenspace carries the representation
$\mathbf 1\oplus\det$. Since the crossing form may have separate signs on the two
blocks, the local contribution is
$\epsilon_{\mathbf{1},j,p_*}[\mathbf 1]+\epsilon_{\det,j,p_*}[\det]$. Summing over all
$j$ and all $p_*\in\mathcal C_{0,j}$ gives the second term in
\eqref{eq:hrank_ROO2_formula}.

Adding all local representation-valued contributions proves the $RO(O(2))$-valued
formula. The ordinary spectral-flow formula follows by applying the dimension map,
using $\dim\rho_k=2$, $\dim\mathbf 1=1$, and $\dim\det=1$.
\end{proof}

\begin{remark}[Anti-periodic spin sector and the representation group]
The theorem is stated in the periodic spin sector because the displayed
$RO(O(2))$-classes use integral Fourier weights. The scalar crossing calculation in the
anti-periodic sector is recorded separately in
\autoref{subsec:hrank_antiperiodic_sector}.
\end{remark}

\subsection{Affine specialization}
\label{subsec:hrank_affine_specialization}

We now specialize to an affine parameter path. In this case the crossing locations are
explicit, and the signs are obtained from the actual scalar parameters
$m_{k,j,\pm}(p)$ through the crossing form.

The next statement fixes the scalar $m$-slope sign convention used in the affine
formulas below. This convention is the one directly represented by the crossing form in
\eqref{eq:hrank_operator_crossing_form}: in a separated scalar block, after the local
graph trivialization of the regularized APS domain, the crossing-form sign is represented
by the derivative of the scalar transfer determinant along the path.

\begin{proposition}[Affine crossing signs from the scalar crossing form]
\label{prop:hrank_affine_crossing_signs}
Let
\begin{equation}\label{eq:hrank_affine_path} A_p=A+\nu p, \qquad \nu \neq0.
\end{equation}
Use the regularity hypotheses of
Theorem~\ref{thm:hrank_ROO2_spectral_flow}. Let $F(m)$ be the scalar transfer
determinant from Definition~\ref{def:hrank_admissible_scalar_crossing}. Since the
regularized APS family is admissible, $F'(0)\neq0$.

For a separated scalar block with effective parameter $m=m_{k,j,\sigma}(p)$, the local
crossing sign is
\begin{equation}\label{eq:hrank_affine_scalar_crossing_sign}
\epsilon_{k,j,\sigma,p_*}
=
\operatorname{sign}\!\left(F'(0)\,m'_{k,j,\sigma}(p_*)\right).
\end{equation}
Equivalently, the sign records whether the actual scalar parameter $m$ increases or
decreases through $0$, with the common orientation factor $F'(0)$.

For the affine path \eqref{eq:hrank_affine_path}, the nonzero rotating branches satisfy
\begin{equation}\label{eq:hrank_affine_m_derivatives} m'_{k,j,+}(p_*)=\mu_j\nu, \qquad m'_{k,j,-}(p_*)=-\mu_j\nu.
\end{equation}
Hence
\begin{equation}\label{eq:hrank_affine_nonzero_signs}
\epsilon_{k,j,+,p_*}
=
\operatorname{sign}\!\left(F'(0)\mu_j\nu\right),
\qquad
\epsilon_{k,j,-,p_*}
=
\operatorname{sign}\!\left(-F'(0)\mu_j\nu\right).
\end{equation}
Since $\mu_j>0$, this is equivalently
\begin{equation}\label{eq:hrank_affine_nonzero_signs_simplified}
\epsilon_{k,j,+,p_*}
=
\operatorname{sign}\!\left(F'(0)\nu\right),
\qquad
\epsilon_{k,j,-,p_*}
=
-\operatorname{sign}\!\left(F'(0)\nu\right).
\end{equation}

For the rotating zero block, with the convention that the $\mathbf 1$-line is represented
by the $+$-branch $m_{0,j,+}(p)=\mu_jA_p$ and the $\det$-line is represented by the
$-$-branch $m_{0,j,-}(p)=-\mu_jA_p$, one obtains
\begin{equation}\label{eq:hrank_affine_zero_signs}
\epsilon_{\mathbf{1},j,p_*}
=
\operatorname{sign}\!\left(F'(0)\nu\right),
\qquad
\epsilon_{\det,j,p_*}
=
-\operatorname{sign}\!\left(F'(0)\nu\right).
\end{equation}
If the opposite zero-block labeling is used, the two zero-block signs in
\eqref{eq:hrank_affine_zero_signs} are interchanged.
\end{proposition}

\begin{proof}
In a separated scalar block, the crossing form in
\eqref{eq:hrank_operator_crossing_form} is represented, after the local graph
trivialization of the regularized boundary condition, by the derivative of the scalar
transfer determinant along the path:
\begin{equation}\label{eq:hrank_affine_transfer_derivative_along_path}
\left.\frac{d}{dp}\right|_{p=p_*}
F(m_{k,j,\sigma}(p))
=
F'(0)\,m'_{k,j,\sigma}(p_*).
\end{equation}
Therefore the local scalar crossing sign is
$\operatorname{sign}\!\left(F'(0)\,m'_{k,j,\sigma}(p_*)\right)$.

For the affine path $A_p=A+\nu p$, the two nonzero scalar branches are
$m_{k,j,+}(p)=k+\mu_jA_p$ and $m_{k,j,-}(p)=k-\mu_jA_p$.
Differentiating gives
$m'_{k,j,+}(p_*)=\mu_j\nu$ and $m'_{k,j,-}(p_*)=-\mu_j\nu$.
Substitution into
\eqref{eq:hrank_affine_scalar_crossing_sign} gives
\eqref{eq:hrank_affine_nonzero_signs}. Since $\mu_j>0$, this is equivalent to
\eqref{eq:hrank_affine_nonzero_signs_simplified}.

For the rotating zero block, the scalar parameters are
$m_{0,j,+}(p)=\mu_jA_p$ and $m_{0,j,-}(p)=-\mu_jA_p$.
Thus their derivatives are $\mu_j\nu$ and $-\mu_j\nu$. With the stated convention
identifying the $+$-branch with the $\mathbf 1$-line and the $-$-branch with the
$\det$-line, the crossing-form signs are exactly those in
\eqref{eq:hrank_affine_zero_signs}.
\end{proof}

\begin{remark}[Dependence on the scalar APS orientation]
\label{rem:hrank_affine_sign_convention_dependence}
Proposition~\ref{prop:hrank_affine_crossing_signs} uses the scalar $m$-slope convention
coming from the crossing form. The common factor $F'(0)$ records the orientation of the
chosen scalar regularized APS boundary family. If the scalar regularization is chosen so
that $F'(0)>0$, then the $+$-branches have sign $\operatorname{sign}(\nu)$, while
the $-$-branches have sign $-\operatorname{sign}(\nu)$. For another admissible
regularization, the same formulas hold with the corresponding sign of $F'(0)$.
\end{remark}

\begin{corollary}[Affine higher-rank formula]
\label{cor:hrank_affine_formula}
Under the hypotheses of Theorem~\ref{thm:hrank_ROO2_spectral_flow}, suppose in addition
that
\begin{equation}\label{eq:hrank_affine_corollary_path} A_p=A+\nu p, \qquad \nu\neq0.
\end{equation}
Let
\begin{equation}\label{eq:hrank_affine_delta_definition} \delta_{\text{scal}} = \operatorname{sign}\!\left(F'(0)\nu\right).
\end{equation}
In the formula below we keep the zero-block labeling used in
Proposition~\ref{prop:hrank_affine_crossing_signs}, namely the $+$-branch contributes
to $\mathbf 1$ and the $-$-branch contributes to $\det$. Then, with the scalar
$m$-slope crossing convention of Proposition~\ref{prop:hrank_affine_crossing_signs},
\begin{equation}\label{eq:hrank_affine_ROO2_formula}
\begin{split}
\operatorname{sf}_{O(2)}
\left(
D_{p,\text{APS,reg}}^{E_{\mathbb C}^{(n)}}
\right)
=
&
\delta_{\text{scal}}
\sum_{j=1}^{r}
\sum_{k\in\mathcal K_+}
\left(
\#\left\{
p\in(0,1):
A+\nu p=-\frac{k}{\mu_j}
\right\}
-
\#\left\{
p\in(0,1):
A+\nu p=\frac{k}{\mu_j}
\right\}
\right)
[\rho_k]
\\
&
+
\delta_{\text{scal}}
\sum_{j=1}^{r}
\#\left\{
p\in(0,1):
A+\nu p=0
\right\}
\left(
[\mathbf 1]-[\det]
\right).
\end{split}
\end{equation}
After applying the dimension map, one obtains
\begin{equation}\label{eq:hrank_affine_integer_formula}
\begin{split}
\operatorname{sf}
\left(
D_{p,\text{APS,reg}}^{E_{\mathbb C}^{(n)}}
\right)
=
&
2\delta_{\text{scal}}
\sum_{j=1}^{r}
\sum_{k\in\mathcal K_+}
\left(
\#\left\{
p\in(0,1):
A+\nu p=-\frac{k}{\mu_j}
\right\}
-
\#\left\{
p\in(0,1):
A+\nu p=\frac{k}{\mu_j}
\right\}
\right).
\end{split}
\end{equation}
\end{corollary}

\begin{proof}
For the affine path, the nonzero crossing equations are
$k+\mu_j(A+\nu p)=0$ and $k-\mu_j(A+\nu p)=0$.
These are equivalent to
$A+\nu p=-\frac{k}{\mu_j}$ and $A+\nu p=\frac{k}{\mu_j}$,
respectively. By Proposition~\ref{prop:hrank_affine_crossing_signs}, the first type of
crossing has sign $\delta_{\text{scal}}$, while the second type has sign
$-\delta_{\text{scal}}$. This gives the first line of
\eqref{eq:hrank_affine_ROO2_formula}.

For the rotating zero block, the crossing equation is $A+\nu p=0$. With the zero-block
labeling used in Proposition~\ref{prop:hrank_affine_crossing_signs}, the
$\mathbf 1$-line has sign $\delta_{\text{scal}}$, while the $\det$-line has sign
$-\delta_{\text{scal}}$. This gives the second line of
\eqref{eq:hrank_affine_ROO2_formula}. Applying the dimension map gives
\eqref{eq:hrank_affine_integer_formula}, because
$\dim[\rho_k]=2$ and
$\dim([\mathbf 1]-[\det])=1-1=0$.
\end{proof}

\subsection{Identical-block rank-\texorpdfstring{$2r$}{2r} specialization}
\label{subsec:hrank_identical_block_specialization}

The cleanest higher-rank case occurs when all rotating weights are equal to one and
there is no neutral sector. Thus
\begin{equation}\label{eq:hrank_identical_block_J} J_n=J\oplus\cdots\oplus J, \qquad n=2r.
\end{equation}
In this case the higher-rank model is the direct sum of $r$ identical rank-two
orthogonal twists.

\begin{corollary}[Identical-block rank-\texorpdfstring{$2r$}{2r} formula]
\label{cor:hrank_identical_block_formula}
Let $J_n=J\oplus\cdots\oplus J$, with $n=2r$, and impose the hypotheses of
Theorem~\ref{thm:hrank_ROO2_spectral_flow}. Suppose moreover that the same regularized
APS convention is used on each identical rotating block, so that the local crossing
signs are independent of the block label $j$. Then
\begin{equation}\label{eq:hrank_identical_block_ROO2_formula}
\operatorname{sf}_{O(2)}
\left(
D_{p,\text{APS,reg}}^{E_{\mathbb C}^{(2r)}}
\right)
=
r
\sum_{k\in\mathcal K_+}
\sum_{p_*\in\mathcal C_k}
\epsilon_{k,p_*}[\rho_k]
+
r
\sum_{p_*\in\mathcal C_0}
\left(
\epsilon_{1,p_*}[\mathbf 1]
+
\epsilon_{\det,p_*}[\det]
\right),
\end{equation}
where
\begin{equation}\label{eq:hrank_identical_block_crossing_sets}
\mathcal C_k
=
\left\{
p_*\in(0,1):
A_{p_*}=k
\text{ or }
A_{p_*}=-k
\right\},
\qquad
\mathcal C_0
=
\left\{
p_*\in(0,1):
A_{p_*}=0
\right\}.
\end{equation}
Applying the dimension map gives
\begin{equation}\label{eq:hrank_identical_block_integer_formula}
\operatorname{sf}
\left(
D_{p,\text{APS,reg}}^{E_{\mathbb C}^{(2r)}}
\right)
=
2r
\sum_{k\in\mathcal K_+}
\sum_{p_*\in\mathcal C_k}
\epsilon_{k,p_*}
+
r
\sum_{p_*\in\mathcal C_0}
\left(
\epsilon_{1,p_*}
+
\epsilon_{\det,p_*}
\right).
\end{equation}
\end{corollary}

\begin{proof}
If $J_n=J\oplus\cdots\oplus J$, then all weights satisfy $\mu_j=1$, and there are
$r$ identical rotating blocks. Therefore $\mathcal C_{k,j}=\mathcal C_k$ and
$\mathcal C_{0,j}=\mathcal C_0$ for every $j$. Summing the formula of
Theorem~\ref{thm:hrank_ROO2_spectral_flow} over $j=1,\ldots,r$ gives exactly
\eqref{eq:hrank_identical_block_ROO2_formula}. The integer formula follows by applying
the dimension map.
\end{proof}

\subsection{Rank-two and rank-three specializations}
\label{subsec:hrank_rank_two_rank_three_specializations}

When $n=2$, we have $r=1$, $z=0$, and $\mu_1=1$. Therefore
Theorem~\ref{thm:hrank_ROO2_spectral_flow} reduces to the rank-two formula:
\begin{equation}\label{eq:hrank_rank_two_specialization_formula}
\operatorname{sf}_{O(2)}
\left(
D_{p,\text{APS,reg}}^{E_{\mathbb C}^{(2)}}
\right)
=
\sum_{k\in\mathcal K_+}
\sum_{p_*\in\mathcal C_k}
\epsilon_{k,p_*}[\rho_k]
+
\sum_{p_*\in\mathcal C_0}
\left(
\epsilon_{1,p_*}[\mathbf 1]
+
\epsilon_{\det,p_*}[\det]
\right).
\end{equation}
Thus the original rank-two orthogonal-twist model is the first nontrivial case of the
higher-rank theorem.

When $n=3$, the standard model is
\begin{equation}\label{eq:hrank_rank_three_J} J_3=J\oplus0.
\end{equation}
Using the fixed neutral-sector convention, the moving part is the rank-two rotating
part. Thus the rank-three moving formula reduces to
\begin{equation}\label{eq:hrank_rank_three_formula}
\operatorname{sf}_{O(2)}
\left(
D_{p,\text{APS,reg}}^{E_{\mathbb C}^{(3)}}
\right)
=
\operatorname{sf}_{O(2)}
\left(
D_{p,\text{APS,reg}}^{E_{\mathbb C}^{(2)}}
\right)
\end{equation}
as an equality of moving $RO(O(2))$-valued contributions, provided the neutral
$k=0$ sector is made invertible, or otherwise fixed so that it has zero spectral
flow.

\begin{remark}[Odd rank]
\label{rem:hrank_odd_rank}
Every odd-rank skew-symmetric twist has at least one neutral direction. These neutral
directions do not create moving $A_p$-dependent crossings. In the periodic neutral
$k=0$ sector, one has $m_{0,a,0}(p)=0$ for all $p$, so the neutral-sector convention
is included among the hypotheses in the higher-rank theorem.
\end{remark}

\subsection{Anti-periodic sector}
\label{subsec:hrank_antiperiodic_sector}

In the anti-periodic spin sector, the Fourier lattice is
\begin{equation}\label{eq:hrank_antiperiodic_lattice} \mathcal K=\mathbb Z+\frac12 .
\end{equation}
The scalar and operator crossing analysis is unchanged: a rotating block crosses when $k\pm\mu_jA_p=0$. Thus the scalar crossing set is
\begin{equation}\label{eq:hrank_antiperiodic_crossing_set}
\mathcal C_{k,j}^{\text{ap}}
=
\left\{
p_*\in(0,1):
A_{p_*}=\frac{k}{\mu_j}
\text{ or }
A_{p_*}=-\frac{k}{\mu_j}
\right\},
\qquad
k\in\mathbb Z+\frac12,\quad k>0.
\end{equation}
There is no self-paired zero Fourier block, and hence there are no
$[\mathbf 1]$ or $[\det]$ terms.

However, the representation-valued statement is no longer literally an
$RO(O(2))$-valued formula, because half-integer angular weights are not honest
representations of $O(2)$. They are naturally associated with the corresponding
spin double cover of the circle symmetry. Therefore in this paper we treat only the ordinary scalar crossing consequence in the anti-periodic sector:
\begin{equation}\label{eq:hrank_antiperiodic_integer_formula}
\operatorname{sf}
\left(
D_{p,\text{APS,reg}}^{E_{\mathbb C}^{(n)}}
\right)
=
2
\sum_{j=1}^{r}
\sum_{\substack{k\in\mathbb Z+\frac12\\ k>0}}
\sum_{p_*\in\mathcal C_{k,j}^{\text{ap}}}
\epsilon_{k,j,p_*}.
\end{equation}
A representation-valued anti-periodic formula can be written after replacing $RO(O(2))$ with the representation ring of the relevant spin double cover, but we do not perform that extension here.

\section{Rank-three example and the neutral block}
\label{sec:rank_three_neutral_block}

We now spell out the first odd-rank case, emphasizing how the general neutral-sector
convention enters the otherwise rank-two moving calculation.

\subsection{The rank-three twisting data}
\label{subsec:rank_three_twisting_data}

Let
\begin{equation}\label{eq:rthree_bundle} E_{\mathbb R}^{(3)}=M\times\mathbb R^3.
\end{equation}
We take the skew-symmetric twisting matrix to be
\begin{equation}\label{eq:rthree_J3_def}
J_3
=
J\oplus 0
=
\begin{pmatrix}
0&-1&0\\
1&0&0\\
0&0&0
\end{pmatrix}.
\end{equation}
Thus $J_3$ rotates the first two real fiber coordinates and leaves the third real
fiber coordinate fixed. The corresponding family of orthogonal connections is
\begin{equation}\label{eq:rthree_connection} \nabla_{A_p}^{E_{\mathbb R}^{(3)}} = d+A_pJ_3\,d\theta .
\end{equation}

Let
\begin{equation}\label{eq:rthree_C3_def}
C_3
=
C\oplus(\chi_0)
=
\begin{pmatrix}
1&0&0\\
0&-1&0\\
0&0&\chi_0
\end{pmatrix},
\qquad
\chi_0\in\{+1,-1\}.
\end{equation}
Here $\chi_0$ indicates whether the neutral fiber line is reflection-even or
reflection-odd. Since the neutral block of $J_3$ is zero, either choice of
$\chi_0$ is compatible with reflection. Indeed,
\begin{equation}\label{eq:rthree_C3_J3_relation} C_3J_3C_3^{-1}=-J_3.
\end{equation}
Therefore the lifted reflection on $S\otimes E_{\mathbb C}^{(3)}$ is
\begin{equation}\label{eq:rthree_reflection_lift} \mathcal U_{\mathbf r}^{E_{\mathbb C}^{(3)}} = U_{\mathbf r}C_3\mathbf r^*.
\end{equation}
By the reflection-compatibility relation \eqref{eq:rthree_C3_J3_relation}, the
rank-three twisted Dirac operator is pointwise invariant under this lifted reflection.

\subsection{Complex fiber weights and separated blocks}
\label{subsec:rank_three_blocks}

After complexification, the rotating two-plane gives two fiber eigenvectors and the
neutral line gives one zero-weight eigenvector. Let
\begin{equation}\label{eq:rthree_eigenvectors}
\xi_+
=
\begin{pmatrix}
1\\
-i\\
0
\end{pmatrix},
\qquad
\xi_-
=
\begin{pmatrix}
1\\
i\\
0
\end{pmatrix},
\qquad
\eta_0
=
\begin{pmatrix}
0\\
0\\
1
\end{pmatrix}.
\end{equation}
Then
\begin{equation}\label{eq:rthree_eigenvalue_equations} J_3\xi_+=i\xi_+, \qquad J_3\xi_-=-i\xi_-, \qquad J_3\eta_0=0.
\end{equation}
Thus the complexified fiber weights are $+1$, $-1$, and $0$. For a Fourier mode
$k\in \mathbb Z$, the three separated blocks are
\begin{equation}\label{eq:rthree_separated_blocks}
\Psi_{k,+}
=
e^{ik\theta}
\binom{u_{k,+}(t)}{v_{k,+}(t)}
\otimes\xi_+,
\qquad
\Psi_{k,-}
=
e^{ik\theta}
\binom{u_{k,-}(t)}{v_{k,-}(t)}
\otimes\xi_-,
\qquad
\Psi_{k,0}
=
e^{ik\theta}
\binom{u_{k,0}(t)}{v_{k,0}(t)}
\otimes\eta_0.
\end{equation}
The angular covariant derivative acts by
\begin{equation}\label{eq:rthree_block_parameters}
\left(\partial_\theta+A_pJ_3\right)\Psi_{k,+}
=
i(k+A_p)\Psi_{k,+},
\qquad
\left(\partial_\theta+A_pJ_3\right)\Psi_{k,-}
=
i(k-A_p)\Psi_{k,-},
\qquad
\left(\partial_\theta+A_pJ_3\right)\Psi_{k,0}
=
ik\Psi_{k,0}.
\end{equation}
Hence the three scalar block parameters are
\begin{equation}\label{eq:rthree_scalar_parameters} m_{k,+}(p)=k+A_p, \qquad m_{k,-}(p)=k-A_p, \qquad m_{k,0}(p)=k.
\end{equation}

The first two parameters are exactly the moving rank-two parameters. The third parameter
is independent of $p$. This is the basic structural point: the rank-three model is
the direct sum of one moving rank-two orthogonal twist and one stationary neutral
block.

\subsection{Rotating crossings}
\label{subsec:rank_three_rotating_crossings}

The rotating blocks are governed by the same scalar radial equation as in the rank-two
case. Therefore the regularized APS boundary family produces isolated crossings exactly
when one of the moving scalar parameters $k+A_p$ or $k-A_p$ vanishes.

For a nonzero reflection-paired Fourier block $\{k,-k\}$, with $k\geq1$, the four
rotating scalar parameters are
\begin{equation}\label{eq:rthree_four_rotating_parameters}
m_{k,+}(p)=k+A_p,
\qquad
m_{k,-}(p)=k-A_p,
\qquad
m_{-k,+}(p)=-k+A_p,
\qquad
m_{-k,-}(p)=-k-A_p.
\end{equation}
Thus the nonzero rotating crossing condition is
\begin{equation}\label{eq:rthree_nonzero_crossing_condition} A_{p_*}=k \qquad\text{or}\qquad A_{p_*}=-k.
\end{equation}
At either value, the two zero blocks are exchanged by the lifted reflection, so the
real crossing eigenspace carries the irreducible $O(2)$-representation $\rho_k$.

For the rotating zero Fourier block, the moving scalar parameters are $A_p$ and
$-A_p$. Therefore the zero-mode rotating crossing condition is
\begin{equation}\label{eq:rthree_rotating_zero_crossing_condition} A_{p_*}=0.
\end{equation}
At such a crossing, the real crossing eigenspace decomposes as one reflection-even line
and one reflection-odd line. Hence its $O(2)$-representation is
\begin{equation}\label{eq:rthree_zero_rotating_representation} \mathbf 1\oplus\det.
\end{equation}

\subsection{The neutral block}
\label{subsec:rank_three_neutral_block}

For $J_3=J\oplus0$, the neutral fiber direction has effective parameter
$m_{k,0}(p)=k$.
Thus only the two rotating blocks $k+A_p$ and $k-A_p$ contribute to the moving
rank-three spectral flow; the periodic neutral zero block is handled by
Definition~\ref{def:hrank_neutral_sector_convention}.

\subsection{Rank-three spectral-flow formula}
\label{subsec:rank_three_spectral_flow_formula}

Let
\begin{equation}\label{eq:rthree_nonzero_crossing_set}
\mathcal C_k^{(3)}
=
\left\{
p_*\in(0,1):
A_{p_*}=k
\text{ or }
A_{p_*}=-k
\right\},
\qquad
k\geq1,
\end{equation}
and let
\begin{equation}\label{eq:rthree_zero_crossing_set} \mathcal C_0^{(3)} = \left\{ p_*\in(0,1): A_{p_*}=0 \right\}.
\end{equation}
For $p_*\in\mathcal C_k^{(3)}$, let $\epsilon_{k,p_*}^{(3)}$ be the local sign
of the rotating nonzero-block crossing. For $p_*\in\mathcal C_0^{(3)}$, let
$\epsilon_{1,p_*}^{(3)}$ and $\epsilon_{\det,p_*}^{(3)}$ be the signs on the
trivial and determinant blocks of the rotating zero-block crossing form.

\begin{proposition}[Rank-three formula with fixed neutral sector]
\label{prop:rthree_formula_fixed_neutral}
Use the rank-three version of the hypotheses of
Theorem~\ref{thm:hrank_ROO2_spectral_flow}: the endpoints are invertible, the scalar
regularized APS family and rotating zero-block regularization are admissible, all
rotating crossings are regular, simultaneous crossings are treated by the separated or
block-diagonal convention, and the neutral $k=0$ block is fixed by the
reflection-compatible neutral-sector convention of
Definition~\ref{def:hrank_neutral_sector_convention}. Then
\begin{equation}\label{eq:rthree_ROO2_formula}
\operatorname{sf}_{O(2)}
\left(
D_{p,\text{APS,reg}}^{E_{\mathbb C}^{(3)}}
\right)
=
\sum_{k\geq 1}
\sum_{p_*\in \mathcal C_k^{(3)}}
\epsilon_{k,p_*}^{(3)}[\rho_k]
+
\sum_{p_*\in \mathcal C_0^{(3)}}
\left(
\epsilon_{1,p_*}^{(3)}[\mathbf 1]
+
\epsilon_{\det,p_*}^{(3)}[\det]
\right).
\end{equation}
The neutral block contributes no moving term. With the reflection-compatible neutral
completion of Definition~\ref{def:hrank_neutral_sector_convention}, the formula is the
full $RO(O(2))$-valued spectral flow of the rank-three family.
\end{proposition}

\begin{proof}
The rank-three twist splits into the rotating two-plane and the neutral line. The
rotating two-plane is exactly the rank-two orthogonal-twist model, so its nonzero
Fourier crossings contribute $\epsilon_{k,p_*}^{(3)}[\rho_k]$, and its zero
Fourier crossings contribute
$\epsilon_{1,p_*}^{(3)}[\mathbf 1]+\epsilon_{\det,p_*}^{(3)}[\det]$. The neutral
sector is independent of $A_p$. For $k\neq0$, it is stationary and non-crossing.
For $k=0$, the scaled zero-energy trace would be constant at the two endpoints, while
the fixed neutral convention imposes the opposite-endpoint graph
$\widehat\gamma_T=-\widehat\gamma_0$; hence the persistent neutral zero mode is removed.
It therefore has zero spectral-flow contribution.  Since the neutral convention is reflection-compatible, this is
the full $O(2)$-equivariant contribution of the rank-three family.
Summing the rotating contributions proves \eqref{eq:rthree_ROO2_formula}. The integer
formula follows by applying the dimension map.
\end{proof}

\begin{corollary}[Rank three reduces to the rank-two moving formula]
\label{cor:rthree_equals_rank_two_moving_formula}
Under the hypotheses of Proposition~\ref{prop:rthree_formula_fixed_neutral}, the
rank-three $RO(O(2))$-valued spectral flow agrees with the rank-two moving contribution:
\begin{equation}\label{eq:rthree_equals_ranktwo}
\operatorname{sf}_{O(2)}
\left(
D_{p,\text{APS,reg}}^{E_{\mathbb C}^{(3)}}
\right)
=
\operatorname{sf}_{O(2)}
\left(
D_{p,\text{APS,reg}}^{E_{\mathbb C}^{(2)}}
\right),
\end{equation}
provided the rank-two operator is taken with the same rotating parameter path $A_p$
and the same regularized APS convention on the rotating blocks.
\end{corollary}

\begin{proof}
Both sides have the same moving block parameters $k+A_p$ and $k-A_p$, and hence
the same rotating crossing sets and local signs. The only additional rank-three block
is the neutral line, which contributes zero by the fixed neutral-sector convention of
Definition~\ref{def:hrank_neutral_sector_convention}.
\end{proof}

\subsection{Affine rank-three examples}
\label{subsec:rank_three_affine_examples}

Suppose now that
\begin{equation}\label{eq:rthree_affine_path} A_p=A+\nu p, \qquad \nu\neq0.
\end{equation}
We use the scalar $m$-slope sign convention fixed in
Proposition~\ref{prop:hrank_affine_crossing_signs}. Thus the local sign is determined
by the derivative of the scalar transfer determinant along the actual scalar
parameter:
\begin{equation}\label{eq:rthree_affine_crossing_sign_from_m} \delta = \operatorname{sign}\left(F'(0)m'(p_*)\right).
\end{equation}
This is the convention induced by the crossing form
\eqref{eq:hrank_operator_crossing_form}. In particular, for the affine path
\eqref{eq:rthree_affine_path}, the two rotating scalar parameters in rank three are
\begin{equation}\label{eq:rthree_affine_scalar_parameters} m_{k,+}(p)=k+A_p, \qquad m_{k,-}(p)=k-A_p,
\end{equation}
and hence
\begin{equation}\label{eq:rthree_affine_scalar_derivatives} m'_{k,+}(p)=\nu, \qquad m'_{k,-}(p)=-\nu.
\end{equation}
Thus the two branches have opposite scalar crossing signs.

Assume, for this subsection, that the scalar regularization is oriented so that
\begin{equation}\label{eq:rthree_affine_Fprime_positive} F'(0)>0.
\end{equation}
Then the $+$-branch contributes with sign $\operatorname{sign}(\nu)$, while
the $-$-branch contributes with sign $-\operatorname{sign}(\nu)$. Therefore the
rank-three affine formula is
\begin{equation}\label{eq:rthree_affine_formula}
\begin{split}
\operatorname{sf}_{O(2)}
\left(
D_{p,\text{APS,reg}}^{E_{\mathbb C}^{(3)}}
\right)
=
&
\operatorname{sign}(\nu)
\sum_{k\in\mathcal K_+}
\#\left\{
p\in(0,1):
A+\nu p=-k
\right\}
[\rho_k]
\\
&
-
\operatorname{sign}(\nu)
\sum_{k\in\mathcal K_+}
\#\left\{
p\in(0,1):
A+\nu p=k
\right\}
[\rho_k]
\\
&
+
\operatorname{sign}(\nu)
\#\left\{
p\in(0,1):
A+\nu p=0
\right\}
\left(
[\mathbf 1]-[\det]
\right).
\end{split}
\end{equation}
Here the convention in the last line is that the scalar branch $m_{0,+}=A_p$
contributes to the $\mathbf 1$-line and the scalar branch $m_{0,-}=-A_p$
contributes to the $\det$-line. With the opposite zero-block identification, the
last factor would be $-[\mathbf 1]+[\det]$. The neutral block does not appear in
this expression.

For example, if $A_p=-1.2+p$, then $\nu=1$, and the unique positive-label
crossing occurs at
$A_p=-1$, with $p_*=0.2$.
The representative branch is
$q(p)=m_{1,+}(p)=1+A_p=p-0.2$,
so $q'(p_*)=1$.  Since this subsection assumes $F'(0)>0$, the representative
crossing-form sign is
$\operatorname{sign}\left(F'(0)q'(p_*)\right)=+1$.
The reflected scalar coordinate is
$\widetilde q(p)=m_{-1,-}(p)=-1-A_p=0.2-p=-q(p)$,
so it has the opposite raw derivative.  This does not change the $RO(O(2))$
coefficient, because in the reflected chart the determinant is
$\widetilde F(\widetilde q)=F(-\widetilde q)$, and therefore
$\widetilde F'(0)\widetilde q'(p_*)=F'(0)q'(p_*)$.  Equivalently, the reflected
block is the reflected copy of the representative block, and together
$(1,+)\oplus(-1,-)$ carries one copy of $\rho_1$.  Hence
\begin{equation}\label{eq:rthree_example_nonzero_block}
\operatorname{sf}_{O(2)}
\left(
D_{p,\text{APS,reg}}^{E_{\mathbb C}^{(3)}}
\right)
=
[\rho_1].
\end{equation}

If $A_p=-0.2+p$, then $\nu=1$, and the unique moving rotating zero-mode crossing
occurs at
$A_p=0$, with $p_*=0.2$.
At this point the two scalar zero-mode branches $m_{0,+}=A_p$ and
$m_{0,-}=-A_p$ belong to the two self-paired reflection sectors $\mathbf 1$
and $\det$.  Their opposite $m$-slopes therefore remain as opposite signs on
two distinct $O(2)$-types, and with the zero-block identification fixed above,
\begin{equation}\label{eq:rthree_example_zero_block}
\operatorname{sf}_{O(2)}
\left(
D_{p,\text{APS,reg}}^{E_{\mathbb C}^{(3)}}
\right)
=
[\mathbf 1]-[\det].
\end{equation}
This class is killed by the dimension map:
\begin{equation}\label{eq:rthree_dimension_map_examples} \dim\left([\mathbf 1]-[\det]\right)=0.
\end{equation}
Thus the zero-mode example is the one where the opposite signs survive as a
virtual reflection-sector class rather than being absorbed into a single nonzero
Fourier representation $\rho_k$.

\begin{remark}[What would go wrong without the neutral convention]
\label{rem:rthree_without_neutral_convention}
If the neutral $k=0$ sector is not fixed by an auxiliary $p$-independent boundary
condition, then the block $m_{0,0}(p)=0$ persists for all $p$. This is not a
regular crossing and should not be inserted into the local crossing formula. In that
case the family may fail to have invertible endpoints, and the usual spectral-flow
formula is not directly applicable. The neutral-sector convention is therefore not a
cosmetic choice; it is required to separate the moving rank-two spectral-flow mechanism
from the stationary neutral zero mode.
\end{remark}

\section{Examples and loss of information under the dimension map}
\label{sec:examples_dimension_map_loss}

Now we give explicit examples illustrating the two loss mechanisms previewed in the
introduction: the dimension map can identify distinct nonzero Fourier representations,
and it can kill a signed zero-mode class. We use the representation-ring convention
fixed in Remark~\ref{rem:representation_ring_convention}, so
\begin{equation}\label{eq:examples_dimension_map} \dim:RO(O(2))\longrightarrow \mathbb Z
\end{equation}
is normalized by
\begin{equation}\label{eq:examples_basic_dimensions} \dim[\rho_k]=2, \qquad \dim[\mathbf 1]=1, \qquad \dim[\det]=1.
\end{equation}
while
\begin{equation}\label{eq:examples_zero_difference_dimension} \dim\left([\mathbf 1]-[\det]\right)=0.
\end{equation}

Throughout the examples, we work in the periodic spin sector and use affine paths
\begin{equation}\label{eq:examples_affine_path_general} A_p=A+\nu p, \qquad p\in[0,1], \qquad \nu\neq0.
\end{equation}
All examples in this section use the scalar $m$-slope sign convention of
Proposition~\ref{prop:hrank_affine_crossing_signs}. Thus, the local sign of a scalar
crossing is determined by
\begin{equation}\label{eq:examples_m_slope_sign_convention} \delta = \operatorname{sign}\left(F'(0)m'(p_*)\right),
\end{equation}
which is the sign induced by the crossing form
\eqref{eq:hrank_operator_crossing_form}. For simplicity, throughout the concrete
examples below, we assume that the scalar regularization is oriented so that
\begin{equation}\label{eq:examples_Fprime_positive} F'(0)>0.
\end{equation}

For a nonzero rotating block in rank two, the two scalar branches are
\begin{equation}\label{eq:examples_rank_two_scalar_branches} m_{k,+}(p)=k+A_p, \qquad m_{k,-}(p)=k-A_p.
\end{equation}
Hence, for the affine path $A_p=A+\nu p$,
\begin{equation}\label{eq:examples_rank_two_scalar_branch_derivatives} m'_{k,+}(p)=\nu, \qquad m'_{k,-}(p)=-\nu.
\end{equation}
Thus the two branches have opposite scalar crossing signs. In the rotating zero
block, we use the convention that the branch $m_{0,+}=A_p$ contributes to the
$\mathbf 1$-line and the branch $m_{0,-}=-A_p$ contributes to the $\det$-line.
Therefore, under \eqref{eq:examples_Fprime_positive}, a zero-mode crossing contributes
\begin{equation}\label{eq:examples_zero_mode_signed_class} \operatorname{sign}(\nu)\left([\mathbf 1]-[\det]\right).
\end{equation}
With the opposite zero-block identification, this class would be
$\operatorname{sign}(\nu)(-[\mathbf 1]+[\det])$.

\subsection{Rank-two benchmark: nonzero block versus zero block}
\label{subsec:examples_rank_two_benchmark}

We first recall the basic rank-two phenomenon. Take $n=2$, so that $J_2=J$, and work in
the periodic sector. Consider the path
\begin{equation}\label{eq:examples_rank_two_path_nonzero} A_p=-1.2+p.
\end{equation}
The endpoints are $A_0=-1.2$ and $A_1=-0.2$, so neither endpoint is an integer. The
unique crossing occurs when $A_p=-1$, namely at
\begin{equation}\label{eq:examples_rank_two_nonzero_crossing_time} p=0.2.
\end{equation}
This is the $+$-branch crossing
\begin{equation}\label{eq:examples_rank_two_nonzero_plus_branch} m_{1,+}(p)=1+A_p=0.
\end{equation}
Since $\nu=1$ and $F'(0)>0$, the local sign is positive. The crossing lies in the
nonzero reflection-paired block $\{1,-1\}$, and hence
\begin{equation}\label{eq:examples_rank_two_nonzero_contribution}
\operatorname{sf}_{O(2)}
\left(
D_{p,\text{APS,reg}}^{E_{\mathbb C}^{(2)}}
\right)
=
[\rho_1].
\end{equation}
Applying the dimension map gives ordinary spectral flow
\begin{equation}\label{eq:examples_rank_two_nonzero_dimension} \dim[\rho_1]=2.
\end{equation}

Now consider the shifted path
\begin{equation}\label{eq:examples_rank_two_path_zero} A_p=-0.2+p.
\end{equation}
The endpoints are $A_0=-0.2$ and $A_1=0.8$, so again the endpoints are invertible. The
unique moving rotating zero-mode crossing occurs when $A_p=0$, again at
\begin{equation}\label{eq:examples_rank_two_zero_crossing_time} p=0.2.
\end{equation}
At this crossing, the two zero-mode scalar branches are
\begin{equation}\label{eq:examples_rank_two_zero_branches} m_{0,+}(p)=A_p, \qquad m_{0,-}(p)=-A_p.
\end{equation}
They cross with opposite $m$-slopes. With the zero-block convention fixed above,
the $\mathbf 1$-line has positive sign and the $\det$-line has negative sign.
Therefore
\begin{equation}\label{eq:examples_rank_two_zero_contribution}
\operatorname{sf}_{O(2)}
\left(
D_{p,\text{APS,reg}}^{E_{\mathbb C}^{(2)}}
\right)
=
[\mathbf 1]-[\det].
\end{equation}
Applying the dimension map gives
\begin{equation}\label{eq:examples_rank_two_zero_dimension} \dim\left([\mathbf 1]-[\det]\right)=0.
\end{equation}

Thus, the two paths no longer have the same ordinary spectral flow under the
$m$-slope convention. Instead, they illustrate two different features. The nonzero
crossing contributes the two-dimensional representation $[\rho_1]$, while the
rotating zero-mode crossing contributes the signed class $[\mathbf 1]-[\det]$,
which is nonzero in $RO(O(2))$ but becomes invisible after applying the ordinary dimension
map:
\begin{equation}\label{eq:examples_rank_two_zero_invisible}
[\mathbf 1]-[\det]\neq0
\quad\text{in }RO(O(2)),
\qquad
\dim\left([\mathbf 1]-[\det]\right)=0.
\end{equation}

\begin{remark}[Interpretation of the shift]
\label{rem:examples_rank_two_shift_interpretation}
The two paths above differ by an integer shift of the scalar parameter, but here we are
not quotienting the path by an equivariant gauge identification. We keep the fixed
Fourier decomposition and the fixed $O(2)$-linearization. With this convention, the
first path is represented in the nonzero Fourier block $\{1,-1\}$, while the second
path is represented in the self-paired zero block. Under the $m$-slope sign
convention, the self-paired zero block carries opposite signs on the two real
one-dimensional reflection sectors.
\end{remark}

\subsection{Identical higher-rank blocks}
\label{subsec:examples_identical_higher_rank_blocks}

We now repeat the same construction for the identical-block higher-rank twist
\begin{equation}\label{eq:examples_identical_J} J_{2r}=J\oplus\cdots\oplus J.
\end{equation}
There are $r$ identical rotating two-planes, so every rank-two moving crossing is copied
$r$ times.

For the path $A_p=-1.2+p$, the unique moving crossing occurs in the nonzero block
$\{1,-1\}$. Since there are $r$ identical rotating two-planes,
\begin{equation}\label{eq:examples_identical_nonzero_contribution}
\operatorname{sf}_{O(2)}
\left(
D_{p,\text{APS,reg}}^{E_{\mathbb C}^{(2r)}}
\right)
=
r[\rho_1].
\end{equation}
After applying the dimension map,
\begin{equation}\label{eq:examples_identical_nonzero_dimension} \dim r[\rho_1]=2r.
\end{equation}

For the shifted path $A_p=-0.2+p$, the unique moving crossing occurs in the rotating
zero block. Each rotating two-plane contributes $[\mathbf 1]-[\det]$ under the
zero-block convention fixed above. Hence
\begin{equation}\label{eq:examples_identical_zero_contribution}
\operatorname{sf}_{O(2)}
\left(
D_{p,\text{APS,reg}}^{E_{\mathbb C}^{(2r)}}
\right)
=
r\left([\mathbf 1]-[\det]\right).
\end{equation}
After applying the dimension map,
\begin{equation}\label{eq:examples_identical_zero_dimension} \dim r\left([\mathbf 1]-[\det]\right)=0.
\end{equation}
Thus the identical-block higher-rank example is a multiplicity-enhanced version of the
signed zero-mode phenomenon: the representation-valued class may be nonzero even when
the ordinary spectral flow is zero.

\subsection{Weighted rank-four example}
\label{subsec:examples_weighted_rank_four}

The next example shows a genuinely higher-rank feature. To avoid accidental simultaneous
crossings, take
\begin{equation}\label{eq:examples_weighted_rank_four_J} J_4=2J\oplus3J.
\end{equation}
Thus the rotating weights are $\mu_1=2$ and $\mu_2=3$. In the first rotating plane, the
$k$-th nonzero block crosses when $A_p=\pm k/2$. In the second rotating plane, the
$k$-th nonzero block crosses when $A_p=\pm k/3$.

Consider first the path
\begin{equation}\label{eq:examples_weighted_path_first} A_p=-\frac{3}{5}+\frac{1}{5}p.
\end{equation}
This path has $\nu=1/5>0$ and crosses $A_p=-1/2$ at $p=1/2$. The value
$-1/2$ is the $k=1$ $+$-branch crossing for the first rotating plane, since
\begin{equation}\label{eq:examples_weighted_first_plus_branch} m_{1,1,+}(p)=1+2A_p=0.
\end{equation}
It is not a crossing value for the second rotating plane, because the equation
$-1/2=-k/3$ would force $k=3/2$, which is not an allowed integer Fourier label in
the periodic sector. Therefore, with $F'(0)>0$,
\begin{equation}\label{eq:examples_weighted_first_contribution}
\operatorname{sf}_{O(2)}
\left(
D_{p,\text{APS,reg}}^{E_{\mathbb C}^{(4)}}
\right)
=
[\rho_1].
\end{equation}

Now consider the path
\begin{equation}\label{eq:examples_weighted_path_second} A_p=-\frac{2}{5}+\frac{2}{15}p.
\end{equation}
This path has $\nu=2/15>0$ and crosses $A_p=-1/3$ at $p=1/2$. The value
$-1/3$ is the $k=1$ $+$-branch crossing for the second rotating plane, since
\begin{equation}\label{eq:examples_weighted_second_plus_branch} m_{1,2,+}(p)=1+3A_p=0.
\end{equation}
It is not a crossing value for the first rotating plane, because the equation
$-1/3=-k/2$ would force $k=2/3$, which is not an allowed integer Fourier label.
Therefore
\begin{equation}\label{eq:examples_weighted_second_contribution}
\operatorname{sf}_{O(2)}
\left(
D_{p,\text{APS,reg}}^{E_{\mathbb C}^{(4)}}
\right)
=
[\rho_1].
\end{equation}

These two examples have the same $RO(O(2))$-value even though the internal rotating
two-plane responsible for the crossing is different. Thus the $RO(O(2))$-valued
spectral flow detects the angular representation class but not, by itself, the internal
label $j$.

\begin{remark}[What $RO(O(2))$ can and cannot see]
\label{rem:examples_ROO2_limit}
The $RO(O(2))$-valued invariant sees the representation of the geometric symmetry
group of the base circle. It distinguishes, for example, a nonzero-mode crossing
$[\rho_k]$ from a signed zero-mode crossing such as $[\mathbf 1]-[\det]$, and it
distinguishes $[\rho_k]$ from $[\rho_\ell]$ when $k\neq\ell$. However, if two
different internal rotating two-planes produce crossings with the same angular label
$\rho_k$, then $RO(O(2))$ keeps only their sum as copies of $[\rho_k]$. To remember
the internal plane, one would need a refined invariant keeping track of the internal
centralizer representation.
\end{remark}

\subsection{Different Fourier labels with the same ordinary dimension}
\label{subsec:examples_different_fourier_labels}

The dimension map also forgets which nonzero Fourier mode crossed. For every $k\geq1$,
the representation $\rho_k$ has dimension two:
\begin{equation}\label{eq:examples_rhok_rhol_same_dimension} \dim[\rho_k]=2 \qquad \text{for all }k\geq1.
\end{equation}
But $\rho_k$ and $\rho_\ell$ are distinct real $O(2)$-representation classes whenever
$k\neq\ell$.

For example, in the rank-two model, the path
\begin{equation}\label{eq:examples_path_rho_one} A_p=-1.2+p
\end{equation}
has a unique $+$-branch crossing at $A_p=-1$, so
\begin{equation}\label{eq:examples_rho_one_contribution} \operatorname{sf}_{O(2)} = [\rho_1].
\end{equation}
On the other hand, the path
\begin{equation}\label{eq:examples_path_rho_two} A_p=-2.2+p
\end{equation}
has a unique $+$-branch crossing at $A_p=-2$, so
\begin{equation}\label{eq:examples_rho_two_contribution} \operatorname{sf}_{O(2)} = [\rho_2].
\end{equation}
Both have ordinary spectral flow equal to $2$, but
\begin{equation}\label{eq:examples_rho_one_neq_rho_two} [\rho_1]\neq[\rho_2] \qquad \text{in }RO(O(2)).
\end{equation}
Thus ordinary spectral flow cannot distinguish which angular Fourier representation
carried the crossing.

\subsection{A weighted example distinguishing \texorpdfstring{$\rho_1$ and $\rho_2$}{rho1 and rho2}}
\label{subsec:examples_weighted_rho1_rho2}

Weighted higher-rank twists can also produce different Fourier labels at different
parameter values without forcing simultaneous crossings in the other rotating block.
For this purpose, take
\begin{equation}\label{eq:examples_weighted_rho_J} J_4=J\oplus\sqrt{2}J.
\end{equation}
The irrational weight prevents the weighted crossing values from accidentally
coinciding with the integer crossing values of the first block.

Consider first the path
\begin{equation}\label{eq:examples_weighted_rho1_path} A_p=-\frac{1}{\sqrt{2}}-\frac{1}{10}+\frac{1}{5}p.
\end{equation}
This path has positive slope and crosses $A_p=-1/\sqrt{2}$ at $p=1/2$. This is
the $k=1$ $+$-branch crossing in the weighted block with $\mu_2=\sqrt{2}$, and
it is not a crossing value for the unweighted block. Hence
\begin{equation}\label{eq:examples_weighted_rho1_contribution} \operatorname{sf}_{O(2)} = [\rho_1].
\end{equation}

Now consider the path
\begin{equation}\label{eq:examples_weighted_rho2_path} A_p=-\sqrt{2}-\frac{1}{10}+\frac{1}{5}p.
\end{equation}
This path has positive slope and crosses $A_p=-\sqrt{2}$ at $p=1/2$. This is the
$k=2$ $+$-branch crossing in the weighted block, because
\begin{equation}\label{eq:examples_weighted_rho2_crossing_identity} -\sqrt{2}=-\frac{2}{\sqrt{2}}.
\end{equation}
It is again not a crossing value for the unweighted block. Therefore
\begin{equation}\label{eq:examples_weighted_rho2_contribution} \operatorname{sf}_{O(2)} = [\rho_2].
\end{equation}

Again both $[\rho_1]$ and $[\rho_2]$ have dimension two, but they are different
classes in $RO(O(2))$. This shows that the equivariant invariant remembers the
Fourier label even when the ordinary spectral flow does not.

\subsection{General loss-of-information statement}
\label{subsec:examples_general_loss_statement}

We summarize the preceding examples in a single statement.

\begin{proposition}[The dimension map loses equivariant spectral-flow information]
\label{prop:examples_dimension_map_loses_information}
In the higher-rank orthogonal-twist model, the ordinary spectral flow obtained from
$\operatorname{sf}_{O(2)}$ by applying the dimension map does not determine the
$RO(O(2))$-valued spectral flow.

More precisely, there are two basic ways in which information is lost. First, the
dimension map can kill a nontrivial signed zero-mode class. For example, under the
zero-block convention fixed above, one may obtain
\begin{equation}\label{eq:examples_general_zero_kernel_class} \operatorname{sf}_{O(2)} = [\mathbf 1]-[\det],
\end{equation}
which has image $0\in\mathbb Z$ under the dimension map. This has the same ordinary
spectral flow as a path with no crossings, but a different $RO(O(2))$-valued spectral
flow.

Second, for $k\neq\ell$, one may obtain
\begin{equation}\label{eq:examples_general_loss_rhok_rhol}
\operatorname{sf}_{O(2)}=[\rho_k]
\qquad
\text{and}
\qquad
\operatorname{sf}_{O(2)}=[\rho_\ell],
\end{equation}
both of which have image $2\in\mathbb Z$ under the dimension map, although
$[\rho_k]\neq[\rho_\ell]$ in $RO(O(2))$.
\end{proposition}

\begin{proof}
The first phenomenon is realized by the rank-two path $A_p=-0.2+p$, or by the same
moving part inside the rank-three model with fixed neutral sector. This path has a
rotating zero-mode crossing. Under the scalar $m$-slope convention and the zero-block
identification fixed above, its contribution is $[\mathbf 1]-[\det]$. This class is
nonzero in $RO(O(2))$, but its image under the dimension map is zero. A path with no
crossing has $RO(O(2))$-valued spectral flow $0$, and hence has the same ordinary
spectral flow but a different representation-valued spectral flow.

The second phenomenon is realized using affine paths with unique crossings in different
nonzero Fourier blocks, for example $A_p=-1.2+p$ and $A_p=-2.2+p$. These give
$[\rho_1]$ and $[\rho_2]$, respectively. More generally, paths can be chosen so that
the unique $+$-branch crossing occurs at $A_p=-k$ or $A_p=-\ell$. Since
$\rho_k$ and $\rho_\ell$ are distinct for $k\neq\ell$, the equivariant spectral
flows are different, while their ordinary dimensions are both two.
\end{proof}

\begin{remark}[Meaning of the examples]
\label{rem:examples_meaning}
The examples show that $RO(O(2))$-valued spectral flow keeps track of the symmetry
representation of the crossing kernel together with its crossing-form sign. Ordinary
spectral flow keeps only the dimension-weighted signed count. Therefore it can forget
whether a nontrivial signed zero-mode representation occurred, and it can forget which
nonzero Fourier representation carried a crossing.
\end{remark}

\begin{remark}[Two kinds of information loss]
\label{rem:examples_two_kinds_information_loss}
The examples above exhibit two different ways in which the dimension map loses
information. The first is the usual representation-label loss: for $k\neq\ell$,
one has
$[\rho_k]\neq[\rho_\ell]$ in $RO(O(2))$,
but
$\dim[\rho_k]=\dim[\rho_\ell]=2$.
Thus, ordinary spectral flow remembers the signed dimension of the crossing space
but forgets which nonzero Fourier representation carried the crossing.

The second phenomenon is stronger. A rotating zero-mode crossing may contribute
the signed class
$[\mathbf 1]-[\det]$.
This class can be nonzero in $RO(O(2))$, but
$\dim\left([\mathbf 1]-[\det]\right)=0$.
Hence, ordinary spectral flow does not merely identify two different nonzero representation classes; it can completely eliminate a nontrivial signed equivariant zero-mode contribution.
\end{remark}

\section{APS index and endpoint $\eta$ terms}
\label{sec:hrank_APS_index_endpoint_eta}
The moving-family result of the paper is the $RO(O(2))$-valued spectral-flow formula above. Now, we describe the fixed-parameter APS consequences in the cylinder model.

We discuss the APS index consequences of the higher-rank orthogonal-twist model. In the spectral-flow sections, the parameter $A_p$ varied along a path and the regularized APS boundary condition was used to obtain a continuous path of self-adjoint Fredholm realizations. Now, we fix a single parameter value $A$, or more generally two endpoint values $A_0$ and $A_T$, and discuss the corresponding chiral
APS index.

In this two-dimensional cylinder model, the interior local
index density vanishes. For a constant parameter $A$, the higher-rank orthogonal connection
$\nabla_A^{E_{\mathbb R}^{(n)}}=d+AJ_n\,d\theta$
is flat. Thus, the twisting field contributes no curvature density in the interior. The purely Riemannian contribution also gives no two-form density on a two-dimensional spin cylinder. Therefore, the chiral APS index is controlled entirely by endpoint $\eta$ terms, up to the finite-dimensional correction coming from whatever APS kernel convention is chosen for non-invertible boundary operators.

\subsection{The fixed-parameter higher-rank chiral APS operator}
\label{subsec:hrank_fixed_parameter_chiral_APS}

Let $A$ be a real parameter. The higher-rank twisting connection is
\begin{equation}\label{eq:eta_fixed_higher_rank_connection} \nabla_A^{E_{\mathbb R}^{(n)}}=d+AJ_n\,d\theta .
\end{equation}
Its complexification is
\begin{equation}\label{eq:eta_fixed_complex_connection} \nabla_A^{E_{\mathbb C}^{(n)}}=d+AJ_n\,d\theta .
\end{equation}
Let $D_A^{E_{\mathbb C}^{(n)}}$ denote the corresponding twisted Dirac operator on
$S\otimes E_{\mathbb C}^{(n)}$. Since $S=S^+\oplus S^-$, the Dirac operator has
the chiral form
\begin{equation}\label{eq:eta_chiral_decomposition}
D_A^{E_{\mathbb C}^{(n)}}
=
\begin{pmatrix}
0&D_A^-\\
D_A^+&0
\end{pmatrix}.
\end{equation}
We write
\begin{equation}\label{eq:eta_chiral_APS_operator}
D_{A,\text{APS}}^{E_{\mathbb C}^{(n)},+}
:
\Dom\left(D_{A,\text{APS}}^{E_{\mathbb C}^{(n)},+}\right)
\subset
L^2(M;S^+\otimes E_{\mathbb C}^{(n)})
\longrightarrow
L^2(M;S^-\otimes E_{\mathbb C}^{(n)})
\end{equation}
for the chiral operator with APS boundary condition.

The endpoint boundary operators associated to $D_A^{E_{\mathbb C}^{(n)},+}$ are
denoted by
\begin{equation}\label{eq:eta_endpoint_boundary_operators_def} B_{A,0}^{E_{\mathbb C}^{(n)},+}, \qquad B_{A,T}^{E_{\mathbb C}^{(n)},+}.
\end{equation}
With the sign convention fixed in \autoref{sec:setup}, these are the positive-chirality
boundary operators entering the APS projection at $Y_0$ and $Y_T$, respectively.

\subsection{Endpoint boundary spectra in separated blocks}
\label{subsec:hrank_endpoint_boundary_spectra}

We now describe the endpoint tangential spectra for the chiral operator. In the chiral index discussion, the boundary operator is the tangential operator entering the APS condition for $D_A^+$.

With the sign convention fixed by the inward normals in \autoref{sec:setup}, the endpoint chiral tangential operators have the following separated-block eigenvalues. For a rotating two-plane with weight $\mu_j$, let
\begin{equation}\label{eq:eta_rotating_parameters_fixed_A} m_{k,j,+}(A)=k+\mu_jA, \qquad m_{k,j,-}(A)=k-\mu_jA .
\end{equation}
For the neutral block,
\begin{equation}\label{eq:eta_neutral_parameter_fixed_A} m_{k,a,0}(A)=k .
\end{equation}

On a rotating block $(k,j,\sigma)$, the endpoint chiral tangential eigenvalues are
\begin{equation}\label{eq:eta_rotating_boundary_eigenvalues}
\lambda_{0}^{(k,j,\sigma)}(A)
=
\frac{m_{k,j,\sigma}(A)}{f(0)},
\qquad
\lambda_{T}^{(k,j,\sigma)}(A)
=
-\frac{m_{k,j,\sigma}(A)}{f(T)} .
\end{equation}
On a neutral block $(k,a,0)$, the corresponding eigenvalues are
\begin{equation}\label{eq:eta_neutral_boundary_eigenvalues}
\lambda_{0}^{(k,a,0)}(A)
=
\frac{k}{f(0)},
\qquad
\lambda_{T}^{(k,a,0)}(A)
=
-\frac{k}{f(T)} .
\end{equation}

Equivalently, after identifying the two endpoint circles by the angular coordinate,
the endpoint chiral tangential operators satisfy the sign-scale relation 
\begin{equation}\label{eq:eta_endpoint_chiral_sign_scale_pre}
B_{A,T}^{E_{\mathbb C}^{(n)},+}
=
-c\,\mathcal V\,
B_{A,0}^{E_{\mathbb C}^{(n)},+}
\mathcal V^{-1},
\qquad
c=\frac{f(0)}{f(T)}>0,
\end{equation}
whenever the same parameter $A$ is used at both endpoints. Where,
\begin{equation}
\mathcal V:
L^2(Y_0;S^+\otimes E_{\mathbb C}^{(n)}|_{Y_0})
\longrightarrow
L^2(Y_T;S^+\otimes E_{\mathbb C}^{(n)}|_{Y_T})
\end{equation}
be the unitary identification induced by the coordinate identification
$(0,\theta)\mapsto(T,\theta)$, together with the fixed product trivializations
of the spinor and twisting factors. Under this identification, the Fourier
mode $e^{ik\theta}$ and the fiber eigenvectors of $J_n$ are identified at the
two endpoints.

Thus the boundary operators are invertible on the rotating part when no block parameter
$k\pm\mu_jA$ vanishes. In the periodic sector, a neutral block has a $k=0$ boundary
kernel whenever $z>0$. This is why endpoint kernel conventions must be specified in
odd rank.

\begin{definition}[Endpoint nonresonance]
\label{def:eta_endpoint_nonresonance}
We say that the fixed parameter $A$ is rotating-nonresonant if
\begin{equation}\label{eq:eta_rotating_nonresonance_condition}
k+\mu_jA\neq0
\quad\text{and}\quad
k-\mu_jA\neq0
\qquad
\text{for all }k\in\mathcal K\text{ and all }j=1,\ldots,r .
\end{equation}
Equivalently, no rotating endpoint block has zero chiral tangential eigenvalue.
\end{definition}

\begin{remark}[Neutral boundary kernels]
\label{rem:eta_neutral_boundary_kernels}
If $z>0$ and the spin structure are periodic, then the neutral $k=0$ boundary
parameter is zero for every value of $A$. Thus, the endpoint chiral tangential operator
has a stationary neutral kernel even when $A$ is rotating nonresonant. The APS index
formula still makes sense after fixing a self-adjoint kernel convention, but the endpoint
correction must include the corresponding finite-dimensional kernel or spectral-section
correction.
\end{remark}

\subsection{Vanishing of the local index density}
\label{subsec:hrank_local_index_density_vanishes}
In our cylinder model there is no
interior two-form density left to integrate in the APS index formula. The index is
therefore defined by the boundary terms.

First suppose that the parameter $A$ is fixed. The higher-rank orthogonal connection is
\begin{equation}\label{eq:eta_fixed_connection_density} \nabla_A^{E_{\mathbb C}^{(n)}}=d+AJ_n\,d\theta .
\end{equation}
Its curvature is
\begin{equation}\label{eq:eta_curvature_vanishes}
F_{\nabla_A^{E_{\mathbb C}^{(n)}}}
=
d(AJ_n\,d\theta)
+
(AJ_n\,d\theta)\wedge(AJ_n\,d\theta)
=
0 .
\end{equation}
Indeed, $A$ is constant, so $d(AJ_n\,d\theta)=0$, and the second term vanishes
because $d\theta\wedge d\theta=0$.

Thus the twisting field contributes no interior curvature density. The Riemannian
part also contributes no two-form density on a two-dimensional spin cylinder: its first
possible nonzero curvature contribution occurs in degree four. Since the cylinder is
two-dimensional, there is no degree-four form to integrate. Hence the whole interior
local density in the APS index formula vanishes:
\begin{equation}\label{eq:eta_local_density_vanishes} \text{local interior APS index density}=0 .
\end{equation}

In the following proposition, both endpoint boundary operators are understood with the
inward-normal sign convention fixed in \autoref{sec:setup}. Thus the sign difference
between $Y_0$ and $Y_T$ is already built into the operators
$B_{A,0}^{E_{\mathbb C}^{(n)},+}$ and $B_{A,T}^{E_{\mathbb C}^{(n)},+}$. Translating to an
outward-normal APS convention changes the notation for the endpoint tangential
operators, but not the cancellation statement below.

\begin{proposition}[APS index reduces to endpoint $\eta$ terms]
\label{prop:eta_APS_index_endpoint_terms}
Consider the fixed-parameter higher-rank orthogonal model. Let
$D_{A,\text{APS}}^{E_{\mathbb C}^{(n)},+}$ be the chiral APS operator obtained from
$D_A^{E_{\mathbb C}^{(n)}}$, with endpoint chiral tangential operators
$B_{A,0}^{E_{\mathbb C}^{(n)},+}$ and $B_{A,T}^{E_{\mathbb C}^{(n)},+}$, and with
finite-dimensional endpoint kernel conventions $\mathcal P_0$ and $\mathcal P_T$
whenever endpoint kernels occur. Then the local interior index density vanishes, and
the APS index has the endpoint form
\begin{equation}\label{eq:eta_APS_index_endpoint_formula}
\operatorname{ind}
\left(
D_{A,\text{APS}}^{E_{\mathbb C}^{(n)},+}
\right)
=
-\bar\eta_{\mathcal P_0}
\left(
B_{A,0}^{E_{\mathbb C}^{(n)},+}
\right)
-
\bar\eta_{\mathcal P_T}
\left(
B_{A,T}^{E_{\mathbb C}^{(n)},+}
\right).
\end{equation}
Here $\bar\eta_{\mathcal P_0}$ and $\bar\eta_{\mathcal P_T}$ denote the reduced $\eta$ corrections determined by the chosen endpoint APS kernel conventions. If both endpoint
chiral tangential operators are invertible, these are the usual reduced $\eta$ invariants.
\end{proposition}

\begin{proof}
The APS index theorem expresses the chiral APS index as the sum of the local interior
index integral and the endpoint $\eta$ corrections determined by the APS boundary
operators and their kernel conventions. By \eqref{eq:eta_local_density_vanishes}, the
local interior density vanishes in the fixed-parameter orthogonal model.
Therefore only the endpoint $\eta$ corrections remain. The two minus signs in
\eqref{eq:eta_APS_index_endpoint_formula} are part of the endpoint convention used here:
the endpoint tangential operators
$B_{A,0}^{E_{\mathbb C}^{(n)},+}$ and $B_{A,T}^{E_{\mathbb C}^{(n)},+}$ already
include the sign difference coming from the two inward-normal choices. If an endpoint
operator has kernel, the notation $\bar\eta_{\mathcal P_t}$ encodes the corresponding
finite-dimensional kernel correction.
\end{proof}

\subsection{Endpoint sign and scale relation}
\label{subsec:hrank_endpoint_sign_scale_relation}

The two endpoint boundary operators are related by a sign and a positive scale. This is the same mechanism already present in the scalar and rank-two models.

First take the same fixed parameter $A$ at both endpoints. After
identifying $Y_0$ and $Y_T$ by the angular coordinate $\theta$, the endpoint
operators satisfy
\begin{equation}\label{eq:eta_endpoint_sign_scale_relation}
B_{A,T}^{E_{\mathbb C}^{(n)},+}
=
-c\,\mathcal V\,
B_{A,0}^{E_{\mathbb C}^{(n)},+}
\mathcal V^{-1},
\qquad
c=\frac{f(0)}{f(T)}>0,
\end{equation}
where $\mathcal V$ is the unitary identification of the endpoint spinor and twisting
fibers induced by the chosen trivialization.

Positive rescaling does not change the $\eta$ invariant, while multiplication of a
self-adjoint operator by $-1$ reverses the sign of its $\eta$ invariant. Hence, when
the endpoint operators are invertible,
\begin{equation}\label{eq:eta_endpoint_eta_cancellation}
\bar\eta
\left(
B_{A,T}^{E_{\mathbb C}^{(n)},+}
\right)
=
-
\bar\eta
\left(
B_{A,0}^{E_{\mathbb C}^{(n)},+}
\right).
\end{equation}

\begin{corollary}[Vanishing of the fixed-parameter chiral APS index]
\label{cor:eta_fixed_parameter_index_vanishes}
Take the same parameter $A$ at both endpoints. Require also that the endpoint
boundary operators are invertible, or that the endpoint kernel conventions are chosen
compatibly with the sign-scale identification
\eqref{eq:eta_endpoint_sign_scale_relation}. Then
\begin{equation}\label{eq:eta_fixed_parameter_index_zero} \operatorname{ind} \left( D_{A,\text{APS}}^{E_{\mathbb C}^{(n)},+} \right) = 0.
\end{equation}
\end{corollary}

\begin{proof}
By Proposition~\ref{prop:eta_APS_index_endpoint_terms}, the index is the negative of
the sum of the two endpoint reduced $\eta$ corrections. Under the sign-scale relation
\eqref{eq:eta_endpoint_sign_scale_relation}, the two reduced $\eta$ corrections cancel:
\begin{equation}\label{eq:eta_endpoint_sum_zero}
\bar\eta_{\mathcal P_0}
\left(
B_{A,0}^{E_{\mathbb C}^{(n)},+}
\right)
+
\bar\eta_{\mathcal P_T}
\left(
B_{A,T}^{E_{\mathbb C}^{(n)},+}
\right)
=
0.
\end{equation}
Therefore the index vanishes.
\end{proof}

\subsection{Distinct endpoint parameters}
\label{subsec:hrank_distinct_endpoint_parameters}

One may also allow the angular connection parameter to take different values at the two
boundary components. To make this into a genuine bulk boundary-value problem, let $a$ be a
smooth function
\begin{equation}\label{eq:eta_tdependent_parameter} a\in C^\infty([0,T];\mathbb R), \qquad a(0)=A_0, \qquad a(T)=A_T,
\end{equation}
and use the orthogonal connection
\begin{equation}\label{eq:eta_tdependent_connection} \nabla^{E_{\mathbb R}^{(n)}}_a = d+a(t)J_n\,d\theta .
\end{equation}
Its complexification is
\begin{equation}\label{eq:eta_tdependent_complex_connection} \nabla^{E_{\mathbb C}^{(n)}}_a = d+a(t)J_n\,d\theta .
\end{equation}
The curvature is
\begin{equation}\label{eq:eta_tdependent_curvature} F_{\nabla_a} = a'(t)J_n\,dt\wedge d\theta .
\end{equation}
This curvature need not vanish. However, the local twisting density that could
contribute to the two-dimensional index density is obtained by taking the fiber trace
of this curvature term. Since $J_n\in\mathfrak{so}(n)$ is skew-symmetric, it has trace
zero:
\begin{equation}\label{eq:eta_trace_Jn_zero} \operatorname{Tr}(J_n)=0 .
\end{equation}
Therefore
\begin{equation}\label{eq:eta_tdependent_twisting_density_zero}
\operatorname{Tr}
\left(
F_{\nabla_a}
\right)
=
a'(t)\operatorname{Tr}(J_n)\,dt\wedge d\theta
=
0 .
\end{equation}
So even for a $t$-dependent parameter $a(t)$, the twisting field contributes no
interior two-form density. As before, the purely Riemannian part has no degree-two
contribution on the cylinder. Hence the local interior APS index density again
vanishes.

Let $B_{A_0,0}^{E_{\mathbb C}^{(n)},+}$ be the endpoint boundary operator at $Y_0$
computed from $a(0)=A_0$, and let $B_{A_T,T}^{E_{\mathbb C}^{(n)},+}$ be the
endpoint boundary operator at $Y_T$ computed from $a(T)=A_T$. The chiral APS index
therefore has the endpoint form
\begin{equation}\label{eq:eta_distinct_endpoint_index_formula}
\operatorname{ind}
\left(
D_{\text{APS},a}^{E_{\mathbb C}^{(n)},+}
\right)
=
-\bar\eta_{\mathcal P_0}
\left(
B_{A_0,0}^{E_{\mathbb C}^{(n)},+}
\right)
-
\bar\eta_{\mathcal P_T}
\left(
B_{A_T,T}^{E_{\mathbb C}^{(n)},+}
\right).
\end{equation}
Here $\mathcal P_0$ and $\mathcal P_T$ encode the endpoint APS kernel conventions.

If there is a unitary identification and a positive scale $c>0$ such that
\begin{equation}\label{eq:eta_distinct_endpoint_sign_scale}
B_{A_T,T}^{E_{\mathbb C}^{(n)},+}
=
-c\,\mathcal V\,
B_{A_0,0}^{E_{\mathbb C}^{(n)},+}
\mathcal V^{-1},
\end{equation}
then the nonzero $\eta$ contributions cancel, subject again to compatible treatment of
endpoint kernels. In that case the index vanishes. Without such an endpoint comparison,
the index is generally an endpoint $\eta$ expression rather than zero.

\begin{remark}[Endpoint formula versus moving spectral flow]
\label{rem:eta_endpoint_vs_moving_spectral_flow}
The formula \eqref{eq:eta_distinct_endpoint_index_formula} is an endpoint APS index
formula for one chiral Fredholm problem determined by the bulk connection
$a(t)$. It should not be confused with the moving-family spectral-flow formula of
\autoref{sec:hrank_ROO2_spectral_flow_formula}. The moving-family formula counts
regularized crossings along $p\mapsto A_p$, whereas the endpoint APS index formula
computes a single chiral Fredholm index from the $\eta$ data of the two boundary
components.
\end{remark}

\subsection{Reflection and endpoint $\eta$ terms}
\label{subsec:hrank_reflection_equivariant_eta}

The reflection symmetry used in the self-adjoint spectral-flow part acts naturally on
the full two-component boundary data. In the chiral APS index problem, however, one
must be more careful: the chosen spinorial lift $U_{\mathbf r}=\sigma_1$ exchanges the two
spinor components, and hence does not automatically define a symmetry of the chiral
endpoint tangential operator by itself.

For this reason we do not use a reflection-equivariant $\eta$ formula in the proof of the
ordinary chiral APS index statement above. The ordinary endpoint formula
\eqref{eq:eta_APS_index_endpoint_formula} only requires the chiral tangential operators
and their reduced $\eta$ corrections.

If one uses an additional boundary linearization for which a reflection operator
$\tau_t$ preserves the relevant endpoint chiral boundary problem and commutes with
the chiral tangential operator $B$, then one may define the corresponding equivariant
$\eta$ function by
\begin{equation}\label{eq:eta_equivariant_eta_function}
\eta_{\tau_t}(B;s)
=
\sum_{\lambda\in\operatorname{Spec}(B)\setminus\{0\}}
\operatorname{sign}(\lambda)|\lambda|^{-s}
\operatorname{Tr}_{E_\lambda}(\tau_t),
\qquad
\operatorname{Re}(s)\gg0,
\end{equation}
where $E_\lambda=\ker(B-\lambda)$. Its value at $s=0$, after meromorphic
continuation, is denoted by $\eta_{\tau_t}(B)$, and the corresponding reduced
correction is
\begin{equation}\label{eq:eta_equivariant_reduced_eta}
\bar\eta_{\tau_t}(B)
=
\frac{\eta_{\tau_t}(B)+h_{\tau_t}(B)}{2},
\qquad
h_{\tau_t}(B)=\operatorname{Tr}_{\ker B}(\tau_t).
\end{equation}

Under such an additional commuting-boundary-symmetry hypothesis, the nonzero
reflection-paired Fourier modes give zero trace contributions to the equivariant $\eta$ function. Neutral zero modes, when present, must still be treated separately because
their contribution depends on the chosen neutral reflection convention and on the endpoint
kernel convention.

\appendix

\section{Maslov framework for reflected nonzero crossings}
\label{app:maslov_reflected_frame_check}
We show an explicit finite-dimensional check of the sign of a reflected crossing. The purpose is to distinguish three related objects: the raw scalar boundary parameter $m$, the reflected scalar boundary parameter $-m$, and the lifted reflected Maslov problem. The scalar boundary slopes of $m$ and $-m$ are opposite. However, the two boundary problems are related by the boundary restriction of the lifted reflection, and the signed Maslov crossing is invariant after passing to the reflected boundary frame.

We work in one scalar boundary block, with ordered basis
$e_+=\binom{1}{0}$ and $e_-=\binom{0}{1}$.
Let
$c(m)=\cos\alpha(m/\delta)$ and $s(m)=\sin\alpha(m/\delta)$.
Let $v_0(m), v_T(m)$ represent the regularized APS domain at $Y_0$ and $Y_T$ respectively, then
$v_0(m)=\binom{c(m)}{i\,s(m)}$ and $v_T(m)=\binom{-i\,s(m)}{c(m)}$.
Thus
\begin{equation}
\ell_0^{\text{dom}}(m)=\mathbb C v_0(m)
=
\mathbb C\binom{c(m)}{i\,s(m)},
\qquad
\ell_T^{\text{dom}}(m)=\mathbb C v_T(m)
=
\mathbb C\binom{-i\,s(m)}{c(m)}.
\end{equation}
 With $L_f$ defined in \eqref{eq:hrank_Lf_def}, the zero-energy transfer matrix is 
\begin{equation}
\Phi_m=
\begin{pmatrix}
e^{mL_f} & 0\\
0 & e^{-mL_f}
\end{pmatrix},
\end{equation}
so
$\Phi_m v_0(m) = \binom{e^{mL_f}c(m)}{i\,e^{-mL_f}s(m)}$.

The transfer Maslov matrix is
\begin{equation}
M(m)
=
\begin{pmatrix}
-i\,s(m) & e^{mL_f}c(m)\\
c(m) & i\,e^{-mL_f}s(m)
\end{pmatrix}.
\end{equation}
Its determinant is
$F(m):=\det M(m) = s(m)^2e^{-mL_f} - c(m)^2e^{mL_f}$.
The equation $F(m)=0$ is the condition that the transported line $\Phi_m\ell_0^{\text{dom}}(m)$ meets the terminal line $\ell_T^{\text{dom}}(m)$.

Now we take the affine example
$A_p=-1.2+p, k=1, \mu=1$.
Then $m(p)=1+A_p=-0.2+p$ for $p \in (0, 1)$, then the crossing occurs at
$p_*=0.2, m(p_*)=0$.
The representative determinant is $F_{\text{rep}}(p)=F(m(p))$, and
$F_{\text{rep}}'(p_*)=F'(0)m'(p_*)$.

The reflected boundary parameter is $-m(p)$. Its scalar slope is opposite:
\begin{equation}
\frac{d}{dp}(-m(p))\bigg|_{p=p_*} = -m'(p_*).
\end{equation}
Thus, at the level of raw scalar boundary labels, the two crossings have opposite boundary slopes.

If one only reverses the scalar coordinate and keeps the old scalar frame, then, since $\alpha(-x)=\pi/2-\alpha(x)$, one has
$c(-m)=s(m)$ and $s(-m)=c(m)$.
Therefore
\begin{equation}
M(-m)
=
\begin{pmatrix}
-i\,c(m) & e^{-mL_f}s(m)\\
s(m) & i\,e^{mL_f}c(m)
\end{pmatrix},
\end{equation}
and hence
$\det M(-m) = c(m)^2e^{mL_f} - s(m)^2e^{-mL_f} = -F(m)$.
Consequently,
$\frac{d}{dp}F(-m(p))\bigg|_{p=p_*} = -F'(0)m'(p_*)$.
This is the raw boundary-coordinate sign reversal. It is real, but it is not yet the equivariant Maslov comparison.

We now include the boundary reflection. Let $\mathcal R_t$ denote the boundary restriction of the lifted reflection $\mathcal U_{\mathbf r}$. In this two-dimensional boundary coordinate model,
\begin{equation}
\mathcal R_t
=
\mathcal U_{\mathbf r}
=
\sigma_1
=
\begin{pmatrix}
0&1\\
1&0
\end{pmatrix}.
\end{equation}
It acts on the basis vectors by
$\mathcal R_t e_+=e_-$ and $\mathcal R_t e_-=e_+$.
It also intertwines the two boundary blocks:
$\mathcal R_t B_m\mathcal R_t^{-1}=B_{-m}$.
Thus the reflected boundary-frame view is
\begin{equation}
\left(
B_m,\ell^{\text{dom}}(m)
\right)
\ \xrightarrow{\ \mathcal R_t\ }\
\left(
B_{-m},\mathcal R_t\ell^{\text{dom}}(m)
\right).
\end{equation}
Here $\mathcal R_t\ell^{\text{dom}}(m)$ means the image of the domain line under $\mathcal R_t$. Explicitly,
\begin{equation}
\mathcal R_t\ell_0^{\text{dom}}(m)
=
\mathcal R_t\left(\mathbb C\binom{c(m)}{i\,s(m)}\right)
=
\mathbb C\,\mathcal R_t\binom{c(m)}{i\,s(m)}.
\end{equation}
Since
\begin{equation}
\mathcal R_t\binom{c(m)}{i\,s(m)}
=
\begin{pmatrix}
0&1\\
1&0
\end{pmatrix}
\binom{c(m)}{i\,s(m)}
=
\binom{i\,s(m)}{c(m)},
\end{equation}
we get
$\mathcal R_t\ell_0^{\text{dom}}(m) = \mathbb C\binom{i\,s(m)}{c(m)}$.
Similarly,
\begin{equation}
\mathcal R_t\ell_T^{\text{dom}}(m)
=
\mathcal R_t\left(\mathbb C\binom{-i\,s(m)}{c(m)}\right)
=
\mathbb C\,\mathcal R_t\binom{-i\,s(m)}{c(m)}.
\end{equation}
Since
\begin{equation}
\mathcal R_t\binom{-i\,s(m)}{c(m)}
=
\begin{pmatrix}
0&1\\
1&0
\end{pmatrix}
\binom{-i\,s(m)}{c(m)}
=
\binom{c(m)}{-i\,s(m)},
\end{equation}
we get
$\mathcal R_t\ell_T^{\text{dom}}(m) = \mathbb C\binom{c(m)}{-i\,s(m)}$.
Therefore the boundary crossing with parameter $-m$ is not being counted as an independent crossing in the original scalar frame. It is the same boundary crossing transported by the boundary reflection $\mathcal R_t$.

Now apply $\mathcal R_t$ to the two Maslov columns. First,
$\mathcal R_t v_T(m) = \binom{c(m)}{-i\,s(m)}$.
Second,
\begin{equation}
\mathcal R_t\Phi_m v_0(m)
=
\mathcal R_t
\binom{e^{mL_f}c(m)}{i\,e^{-mL_f}s(m)}
=
\binom{i\,e^{-mL_f}s(m)}{e^{mL_f}c(m)}.
\end{equation}
If these reflected vectors are still written in the old ordered basis $\{e_+,e_-\}$, the raw reflected matrix is
\begin{equation}
M_{\text{ref}}^{\text{old}}(m)
=
\mathcal R_t M(m)
=
\begin{pmatrix}
c(m) & i\,e^{-mL_f}s(m)\\
-i\,s(m) & e^{mL_f}c(m)
\end{pmatrix}.
\end{equation}
Therefore
$\det M_{\text{ref}}^{\text{old}}(m) = \det(\mathcal R_t)\det M(m) = -\det M(m)$,
because $\det(\mathcal R_t)=-1$. This minus sign comes only from measuring reflected vectors in the old ordered basis.

The natural boundary frame for the reflected block is the reflected ordered basis
\begin{equation}
e_+^{\text{ref}}:=\mathcal R_t e_+=e_-,
\qquad
e_-^{\text{ref}}:=\mathcal R_t e_-=e_+.
\end{equation}
Let $B_{\text{ref}}$ be the change-of-basis matrix whose columns are $e_+^{\text{ref}}$ and $e_-^{\text{ref}}$, written in the old basis. Since $e_+^{\text{ref}}=e_-=(0,1)^T$ and $e_-^{\text{ref}}=e_+=(1,0)^T$, one has
\begin{equation}
B_{\text{ref}}
=
\begin{pmatrix}
0&1\\
1&0
\end{pmatrix}
=
\mathcal R_t.
\end{equation}
This matrix sends reflected-basis coordinates to old-basis coordinates. Hence old-basis coordinates must be multiplied by $B_{\text{ref}}^{-1}$ to obtain reflected-basis coordinates. Therefore the reflected Maslov matrix in the natural reflected boundary frame is
$M_{\text{ref}}^{\text{nat}}(m) = B_{\text{ref}}^{-1}M_{\text{ref}}^{\text{old}}(m)$.
Using $B_{\text{ref}}=\mathcal R_t$, we get
$M_{\text{ref}}^{\text{nat}}(m) = \mathcal R_t^{-1}\mathcal R_t M(m) = M(m)$.
Thus
$\det M_{\text{ref}}^{\text{nat}}(m) = \det M(m)$.

This is the boundary Maslov explanation of the sign. The scalar boundary parameters $m(p)$ and $-m(p)$ have opposite raw slopes, but the two boundary problems are identified by the boundary reflection $\mathcal R_t$. Once the reflected Maslov data are written in the reflected boundary frame, the signed Maslov crossing is the same as the representative one.

Equivalently, the two sign changes cancel:
$m'(p_*)\longmapsto -m'(p_*)$ and $F'(0)\longmapsto -F'(0)$,
so that
$(-F'(0))(-m'(p_*)) = F'(0)m'(p_*)$.
Therefore the boundary Maslov sign agrees with the bulk crossing sign after the equivariant reflected identification.

At the operator level, the same conclusion follows from reflection equivariance:
$\mathcal U_{\mathbf r}D_p\mathcal U_{\mathbf r}^{-1}=D_p$.
If $D_pu(p)=\lambda(p)u(p)$, then
$D_p\mathcal U_{\mathbf r}u(p) = \lambda(p)\mathcal U_{\mathbf r}u(p)$.
The reflected eigenvector has the same eigenvalue branch $\lambda(p)$, not the opposite branch. Hence the bulk eigenvalue crosses zero in the same direction.

Therefore a nonzero reflected pair contributes one signed real $O(2)$-representation,
$\epsilon_{k,j,p_*}[\rho_k]$,
rather than two opposite one-dimensional signs.

\begin{figure}[H]
\centering
\includegraphics[width=0.48\textwidth]{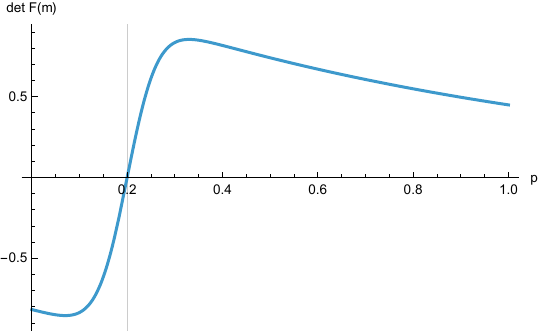}
\hfill
\includegraphics[width=0.48\textwidth]{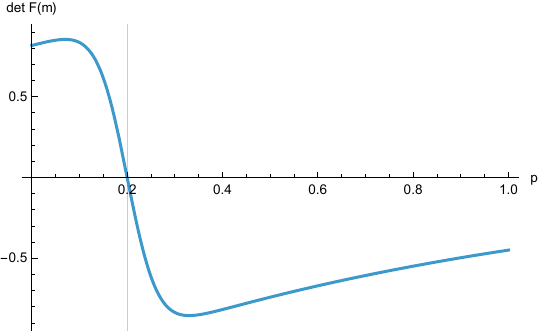}
\caption{Maslov determinants for the self-paired zero Fourier block in the affine example
$A_p=-0.2+p$, with $k=0$, $\mu=1$, $p_*=0.2$, $L_f=1$, and
$\delta=0.3$. We use
$\alpha(x)=\frac{\pi}{4}(1+\tanh(3x))$. Left: the determinant for the
reflection-even sector, with effective parameter
$m_{0,+}(p)=A_p=-0.2+p$, contributing $[\mathbf 1]$. Right: the determinant
for the reflection-odd sector, with effective parameter
$m_{0,-}(p)=-A_p=0.2-p$, contributing $-[\det]$. The two determinants vanish
at the same value $p_*=0.2$ but have opposite crossing directions. Hence the
zero-mode contribution is $[\mathbf 1]-[\det]$, which is nonzero in
$RO(O(2))$ but maps to zero under the ordinary dimension map.}
\label{fig:maslov_zero_block}
\end{figure}

\subsection*{The self-paired zero mode}

The previous calculation applies to a nonzero reflected pair. There the two scalar boundary labels $m$ and $-m$ belong to the two distinct Fourier modes with angular weights $k$ and $-k$, for some $k>0$, and these modes are exchanged by reflection. This is why the reflected boundary problem is the same
as the crossing transported by $\mathcal R_t$, and why the two modes combine into one signed copy of $\rho_k$. 

The zero Fourier block is different. At $k=0$, reflection does not exchange two distinct Fourier modes, because $0=-0$. Instead, the zero block is already self-paired. The real crossing space splits into the reflection-even and reflection-odd lines,
$V_0 = V_0^+\oplus V_0^-$.
Reflection acts on these two lines by
$\mathcal R_t|_{V_0^+}=+1, \mathcal R_t|_{V_0^-}=-1$.
Thus they are genuine one-dimensional $O(2)$-representations:
$V_0^+\cong \mathbf 1, V_0^-\cong \det$.
Hence the zero block is not one irreducible two-dimensional representation $\rho_k$. It is
$V_0\cong \mathbf 1\oplus \det$.

In this case, the two scalar signs are not forced to agree by a reflected-pair identification. Instead, the zero-block crossing form restricts separately to the two reflection sectors:
$\Gamma_{0,p_*} = \Gamma_{\mathbf{1},p_*}\oplus \Gamma_{\det,p_*}$.
Equivalently, if $e\in V_0^+$ and $o\in V_0^-$ are unit vectors, then
\begin{equation}
\Gamma_{\mathbf{1},p_*}
=
\Gamma_{0,p_*}(e,e),
\qquad
\Gamma_{\det,p_*}
=
\Gamma_{0,p_*}(o,o).
\end{equation}
The local zero-mode contribution therefore has the form
$\epsilon_{\mathbf{1},j,p_*}[\mathbf 1] + \epsilon_{\det,j,p_*}[\det]$,
where
\begin{equation}
\epsilon_{\mathbf{1},j,p_*}
=
\operatorname{sign}\Gamma_{\mathbf{1},p_*},
\qquad
\epsilon_{\det,j,p_*}
=
\operatorname{sign}\Gamma_{\det,p_*}.
\end{equation}

It is important that the opposite signs do not come merely from applying reflection. Indeed,
$\mathcal R_t e=e$ and $\mathcal R_t o=-o$.
For a reflection-invariant crossing form,
$\Gamma_{0,p_*}(\mathcal R_t u,\mathcal R_t v) = \Gamma_{0,p_*}(u,v)$.
Thus
$\Gamma_{0,p_*}(\mathcal R_t e,\mathcal R_t e) = \Gamma_{0,p_*}(e,e)$,
and
\begin{equation}
\Gamma_{0,p_*}(\mathcal R_t o,\mathcal R_t o)
=
\Gamma_{0,p_*}(-o,-o)
=
\Gamma_{0,p_*}(o,o).
\end{equation}
So reflection invariance alone does not force the two signs to be opposite. It only says that the crossing form is compatible with the reflection splitting.

For the sign-normalized zero-block convention used in the affine examples, the crossing form is chosen diagonal with opposite coefficients on the two reflection sectors:
\begin{equation}
\Gamma_{0,p_*}
=
c\,A_p'(p_*)\,(\,\cdot\,,\,\cdot\,)_{V_0^+}
-
c\,A_p'(p_*)\,(\,\cdot\,,\,\cdot\,)_{V_0^-},
\qquad
c>0.
\end{equation}
Equivalently,
\begin{equation}
\Gamma_{0,p_*}(e,e)
=
c\,A_p'(p_*),
\qquad
\Gamma_{0,p_*}(o,o)
=
-c\,A_p'(p_*).
\end{equation}
Therefore
\begin{equation}
\epsilon_{\mathbf{1},j,p_*}
=
\operatorname{sign}A_p'(p_*),
\qquad
\epsilon_{\det,j,p_*}
=
-\operatorname{sign}A_p'(p_*).
\end{equation}
For an affine path $A_p=A+\nu p$, this becomes
\begin{equation}
\epsilon_{\mathbf{1},j,p_*}
=
\operatorname{sign}\nu,
\qquad
\epsilon_{\det,j,p_*}
=
-\operatorname{sign}\nu.
\end{equation}
Thus, the zero-mode contribution is
$\operatorname{sign}(\nu)\bigl([\mathbf 1]-[\det]\bigr)$.

We see for $k\neq0$, the two scalar halves are exchanged by reflection and form one real irreducible block $\rho_k$. For $k=0$, there is no distinct reflected Fourier partner. The block splits into two invariant reflection sectors, $\mathbf 1$ and $\det$, and opposite signs can remain visible in $RO(O(2))$ once the zero-block regularization has been fixed with the above sign-normalized convention.

\bibliographystyle{utphys}
\bibliography{ref}

\end{document}